\documentclass[review]{elsarticle}

\usepackage[T1]{fontenc}
\usepackage[utf8]{inputenc}
\usepackage{lmodern}
\usepackage{microtype}
\usepackage{graphicx}
\usepackage{geometry}
\usepackage{array}
\usepackage{fancyhdr}
\usepackage{tabularx}
\usepackage{url}
\usepackage{xcolor}
\usepackage{pict2e}
\usepackage{color}
\usepackage{lineno,hyperref}
\hypersetup{
    colorlinks,
    linkcolor={blue!99!black},
    citecolor={blue!80!black},
    urlcolor={blue!20!black}
}
\usepackage{subcaption}
\usepackage{amsmath,amssymb,amsfonts,esint,amsbsy}
\usepackage{amsthm}
\usepackage{tikz}
\usepackage{ulem}
\usepackage[lined,boxed]{algorithm2e}

\newtheorem{theorem}{Theorem}[section]
\newtheorem{corollary}[theorem]{Corollary}
\newtheorem{lemma}[theorem]{Lemma}
\newtheorem{definition}[theorem]{Definition}

\renewenvironment{proof}{\noindent\rule{\textwidth}{1pt} {\bfseries{Proof.}}}{\qed \\ \noindent\rule{\textwidth}{1pt}}

\modulolinenumbers[5]
\journal{Journal}

\biboptions{numbers,sort&compress}

\begin{document}

\begin{frontmatter}
	\title{On the Geometric Conservation Law for the Non Linear Frequency Domain and Time-Spectral Methods}
	\author{Marc Benoit}
	\author{Siva Nadarajah}
	\address{Department of Mechanical Engineering, McGill University, Montreal, Quebec, H3A 0C3, Canada}
	
	\begin{abstract}
	The aim of this paper is to present and validate two new procedures to enforce the Geometric Conservation Law (GCL) on a moving grid for an Arbitrary Lagrangian Eulerian (ALE) formulation of the Euler equations discretized in time for either the Non Linear Frequency Domain (NLFD) or Time-Spectral (TS) methods. The equations are spatially discretized by a structured finite-volume scheme on a hexahedral mesh. The derived methodologies follow a general approach where the positions and the velocities of the grid points are known at each time step. The integrated face mesh velocities are derived either from the Approximation of the Exact Volumetric Increments (AEVI) relative to the undeformed mesh or exactly computed based on a Trilinear Mapping (TRI-MAP) between the physical space and the computational domain. The accuracy of the AEVI method highly depends on the computation of the volumetric increments and limits the temporal-order of accuracy of the deduced integrated face mesh velocities to between one and two. Thus defeating the purpose of the NLFD method which possesses spectral rate of convergence. However, the TRI-MAP method has proven to be more computationally efficient, ensuring the satisfaction of the GCL once the convergence of the time derivative of the cell volume is reached in Fourier space. The methods are validated numerically by verifying the conservation of uniform flow and by comparing the integrated face mesh velocities to the exact values derived from the mapping.	
	
	\textbf{Keywords :} Non Linear Frequency Domain (NLFD) method;
Time-Spectral method;
Geometric Conservation Law (GCL);
Deforming mesh;
Computational fluid dynamics;
Numerical method	
	\end{abstract}
\end{frontmatter}

\linenumbers

\newpage
\section{Introduction}
\label{sec_Intro}
The accurate computation of unsteady aerodynamic flows in a reasonable amount of time still presents a challenge in the field of computational fluid dynamics. Compared to steady flow problems which only require an accurate spatial discretization, unsteady flow solvers have to provide an accurate time resolution of the flow. Until lately, the most popular approaches were time marching techniques for which the solution is constructed in time from an initial free-stream solution. Despite acceleration techniques such as multigrid or local time stepping, these methods remain computationally costly, partly because of the transient effects, requiring numerous time steps to converge. However in the case of periodic flows encountered in problems such as aeroleastic simulations or turbomachinery, the Fourier collocation method can be used to accurately and efficiently represent the solution. A nonlinear harmonic method known as the Harmonic Balance (HB) was initially introduced by Hall et al. \cite{Hall2002}. Then, McMullen et al. \cite{McMullen2002a,McMullen2006a,McMullen2006} developed the Non Linear Frequency Domain (NLFD) method in order to solve the Euler and Navier-Stokes equations directly in frequency domain. An alternative to this approach is the Time-Spectral (TS) method presented by Gopinath et al. \cite{Gopinath2005,Gopinath2006}, it avoids the explicit use of a Discrete Fourier Transform (DFT) and discretizes the temporal derivative operator through a Fourier collocation matrix. These various methods prove to significantly decrease the required time to obtain the solution compared to time marching solvers.

The NLFD method was used to develop a two dimensional aeroelastic flow solver \cite{Kachra2006} and then expanded to three-dimensional cases on moving grids \cite{Tardif2017}. In order to study aeroelastic problems, the mesh is moved in time and additional care has to be taken to compute the mesh velocities and the metrics from the physical to the computational space. Thomas et al. \cite{Thomas1979} were the first to formally define the necessity to solve additional laws to preserve the conservation of the solver numerical scheme. Termed as the Geometric Conservation Law, it is composed of two subsets of laws known as the Surface Conservation Law (SCL) and the Volume Conservation Law (VCL). A mathematical interpretation of the SCL relates that any cell volume has to be closed by its surfaces whereas the VCL states that the temporal rate of change of the cell volume is equal to the sum of the temporal rate of change of the algebraic volumes swept by each face enclosing it through time. SCL differs from the VCL in the way that they need to be verified even for fixed grids (steady state) while the VCL appears only on deforming grids. The violation of any of these laws may result in error in the flow solution, for instance it was reported that the violation of the GCL leads to inaccurate flutter prediction for aeroelastic cases \cite{Lesoinne1996}. Further investigation on time marching schemes clarified the theoretical status of the GCL, exposing its link to temporal-order accuracy \cite{Herve2000}, or stability conditions \cite{Farhat2001}. In addition, the Discrete Geometric Conservation Law (DGCL) were derived in order to preserve the time accuracy for high temporal-order schemes \cite{Mavripilis2006}. These methods are well adapted for time marching approaches, but their extension to the NLFD or Time-Spectral methods is not straightforward since it becomes necessary to compute all quantities : state vector, fluxes, mesh positions and mesh velocities, at all time steps before applying the Fourier discretization. 

In this work, we focus on the VCL part of the GCL which arises from the unsteadiness of the flow. The framework of our study is based on the previous development of a three dimensional aeroelastic solver by Tardiff et al. \cite{Tardif2017} who introduced a methodology to enforce the GCL. The expression of the governing equations according to an Arbitrary Lagrangian-Eulerian formulation, the principles of the NLFD discretization and the dynamic mesh deformation using Radial Basis Functions are presented in section \ref{sec_ALE_NLFD_RBF}. Additional analytical developments are introduced to demonstrate the limitations of the approach of Tardiff et al. \cite{Tardif2017}. A modified method is derived to satisfy the GCL by computing the integrated face mesh velocities at a second-order accuracy in time. Since such a condition limits the inherent spectral in time accuracy of the NLFD approach, an alternative method based on a trilinear mapping between the physical and computational space is introduced. The results are then extent to the Time-Spectral method. These various approaches are presented section \ref{sec_GCL}. The different methodologies are numerically investigated in order to verify that the GCL are enforced through a correct computation of the integrated face mesh velocities in section \ref{sec_num} and the results discussed in section \ref{sec_con_disc}.

\section{Discretization of the governing equation and mesh deformation}
\label{sec_ALE_NLFD_RBF}
This section presents the formulation of the Euler equations on a moving mesh using the Arbitrary Lagrangian-Eulerian approach, its discretization using the Non-Linear Frequency Domain, and the mesh deformation method through the Radial Basis Functions.

\subsection{Arbitrary Lagrangian-Eulerian formulation of the Euler equations}
\label{subsec_ALE}
When solving the Euler equations on a moving grid a popular approach is to use an Arbitrary Lagrangian-Eulerian (ALE) formulation \cite{Blazek2005}. For a control volume $\Omega$ enclosed by a boundary $\partial \Omega$ and without source terms, the integral form of this formulation is given as follows :
\begin{equation}
\int_{\Omega} \mathbf{w} \text{d}\Omega + \oint_{\partial \Omega} \mathbf{F_c^M} \text{d}S = 0,
\label{eq_ALE_int}
\end{equation}
with the state vector $\mathbf{w}$ and the vector of the convective fluxes on a moving grid $\mathbf{F_c^M}$ expressed as :
\begin{equation}
\mathbf{w}=
\left\{
\begin{array}{c}
\rho \\ \rho u \\ \rho v \\ \rho w \\ \rho E
\end{array}
\right\}
\text{ ; }
\mathbf{F_c^M} = \mathbf{F_c}-V_t \mathbf{w} =
\left\{
\begin{array}{c}
\rho V \\ \rho u V + n_x p\\ \rho v V + n_y p\\ \rho w V + n_z p \\ \rho H V
\end{array}
\right\}
-
\left\{
\begin{array}{c}
\rho V_t\\ \rho u V_t\\ \rho v V_t\\ \rho w V_t \\ \rho E V_t
\end{array}
\right\},
\label{eq_WR_vectors}
\end{equation}  
where $\mathbf{F_c}$ would be the convective flux vector on a fixed grid and $\rho$, $u$, $v$, $w$, $E$, $p$ and $H$ are respectively, the density, the Cartesian velocities of the fluid, the total energy per unit mass, the pressure and the total enthalpy per unit mass defined by :
\begin{equation}
H=E+\frac{p}{\rho}.
\label{eq_H}
\end{equation}

In order to close the system of equations, the pressure is evaluated under the assumption of ideal gas through the state equation (\ref{eq_P}) :
\begin{equation}
p = (\gamma - 1) \left( \rho E- \frac{(\rho u)^2+(\rho v)^2 + (\rho w)^2}{2 \rho}\right).
\label{eq_P}
\end{equation}

Also, $V$ is the contravariant velocity of the fluid and $V_t$ is the contravariant velocity of the boundary enclosing the control volume. $n_x$, $n_y$ and $n_z$ are the components of the boundary unit normal vector pointing outward of the control volume. It yields :
\begin{equation}
\begin{array}{lcl}
V = un_x + vn_y + wn_z,\\
V_t = \dfrac{\partial x}{\partial t}n_x + \dfrac{\partial y}{\partial t}n_y + \dfrac{\partial z}{\partial t}n_z.
\end{array}
\label{eq_contravariant_velocities}
\end{equation}

By introducing a discretized control volume and an artificial dissipation flux vector $\mathbf{F_d}$ to avoid an odd-even decoupling of the solution and to increase the accuracy at discontinuities, the equation (\ref{eq_ALE_int}) can be written under a semi-discretized form as :
\begin{equation}
\frac{\partial (\Omega \mathbf{w})}{\partial t} + \sum_{\partial \Omega} \left[ (\mathbf{F_c^M})S-\mathbf{F_d} \right] = 0.
\label{eq_semi_discrete_Euler_1}
\end{equation}

The previous set of equations has to hold for each control volume and can be expressed as a semi-discrete system of ordinary differential equations in time :
\begin{equation}
\frac{\partial (\Omega \mathbf{w})}{\partial t} + \mathbf{R(\mathbf{w})} = 0,
\label{eq_semi_discrete_Euler_2}
\end{equation}
where $\mathbf{R}(\mathbf{w})= \displaystyle \sum_{\partial \Omega} \left[ (\mathbf{F_c^M})S-\mathbf{F_d} \right]$ is the discrete residual vector. In this work, the discretization in space is performed according to the finite volume method on a structured hexahedral grid, the modified convective flux is computed as the average of the fluxes at a cell face and  the artificial dissipation is evaluated using the JST scheme \cite{JAMESON1981}. The residual vector is calculated as the summation over the faces of the control volume of the different fluxes.	

\subsection{Temporal discretization using the Non-Linear Frequency Domain method}
\label{subsec_NLFD}
The temporal discretization of the  flow solver employs the NLFD approach developed by McMullen et al. \cite{McMullen2002a}. Under the assumption that both the modified state vector $\bar{\mathbf{w}}=\Omega\mathbf{w}$ and the residual vector $\mathbf{R}(\mathbf{w})$ are periodic in physical time, the two quantities can be expanded as discrete Fourier series using a finite number of harmonics equations (\ref{eq_DFT_w}) and (\ref{eq_DFT_R}) :
\begin{equation}
\bar{\mathbf{w}}(t)=\sum_{k=-N}^N \mathbf{\hat{w}}_k \text{e}^{i(2\pi k/T)t},
\label{eq_DFT_w}
\end{equation}
\begin{equation}
\mathbf{R(\mathbf{w}}(t))=\sum_{k=-N}^N \mathbf{\hat{R}}_k \text{e}^{i(2\pi k/T)t},
\label{eq_DFT_R}
\end{equation}
where $i=\sqrt{-1}$ is the imaginary unit, $T$ is the time period, $k$ is the wave number, and $N$ is the number of modes employed in the Discrete Fourier Transform (DFT). The $k^{th}$ Fourier coefficients $\mathbf{\hat{w}}_k$ and $\mathbf{\hat{R}}_k$ are given by the following equations (\ref{eq_Fourier_coeff_w}) and (\ref{eq_Fourier_coeff_R}), for $-N \leq k \leq N$ :
\begin{equation}
\mathbf{\hat{w}}_k = \frac{1}{2N+1}\sum_{n=0}^{2N} \Omega(t_n)\mathbf{w}(t_n)\text{e}^{-i(2\pi k/T)t_n},
\label{eq_Fourier_coeff_w}
\end{equation}
\begin{equation}
\mathbf{\hat{R}}_k = \frac{1}{2N+1}\sum_{n=0}^{2N} \mathbf{R(\mathbf{w}}(t_n))\text{e}^{-i(2\pi k/T)t_n},
\label{eq_Fourier_coeff_R}
\end{equation}
where their computations require the sampling of the modified state vector and the residual vector for $N_{ts}=2N+1$ time steps at equally spaced time instances such that the $n^{th}$ time sample $t_n$ is :
\begin{equation}
t_n = \frac{n}{2N+1}T \mbox{, for }n=0,..,2N.
\label{eq_time_instance}
\end{equation}

At this point, it is important to emphasize that the state and residual vectors need to be evaluated at all time instances before transferring in the Fourier domain, this is a fundamental difference with the time marching approach. The Fourier representation is then substituted into the semi-discrete form of the Euler equations (\ref{eq_semi_discrete_Euler_2}) to yield :
\begin{equation}
\frac{\partial}{\partial t}\left(\sum_{k=-N}^N \mathbf{\hat{w}}_k \text{e}^{i(2\pi k/T)t}\right) + \sum_{k=-N}^N \mathbf{\hat{R}}_k \text{e}^{i(2\pi k/T)t} = 0,
\end{equation}
\begin{equation}
\Leftrightarrow \sum_{k=-N}^N \frac{i2\pi k}{T}\mathbf{\hat{w}}_k \text{e}^{i(2\pi k/T)t} + \sum_{k=-N}^N \mathbf{\hat{R}}_k \text{e}^{i(2\pi k/T)t} = 0.
\label{eq_semi_discrete_Euler_2_NLFD}
\end{equation} 

By exploiting the orthogonality property of the Fourier basis, this leads to a set of $2N+1$ equations (\ref{eq_semi_discrete_Euler_NLFD_wavenumber}), each being associated to a wave number $k$ :
\begin{equation}
i\frac{2\pi k }{T}\hat{\mathbf{w}}_k+\hat{\mathbf{R}}_k=0 \mbox{ for } -N \leq k \leq N.
\label{eq_semi_discrete_Euler_NLFD_wavenumber}
\end{equation}

Since the representation of $\hat{\mathbf{R}}_k$ as a function of $\hat{\mathbf{w}}_k$ is not straightforward, an unsteady residual $\hat{\mathbf{R}}_k^*$ is defined and driven to zero using a pseudo-time marching approach such that :
\begin{equation}
\left\{
\begin{array}{l}
\hat{\mathbf{R}}_k^* = i\dfrac{2\pi k }{T}\hat{\mathbf{w}}_k+\hat{\mathbf{R}}_k \\
\\
\dfrac{\partial \hat{\mathbf{w}}_k}{\partial t^*} +\hat{\mathbf{R}}_k^* = 0
\end{array}
\right.
\mbox{, for } -N \leq k \leq N.
\label{eq_pseudo_time_NLFD}
\end{equation}
Thus at convergence, $\hat{\mathbf{R}}_k^*=0$ and the equations (\ref{eq_pseudo_time_NLFD}) are satisfied for each wave number. 

The new periodic solution is then transferred back to the physical time domain using an Inverse Fourier Discrete Transform (IDFT) and evaluated at each time instance $t_n$ by dividing by the volume :
\begin{equation}
\mathbf{w}(t_n) = \frac{\bar{\mathbf{w}}(t_n)}{\Omega(t_n)}  \mbox{, for } 0 \leq n \leq 2N.
\label{eq_solution_physical_time}
\end{equation}

The equation in pseudo-time can be solved using any time-stepping scheme. In this work, we use an hybrid five-stage Runge-Kutta scheme with blending coefficients for the artificial dissipation \cite{Jameson1995}. 

\subsection{Dynamic mesh deformation using the Radial Basis Function}
\label{subsec_RBF}
The deformation of the mesh is performed using the Radial Basis Functions (RBF) \cite{Tardif2017}. The method is based on the assumption that the movement of all grid points can be interpolated from the \textit{a priori} known motion of a set of points called the RBF points. In this study, the RBF points are always a subset of the grid points at the boundary of the domain, their displacements relative to the undeformed mesh are prescribed at each time instance using analytical functions. Because of the NLFD method, the mesh positions and velocities are therefore computed and stored for all $N_{ts}$ time steps. For any grid point $p$ of position vector $\mathbf{x}_p$ in the undeformed mesh, its displacement in the $x$-direction $s_x(\mathbf{x}_p,t)$ is defined as :
\begin{equation}
s_x(\mathbf{x}_p,t) = \sum_{i=1}^{N_{rbf}} \alpha_i (t) \phi(||\mathbf{x}_p - \mathbf{x}_i||_2),
\label{eq_disp_RBF}
\end{equation}
where $N_{rbf}$ is the number of RBF points, $\alpha_i$ are the interpolating coefficients, $\mathbf{x}_i$ is the position vector of the $i^{th}$ RBF point in the undeformed grid and $\phi$ is some basis function depending on the Euclidean distance $||\mathbf{x}_p- \mathbf{x}_i||_2$ between the points $p$ and $i$. In this work, Wendland C0's basis function \cite{Wendland1995} is considered, it is defined as follows :
\begin{equation}
\left\{
\begin{array}{lr}
(1-\xi)^2 & \text{ if } \xi < 1 \\
0         & 	\text{ if } \xi \ge 1
\end{array}
\right.
,
\text{ with }
\xi=\frac{||\mathbf{x}_p- \mathbf{x}_i||_2}{R},
\label{eq_WendlandC0}
\end{equation}
where $R$ is the support radius relative to the surface of RBF points. Since the equation (\ref{eq_disp_RBF}) holds for any grid point whether it is an RBF point or a standard grid point, in the following, the RBF points are denoted with the subscript $r$ while the grid (or volume) points are denoted with the subscript $v$. Then in the $x$-direction, the displacements of all RBF points and the interpolated displacements of all grid points are regrouped respectively in the vector $\mathbf{\Delta x}_r$ and in the vector $\mathbf{\Delta x}_v$. Therefore the \textit{a priori} unknown displacements $\mathbf{\Delta x}_v$ are obtained through equation (\ref{eq_disp_volume_points_matrix_form}) :
\begin{equation}
\mathbf{\Delta x}_v = \mathcal{A}(\mathcal{M}^{-1}) \mathbf{\Delta x}_r,
\label{eq_disp_volume_points_matrix_form}
\end{equation}
where :
\begin{equation}
\mathcal{M} =
\left(
\begin{array}{cccc}
\phi_{r_1 r_1} & \phi_{r_1 r_2} & \hdots & \phi_{r_1 r_{N_{rbf}}} \\
\phi_{r_2 r_1} & \ddots         &        &      \vdots            \\
\vdots        &                 &        &                        \\
\phi_{r_{N_{rbf}} r_1} &  \cdots       &        & \phi_{r_{N_{rbf}} r_{N_{rbf}}}
\end{array}
\right)
,
\mathcal{A} =
\left(
\begin{array}{cccc}
\phi_{v_1 r_1} & \phi_{v_1 r_2} & \hdots & \phi_{v_1 r_{N_{rbf}}} \\
\phi_{v_2 r_1} & \ddots         &        &      \vdots            \\
\vdots        &                 &        &                        \\
\phi_{v_{N_{grid}} r_1} &  \cdots       &        & \phi_{v_{N_{grid}} r_{N_{rbf}}}
\end{array}
\right)
\label{eq_RBF_matrix}
\end{equation}
with :
\begin{equation}
\phi_{v_i r_j} = \phi\left(||\mathbf{x}_{v_i}- \mathbf{x}_{r_j}||_2 \right) 
\end{equation}
and $N_{grid}$ is the total number of grid points. The displacements in the $y$ and $z$ directions can be computed with the same matrices given in equation (\ref{eq_RBF_matrix}), by considering the RBF points displacements in the corresponding direction.
 
Similarly, the mesh velocities for any grid point are computed using the Radial Basis Functions for Velocities (RBFV) by interpolating the \textit{a priori} known velocities of the RBF points which leads to the following expression :
\begin{equation}
\mathbf{v}_{dir,v} = \mathcal{A}(\mathcal{M}^{-1}) \mathbf{v}_{dir,r},
\end{equation}
where $\mathbf{v}_{dir,v}$ is the vector of the velocities of the grid points and $\mathbf{v}_{dir,r}$ is the vector of the velocities of the RBF points and the direction is given by $dir=x,y, \mbox{ or }z$.

\section{Derivation and enforcement of the Geometric Conservation Law}
\label{sec_GCL}

\subsection{Derivation of the GCL in the NLFD framework}
\label{subsec_DerivGCL}
As previously stated our interest is focused on the Volume Conservation Law aspect of the GCL. Under integral form the VCL for a control volume $\Omega$ enclosed by a boundary $\partial \Omega$ can be written as follows :
\begin{equation}
\frac{\partial}{\partial t}\int_{\Omega} \text{d}\Omega - \oint_{\partial \Omega} (\mathbf{V}_t \cdot \mathbf{n}) \text{d}S = 0.
\label{eq_GCL_int}
\end{equation} 
where $\mathbf{V}_t = \displaystyle \left(\frac{\partial x}{\partial t},\frac{\partial y}{\partial t},\frac{\partial z}{\partial t}\right)$ is the mesh velocity vector and $\mathbf{n}$ is the normal vector to the surface $\partial \Omega$. The law relates only on geometrical considerations and is always satisfied under continuous form and implicitly satisfied for rigid grid motion. It arises from the deformation of the mesh and is closely related to the preservation of uniform flow by the numerical scheme. Therefore in order to obtain a consistent solution method, the GCL must be discretized using the same numerical scheme employed to discretize the primary conservation laws \cite{Herve2000}. In our case, it yields a hexahedral stuctured finite-volume framework and a temporal discretization using the NLFD method. A first approach to enforce the VCL in the NLFD context was presented by Tardiff et al. \cite{Tardif2017} but more investigation is needed to determine its limitations. In this section, further developments are added to this approach which expose its analytical limits and a new method is proposed.

Considering any discretized control volume $\Omega$ enclosed by $N_f$ faces, then equation (\ref{eq_GCL_int}) can be written as :
\begin{equation}
\frac{\partial \Omega}{\partial t} - \sum_{m=1}^{N_f} \iint_{\partial\Omega_{m}} (\mathbf{V}_t \cdot \mathbf{n_{m}} ) dS =0,
\label{eq_semi_discrete_GCL_1}
\end{equation}
where $\mathbf{n_{m}}$ is the unit normal vector to the face $\partial \Omega_m$. Then the integrated face mesh velocities (IFMV) $G_m(t)$ corresponding to the temporal rate of change of the algebraic volume swept by each face through time are introduced in equation (\ref{eq_IFV_int}) :
\begin{equation}
G_m(t)=\iint_{\partial\Omega_{m}} (\mathbf{V}_t \cdot \mathbf{n_{m}} ) dS,
\label{eq_IFV_int}
\end{equation}
and also $G(t)$ is the sum of the IFMV over all faces of the control volume :
\begin{equation}
G(t)=\sum_{m=1}^{N_f}G_m(t).
\label{eq_G_split}
\end{equation}
Then equation (\ref{eq_semi_discrete_GCL_1}) can be written as :
\begin{equation}
\frac{\partial \Omega}{\partial t} - G(t) = 0.
\label{eq_semi_discrete_GCL_2}
\end{equation}
Under the assumption that the volume $\Omega$ and the sum of the integrated face mesh velocities $G$ are periodic functions of time, the NLFD discretization can be applied :
\begin{equation}
\Omega(t) = \sum_{k=-N}^{N} \hat{\Omega}_k \text{e}^{i(2\pi k/T)t},
\label{eq_DFT_Omega}
\end{equation}
\begin{equation}
G(t) = \sum_{k=-N}^{N} \hat{G}_k \text{e}^{i(2\pi k/T)t}.
\label{eq_DFT_G}
\end{equation}
By substituting these expressions into equation (\ref{eq_semi_discrete_GCL_2}), yields :
\begin{equation}
\frac{\partial}{\partial t}\left(\sum_{k=-N}^N \hat{\Omega}_k \text{e}^{i(2\pi k/T)t}\right) - \sum_{k=-N}^N \hat{G}_k \text{e}^{i(2\pi k/T)t} = 0
\end{equation}
\begin{equation}
\Leftrightarrow \left(\sum_{k=-N}^N \frac{i2\pi k}{T}\hat{\Omega}_k \text{e}^{i(2\pi k/T)t}\right) - \sum_{k=-N}^N \hat{G}_k \text{e}^{i(2\pi k/T)t} = 0.
\end{equation}
Then by exploiting the orthogonality property of the Fourier basis, it leads to a system of $2N+1$ equations, each corresponding to a wave number $k$ :
\begin{equation}
\frac{i2\pi k}{T}\hat{\Omega}_k = \hat{G}_k \mbox{ for } -N \leq k \leq N.
\label{eq_GCL_NLFD_compact}
\end{equation}
The set of equations (\ref{eq_GCL_NLFD_compact}) provides the necessary condition to enforce the GCL in the NLFD approach. Such criterion is not satisfied in general and has to be enforced through the correct computation of the cell volume and the integrated face mesh velocities, in a way consistent with the solver numerical scheme. Since the volume is usually exactly known, one popular approach in time marching methods is to split the GCL over each face \cite{Zhang1993,Mavripilis2006,Mavripilis2011}. In the current framework, the volume of a cell can be expressed as the sum of the volume at a reference initial instant $t_0$ and the algebraic (positive or negative) volumetric increments due to each face $\Omega_m$ relative to this reference instant :
\begin{equation}
\Omega(t) =\Omega(t_0) + \sum_{m=1}^{N_f} \Omega_{m}(t).
\label{eq_Omega_split}
\end{equation}
By substituting relations (\ref{eq_Omega_split}) and (\ref{eq_G_split}) into the equation (\ref{eq_semi_discrete_GCL_2}), yields :
\begin{equation}
\sum_{m=1}^{N_f} \left( \frac{\partial \Omega_{m}}{\partial t}- G_{m}(t) \right)=0.
\label{eq_GCL_split_sum}
\end{equation}
Then for each face $m$ enclosing the discretized control volume, we need to ensure the relation (\ref{eq_GCL_split}) :
\begin{equation}
\frac{\partial \Omega_m}{\partial t}=G_m(t).
\label{eq_GCL_split}
\end{equation}
However, even if the positions of the mesh vertices and their velocities are known at all time instances from the dynamic mesh deformation, the implementation of the GCL using this relation is not straightforward using the NLFD method. In the following, volumetric increments are always considered as algebraic values which can either be positive or negative.

\subsection{Approach of Tardiff et al. \cite{Tardif2017}}
\label{subsec_Tardiff_approach}
The first approach developed by Tardiff et al. \cite{Tardif2017} is based on a linear representation of the volumetric {increments} relative to a reference time instance $t_0$. For any face $m$ defined by its vertices the induced volumetric {change} would simply be represented by drawing straight lines from their initial position at $t_0$ to their position at time instant $t$, see Figure \ref{fig_linear_Omegam}.

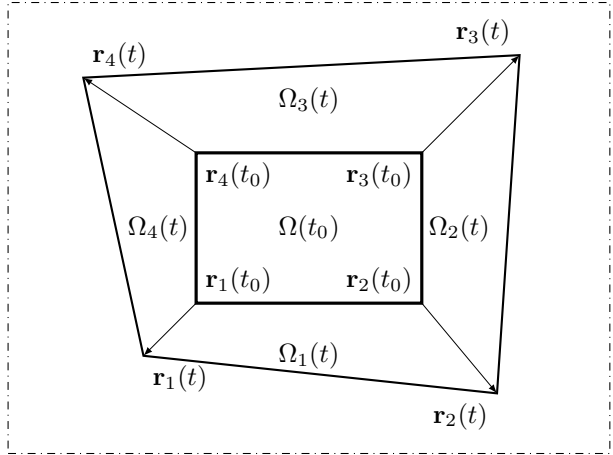
\begin{figure}[!htbp]
\centering
	\begin{tikzpicture}[scale=1]
		\draw[dash dot] (0,0)rectangle(8,6);
		\draw[very thick] (2.5,2)rectangle(5.5,4);
		\draw[very thin,->,>=latex] (2.5,2)--(1.8,1.3);
		\draw[very thin,->,>=latex] (5.5,2)--(6.5,0.8);
		\draw[very thin,->,>=latex] (5.5,4)--(6.8,5.3);
		\draw[very thin,->,>=latex] (2.5,4)--(1,5);
		\draw[thick] (1.8,1.3)--(6.5,0.8)--(6.8,5.3)--(1,5)--cycle;
   			\node[above right] at (2.5,2) {$\mathbf{r}_1(t_0)$};
   			\node[above left] at (5.5,2) {$\mathbf{r}_2(t_0)$};
   			\node[below left] at (5.5,4) {$\mathbf{r}_3(t_0)$};
   			\node[below right] at (2.5,4) {$\mathbf{r}_4(t_0)$};   			
    			\node[below right] at (1.8,1.3) {$\mathbf{r}_1(t)$};
   			\node[below left] at (6.5,0.8) {$\mathbf{r}_2(t)$};
   			\node[above left] at (6.8,5.3) {$\mathbf{r}_3(t)$};
   			\node[above right] at (1,5) {$\mathbf{r}_4(t)$};
			\node[] at (4,3) {$\Omega(t_0)$}; 
   			\node[] at (4,1.3) {$\Omega_1(t)$};
   			\node[] at (6,3) {$\Omega_2(t)$};
   			\node[] at (4,4.7) {$\Omega_3(t)$};
   			\node[] at (2,3) {$\Omega_4(t)$};			   				
	\end{tikzpicture}
\caption{Example of linear volumetric increments in 2D relatively to a reference time instant $t_0$}
\label{fig_linear_Omegam}
\end{figure}

This approach has two advantages : first, it is easy to compute the volumetric {increments} at each time instant using standard cell volume computational algorithms ; second, the volumetric {increments} due to each face are time periodic as long as the movement of the vertices is periodic. 

Once the volumetric {increments} are known for $2N+1$ time instances defined by equation (\ref{eq_time_instance}), their Fourier representations are calculated :
\begin{equation}
\Omega_{m}(t) = \sum_{k=-N}^{N} \hat{\Omega}_{m,k} \text{e}^{(i2\pi k/T)t},
\label{eq_Omegam_Fourier}
\end{equation}
and the Fourier formulations of the integrated face mesh velocities for each face $m$ are introduced :
\begin{equation}
G_{m}(t) = \sum_{k=-N}^{N} \hat{G}_{m,k} \text{e}^{(i2\pi k/T)t}.
\label{eq_Gm_Fourier}
\end{equation} 

Then by substituting, the Fourier representations into criteria (\ref{eq_GCL_split}), and exploiting the orthogonality of the Fourier basis, a system of $2N+1$ equations (\ref{eq_GCL_split_NLFD_simplified}) is obtained for each face $m$  :
\begin{equation}
\frac{i2\pi k}{T}\hat{\Omega}_{m,k} = \hat{G}_{m,k} \mbox{ for } -N \leq k \leq N.
\label{eq_GCL_split_NLFD_simplified}
\end{equation}
Therefore, the GCL are satisfied independently for each face of the control volume by computing the Fourier coefficients $\hat{G}_{m,k}$ and then applying an IDFT to transfer back the integrated face mesh velocities to the temporal domain. Despite its attractiveness, this method is restricted to linear movements due to the manner in which the volumetric increments are computed. In general, the motion would not be linear and such representation of the volumetric increments will not be sufficient to ensure the correct computation of the IFMV. 

Moreover the NLFD method is based on the assumption that the quantities are time periodic and can be expanded in Fourier series, but having a time periodic movement of the vertices does not guarantee time periodic volumetric increments but only that their temporal derivative will be periodic. This statement will be demonstrated through the following example. 

A 2D quadrilateral element is considered with the following motion defined by equation (\ref{eq_non_periodic_OmegaM_move}) and shown Figure  \ref{fig_non_periodic_OmegaM_move} :
\begin{equation}
\left\{
\begin{array}{lcl}
\alpha(t)    & = & 2\pi t, \\
\mathbf{r_1} & = & \mathbf{r_{1,0}}, \\
\mathbf{r_2} & = & \mathbf{r_{2,0}}, \\
\mathbf{r_3} & = & \mathbf{r_{3,0}}+R(1-\cos(\alpha(t)))\mathbf{e_x}+R(\sin(\alpha(t))) \mathbf{e_y}, \\
\mathbf{r_4} & = & \mathbf{r_{4,0}}, \\
\end{array}
\right.
\label{eq_non_periodic_OmegaM_move}
\end{equation}
where the index $0$ refers to the initial position of the grid, $R$ is the radius defining the amplitude of the circular motion and $\mathbf{e_x}$ and $\mathbf{e_y}$ are the unit vectors in respectively the $x$ and $y$ directions.

\setlength{\unitlength}{1mm}
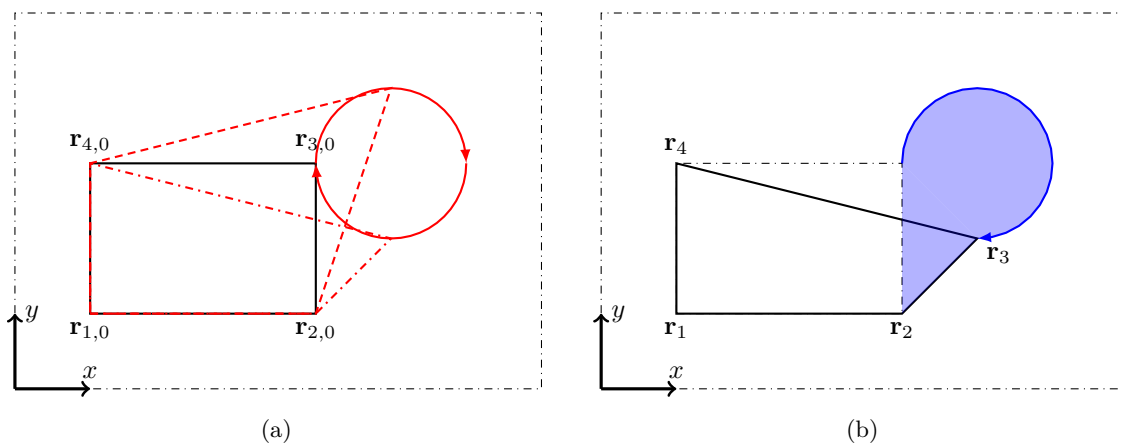
\begin{figure}[!htbp]
    \begin{subfigure}[t]{0.5\textwidth}
    \centering
	\begin{tikzpicture}[scale=1]
   			\draw[dash dot] (0,0)rectangle(7,5);
   			\draw[very thick,->] (0,0) -- (1,0) node[above]{$x$};
   			\draw[very thick,->] (0,0) -- (0,1) node[above,right]{$y$};
   			\draw[black,thick] (1,1) -- (4,1) -- (4,2.999) -- (1,3) -- cycle;
   			\draw[red,thick, densely dashed] (1,1) -- (4,1) -- (5,4) -- (1,3) -- cycle;
   			\draw[red,thick,dash dot] (1,1) -- (4,1) -- (5,2) -- (1,3) -- cycle;
   			\draw[red,thick,domain=0:180,<-,>=latex] plot ({5+cos(\x)}, {3+sin(\x)});
   			\draw[red,thick,domain=180:360,<-,>=latex] plot ({5+cos(\x)}, {3+sin(\x)});
   			\node[below] at (1,1) {$\mathbf{r}_{1,0}$};
   			\node[below] at (4,1) {$\mathbf{r}_{2,0}$};
   			\node[above] at (4,3) {$\mathbf{r}_{3,0}$};
   			\node[above] at (1,3) {$\mathbf{r}_{4,0}$};
	\end{tikzpicture}
	\caption{}
    \end{subfigure}
    ~
    \begin{subfigure}[t]{0.5\textwidth}
    \centering
	\begin{tikzpicture}[scale=1]
            \draw[dash dot] (0,0)rectangle(7,5);
   			\draw[very thick,->] (0,0) -- (1,0) node[above]{$x$};
   			\draw[very thick,->] (0,0) -- (0,1) node[above,right]{$y$};
   			\draw[black,dash dot] (1,1)rectangle(4,3);
   			\draw[black,thick] (1,1) -- (4,1) -- (5,2) -- (1,3) -- cycle;
   			\draw[blue,thick,domain=-90:180,<-,>=latex] plot ({5+cos(\x)}, {3+sin(\x)});
   			\fill[blue,opacity=0.3,thick,domain=-90:180] plot ({5+cos(\x)}, {3+sin(\x)});
   			\fill[blue,opacity=0.3,thick] (4,1) -- (5,2) -- (4,3) -- cycle;
            \node[below] at (1,1) {$\mathbf{r}_{1}$};
   			\node[below] at (4,1) {$\mathbf{r}_{2}$};
   			\node[below right] at (5,2) {$\mathbf{r}_{3}$};
   			\node[above] at (1,3) {$\mathbf{r}_{4}$};
	\end{tikzpicture}
	\caption{}
    \end{subfigure}
    \caption{(a) The initial undeformed quadrilateral element is shown in black while the movement of the mesh points is presented in red with two deformed configurations of the cell in dashed lines, (b) Exact volumetric {increment} for the face 2-3 in the $x$ direction relatively to the initial configuration (in dash dot line) in blue}
    \label{fig_non_periodic_OmegaM_move}
\end{figure}

For the face defined by the vertices $\mathbf{r}_2$ and $\mathbf{r}_3$, the derivation of the expression of the exact volumetric {increment} in the $x$ direction and its time derivative leads to the following expressions respectively (\ref{eq_exact_volinc}) and (\ref{eq_exact_dvolinc}) :
\begin{equation}
\Omega_{23,x}(t) = \frac{R^2}{2}(\alpha(t)-\sin(\alpha(t)))+\frac{Ry_{3,0}}{2}(1-\cos(\alpha(t))), 
\label{eq_exact_volinc}
\end{equation}
\begin{equation}
\frac{\partial \Omega_{23,x}(t)}{\partial t} = \frac{R^2}{2}\frac{\partial \alpha}{\partial t}(1-\cos(\alpha(t)))+\frac{Ry_{3,0}}{2}\frac{\partial \alpha}{\partial t}\sin(\alpha(t)),
\label{eq_exact_dvolinc}
\end{equation}
where the length $y_{3,0}=(\mathbf{r}_{3,0} \cdot \mathbf{e_y})$.

Thus the time derivative of the volumetric {increment} is periodic whereas the volumetric {increment} is the sum of a linear term and a periodic term and the direct application of the NLFD method on the exact volumetric {increment} is not possible since the linear term is not expandable as a Fourier serie. Additional work is required to ensure equation (\ref{eq_GCL_split}) is compliant with the NLFD method.

\subsection{Modified approach based on the exact volumetric {increments}}
\label{subsec_modified_approach}

\subsubsection{Method}
\label{subsubsec_method}

In this section, we are going to demonstrate the following theorem \ref{th_new_method},
\begin{theorem}
Let $\Omega$ be a discretized control volume, enclosed by $N_f$ faces, and subjected to a periodic motion of its vertices. Then given the knowledge of the exact volumetric {increments} $\Omega_m$ for $m=1,...,N_f$, a sufficient condition to ensure the satisfaction of GCL in the NLFD framework is the computation of the integrated face mesh velocities, where the zeroth and higher modes can be expressed as
\begin{equation}
\hat{G}_{m,0} = \frac{\Omega_{m}(T)}{T}, 
\label{eq_Gm0}
\end{equation}
\begin{equation}
\hat{G}_{m,k} = \frac{i2\pi k}{T}\hat{p}_{m,k} \mbox{ for } -N \leq k \leq N, \mbox{ } k \neq 0,
\label{eq_GCL_NLFD_Gmk}
\end{equation} 
where $\hat{G}_{m,k}$ and $\hat{p}_{m,k}$ are the Fourier coefficients of respectively the integrated face mesh velocities and the periodic part of the exact volumetric increments given by,
\begin{equation}
p_{m}(t) =  \Omega_{m}(t)-\left( \frac{\Omega_{m}(T)}{T}\right) t.
\label{eq_alpha2_periodic_simplified}
\end{equation}
\label{th_new_method}
\end{theorem}

\begin{proof}
Under the assumption that the motion of the vertices is periodic, the temporal rate of change of the algebraic volume swept by each face through time is periodic. Thus the temporal derivative of the volumetric increments and the integrated face mesh velocities are periodic, the DFT is applied to the equation (\ref{eq_GCL_split}) leading to :
\begin{equation}
G_m(t) = \frac{\partial \Omega_m}{\partial t} = \hat{G}_{m,0} + \sum_{k=-N,k\neq 0}^{N} \hat{G}_{m,k} \text{e}^{i\frac{2\pi}{T}kt},
\label{eq_DFT_Gm}
\end{equation}
where $\hat{G}_{m,k},\mbox{ for }-N\leq k \leq N$ are the Fourier coefficients of both the derivative of the volumetric increment and the integrated face mesh velocity of a face $m$. 

By integrating the equation in time, the volumetric increment is expressed as :
\begin{equation}
\Omega_{m}(t) = \int \frac{\partial \Omega_{m}}{\partial t} dt = \hat{\Omega}_{m,0} + \hat{G}_{m,0} t + \sum_{k=-N,k \neq 0}^{N} \frac{T}{i2\pi k}\hat{G}_{m,k} e^{i\frac{2\pi}{T}kt}.
\label{eq_int_dOmegam}
\end{equation}
where $\hat{\Omega}_{m,0}$ is a constant of integration. Then any volumetric {increment} can be interpreted as the sum of a linear term $l_{m}(t)$  and a periodic function $p_{m}(t)$ defined by :
\begin{equation}
l_{m}(t) = \hat{G}_{m,0} t,
\label{eq_Omegam_lm}
\end{equation} 
\begin{equation}
p_{m}(t) = \hat{\Omega}_{m,0}+\sum_{k=-N,k \neq 0}^{N} \frac{T}{i2\pi k}\hat{G}_{m,k} e^{i\frac{2\pi}{T}kt}.
\label{eq_Omegam_pm}
\end{equation}
Knowing the values of the volumetric increment at $t=t_0$ and $t=t_0+T$, and exploiting the periodicity of the function $p_{m}$, yields :
$$
\begin{array}{lcl}
\Omega_{m}(t_0) & = & \hat{G}_{m0}t_0 + p_{m}(t_0) \\
\Omega_{m}(t_0+T) & = & \hat{G}_{m0} (t_0+T) + p_{m}(t_0+T) \\
p_{m}(t_0)     & = & p_{m}(t_0+T)
\end{array}
$$

\begin{equation}
\hat{G}_{m,0} = \frac{\Omega_{m}(t_0+T)-\Omega_{m}(t_0)}{T} 
\label{eq_Gm0_proof}
\end{equation}
Hence the zeroth Fourier coefficients of the integrated face mesh velocities are known through equation (\ref{eq_Gm0_proof}) applied for each face $m$ and the linear part $l_{m}$ of the volumetric increments can be computed at each instant. Then, an expression of the periodic part of any volumetric increment $p_{m}$ is obtained as :
\begin{equation}
p_{m}(t) = \Omega_{m}(t)-l_{m}(t) =  \Omega_{m}(t)-\left( \frac{\Omega_{m}(t_0+T)-\Omega_{m}(t_0)}{T}\right) t.
\label{eq_pm_periodic_proof}
\end{equation}
Usually $t_0$ would be taken as the initial time instant $t_0=0$ corresponding to the undeformed configuration of the mesh, for this specific reference time instant $\Omega_{m}(0)=0$, and the previous expression can be further simplified into equation (\ref{eq_pm_periodic_simplified_proof}) :
\begin{equation}
p_{m}(t) =  \Omega_{m}(t)-\left( \frac{\Omega_{m}(T)}{T}\right) t.
\label{eq_pm_periodic_simplified_proof}
\end{equation}
Therefore, at each instant $t$ the periodic part of the volumetric increments $p_{m}$ are known and we introduce the Fourier coefficients for $p_{m}$, noted as $\hat{p}_{m,k}$ for $-N \leq k \leq N$. By calculating the temporal derivative in Fourier space of $p_{m}$ and exploiting the orthogonality of the Fourier basis functions, the rest of the Fourier coefficients of the integrated face mesh velocities $\hat{G}_{m,k}$ are deduced from a system of $2N$ equations  
\begin{equation}
\hat{G}_{m,k} = \frac{i2\pi k}{T}\hat{p}_{m,k} \mbox{ for } -N \leq k \leq N, \mbox{ } k \neq 0.
\label{eq_GCL_NLFD_Gmk_proof}
\end{equation} 
Since the derivation in Fourier space puts to zero the contribution from the zeroth coefficient, the value of the integration constant $\hat{\Omega}_{m,0}$ is not relevant to compute the integrated face mesh velocities.
\end{proof}

Finally the procedure to compute the IFMV to enforce GCL by deducing the time derivative of the volumetric increments for each face is given by the pseudo-code (Algorithm \ref{algo_method}).
\begin{algorithm}[!htbp]
\SetAlgoLined
\For{$n = 0,...,2N$}{Calculate the mesh deformation using the RBF for equally space time instances $t_n$ \;}
\For{$face=1,...,face_{max}$}{
\For{$n = 0,...,2N+1$}{Calculate the volumetric increments $\Omega_{face}(t_n)$ \;}
Deduce the zeroth Fourier coefficient via $\hat{G}_{face,0}=\dfrac{\Omega_{face}(t_{2N+1})}{T}$ \;
\For{$n = 0,...,2N$}{Extract the periodic part of the volumetric increments via $p_{face}(t_n)=\Omega_{face}(t_n)-\hat{G}_{face,0}t_n$ \;}
Compute the Fourier coefficients $\hat{p}_{face,k}$ via FFT on $p_{face}(t)$ \;
\For{$k = -N,...,1 \mbox{ and } k=1,...N$}{Deduce the $k^{th}$ Fourier coefficient of the integrated face mesh velocities via $\hat{G}_{face,k}=\dfrac{i2\pi k}{T}\hat{p}_{face,k}$ \;}
Compute the integrated face mesh velocities $G_{face}(t)$ via IFFT on $\hat{G}_{face,k}$ with $-N\leq k \leq N$\;
}
\caption{Pseudo-code representing the derived procedure to compute the integrated face mesh velocities and ensure GCL in the NLFD framework}
\label{algo_method}
\end{algorithm}
It is important to note that since the values of the volumetric increments are required at $t=T$ in order to deduce the zeroth Fourier coefficients $\hat{G}_{m,0}$ through equation (\ref{eq_Gm0}), one additional time step is needed $t_{2N+1}$ compared to the number of time steps for the flow solver. However for this final time step the configuration of the mesh is the same as the initial (the undeformed mesh), thus no additional time step is needed for the mesh deformation. For this procedure, the key point is to compute the exact volumetric increments as accurately as possible in order to preserve the spectral convergence of the NLFD method.

\subsubsection{Practical enforcement and error estimation}
\label{subsubsec_prac_enf}
In practice the accuracy of the previous method highly depends on the accuracy of the computation of the volumetric increments. For an hexahedral grid, as each face $m$ sweeps through the computational domain, they form hexahedra between time intervals. The volume of any hexahedreon can be computed using a trilinear mapping \cite{Dukowicz1988} between the physical space and the computational domain Figure \ref{fig_tri_mapping}. This yields the following definitions.

\begin{definition}
The volume of any hexahedron as a function of the position vectors of the vertices in the physical space $\mathbf{r}_i$ for $i=1,...,8$, is evaluated through,
\begin{equation}
\left\{
\begin{array}{l}
\Omega_{h}=(\Omega_{4321}+\Omega_{5678}+\Omega_{3487}+\Omega_{1256}+\Omega_{4158}+\Omega_{2376}), \\
\\
\mbox{with } \Omega_{ijkl}=\dfrac{1}{12}(\mathbf{r}_j+\mathbf{r}_k)\cdot ((\mathbf{r}_i+\mathbf{r}_j)\times(\mathbf{r}_i+\mathbf{r}_l)).
\end{array}
\right.
\label{eq_tri_mapping_volume}
\end{equation}
\label{def_vol_hex}
\end{definition}

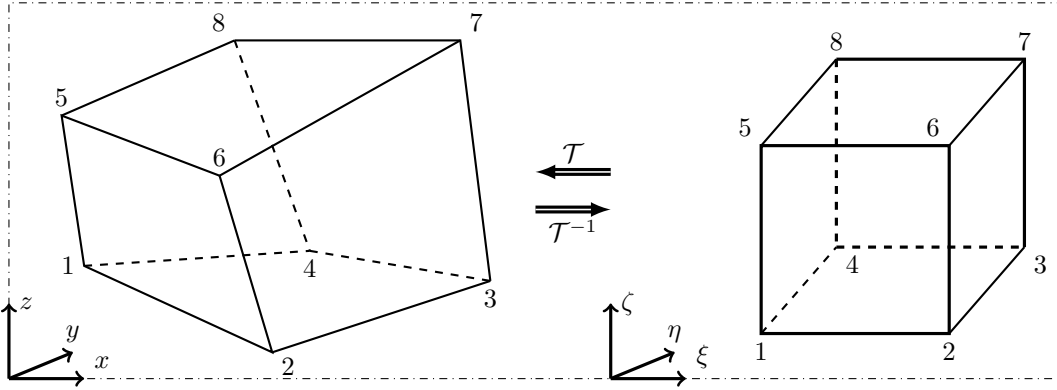
\begin{figure}[!htbp]
\centering
	\begin{tikzpicture}
		\draw[dash dot] (0,0)--(7,0) (7,5)--(0,5)--(0,0);
		   	\draw[very thick,->] (0,0) -- (1,0) node[above right]{$x$};
		   	\draw[very thick,->] (0,0) -- (0.85,0.35) node[above]{$y$};
   			\draw[very thick,->] (0,0) -- (0,1) node[above,right]{$z$};
   			
		\draw[dash dot] (7,0)--(14,0)--(14,5)--(7,5);
		   	\draw[very thick,->] (8,0) -- (9,0) node[above right]{$\xi$};
		   	\draw[very thick,->] (8,0) -- (8.85,0.35) node[above]{$\eta$};
   			\draw[very thick,->] (8,0) -- (8,1) node[above,right]{$\zeta$};		
		
			\draw[thick] (1,1.5)--(3.5,0.35)--(2.8,2.7)--(0.7,3.5)--cycle;
		\draw[thick, dashed] (4,1.7)--(6.4,1.3);
		\draw[thick] (6.4,1.3)--(6,4.5)--(3,4.5);
		\draw[thick, dashed] (3,4.5)--(4,1.7);
			\draw[thick] (3.5,0.35)--(6.4,1.3);
			\draw[thick] (2.8,2.7)--(6,4.5);
			\draw[thick] (0.7,3.5)--(3,4.5);
			\draw[thick,dashed] (1,1.5)--(4,1.7);
			
			\draw[very thick] (10,0.6)--(12.5,0.6)--(12.5,3.1)--(10,3.1)--cycle;
		\draw[very thick, dashed] (11,1.75)--(13.5,1.75);
		\draw[very thick] (13.5,1.75)--(13.5,4.25)--(11,4.25);	
		\draw[very thick, dashed] (11,4.25)--(11,1.75);			
			\draw[thick] (12.5,0.6)--(13.5,1.75);
			\draw[thick] (12.5,3.1)--(13.5,4.25);
			\draw[thick] (10,3.1)--(11,4.25);
			\draw[thick,dashed] (10,0.6)--(11,1.75);		
			
		\node[left] at (1,1.5) {1};
		\node[below right] at (3.5,0.4) {2};	
		\node[above] at (2.8,2.7) {6};
		\node[above] at (0.7,3.5) {5};
		\node[below] at (4,1.7) {4};
		\node[below] at (6.4,1.3) {3};	
		\node[above right] at (6,4.5) {7};
		\node[above left] at (3,4.5) {8};
			\node[below] at (10,0.6) {1};	
			\node[below] at (12.5,0.6) {2};
			\node[above left] at (12.5,3.1) {6};
			\node[above left] at (10,3.1) {5};
			\node[below right] at (11,1.75) {4};
			\node[below right] at (13.5,1.75) {3};
			\node[above] at (13.5,4.25) {7};
			\node[above] at (11,4.25) {8};
			
		\draw[very thick, double,<-,>=latex] (7,2.75) -- (8,2.75) node[above] at (7.5,2.75){$\mathcal{T}$} ;	
		\draw[very thick, double,->,>=latex] (7,2.25) -- (8,2.25) node[below] at (7.5,2.25){$\mathcal{T}^{-1}$} ;			
	\end{tikzpicture}
\caption{Trilinear mapping between a hexahedron in the physical space and a reference cube in the computational domain}	
\label{fig_tri_mapping}
\end{figure}

\begin{definition}
For any face $m$ of an hexahedral cell, the exact volumetric increment is estimated through a sum of hexahedra each corresponding to the approximated volumetric increment between two time samples $t_{n-1}$ and $t_n$ and noted as $\Omega_{m,h}(t_n)$, (see Figure \ref{fig_OmegaM_move}) :
\begin{equation}
\left\{
\begin{array}{l}
\Omega_m(t_0) = 0,\\
\\
\Omega_m(t_n) = \left( \displaystyle \sum_{k=1}^n \Omega_{m,h}(t_k) \right) + \epsilon^T_m(t_{n}) \mbox{, for }1 \leq n \leq 2N+1,
\end{array}
\right.
\label{eq_Omegam_practice}
\end{equation}
where $t_0=0$ is the initial instant corresponding to the undeformed mesh and $\epsilon^T_m(t_{n})$ is the truncation error at time instant $t_n$. 
\label{def_vol_inc}
\end{definition} 

\setlength{\unitlength}{1mm}
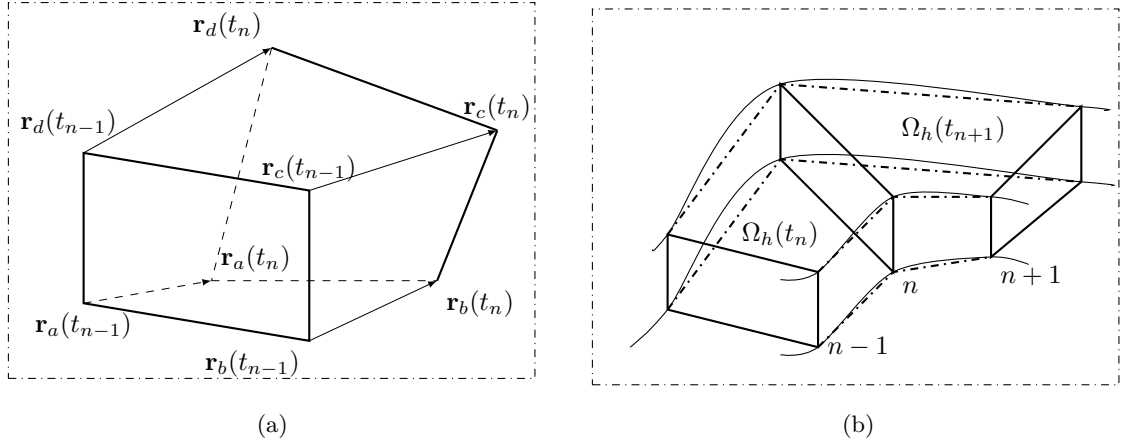
\begin{figure}[!htbp]
    \begin{subfigure}[t]{0.5\textwidth}
    \centering
	\begin{tikzpicture}[scale=1]
   			\draw[dash dot] (0,0)rectangle(7,5);
   			\draw[black,thick] (1,1) -- (4,0.5) -- (4,2.5) -- (1,3) -- cycle;
   				\draw[very thin, dashed,->,>=latex] (1,1)--(2.7,1.3);
   				\draw[very thin,->,>=latex] (4,0.5)--(5.7,1.3);
   				\draw[very thin,->,>=latex] (4,2.5)--(6.5,3.3);
   				\draw[very thin,->,>=latex] (1,3)--(3.5,4.4);
   			\draw[very thin, dashed] (2.7,1.3)--(5.7,1.3);
   			\draw[thick] (5.7,1.3)--(6.5,3.3);
   			\draw[thick] (6.5,3.3)--(3.5,4.4);
   			\draw[very thin, dashed] (3.5,4.4)--(2.7,1.3);	
   			    	\node[below] at (1,1) {$\mathbf{r}_a(t_{n-1})$};
   				\node[below left] at (4,0.5) {$\mathbf{r}_b(t_{n-1})$};
   				\node[above] at (4,2.5) {$\mathbf{r}_c(t_{n-1})$};
   				\node[above] at (0.8,3.1) {$\mathbf{r}_d(t_{n-1})$};
   			\node[above right] at (2.7,1.3) {$\mathbf{r}_a(t_{n})$};
   			\node[below right] at (5.7,1.3) {$\mathbf{r}_b(t_{n})$};
   			\node[above] at (6.5,3.3) {$\mathbf{r}_c(t_{n})$};
   			\node[above left] at (3.5,4.4) {$\mathbf{r}_d(t_{n})$};   				
	\end{tikzpicture}
	\caption{}
    \end{subfigure}
    ~
    \begin{subfigure}[t]{0.5\textwidth}
    \centering
	\begin{tikzpicture}[scale=1]
            \draw[dash dot] (0,0)rectangle(7,5);
			\draw[thick] (1,1)--(3,0.5)--(3,1.5)--(1,2)--cycle; \node[right] at (3,0.5) {$n-1$};
			\draw[thick] (2.5,3)--(4,1.5)--(4,2.5)--(2.5,4)--cycle; \node[below right] at (4,1.5) {$n$};
			\draw[thick] (6.5,2.7)--(5.3,1.7)--(5.3,2.5)--(6.5,3.7)--cycle; \node[below right] at (5.3,1.7) {$n+1$};		
				\draw [thick, dash dot] (1,1)--(2.5,3)--(6.5,2.7);
				\draw [thick, dash dot] (3,0.5)--(4,1.5)--(5.3,1.7);
				\draw [thick, dash dot] (3,1.5)--(4,2.5)--(5.3,2.5);
				\draw [thick, dash dot] (1,2)--(2.5,4)--(6.5,3.7);
			\draw [thin] plot [smooth] coordinates {(0.5, 0.5) (1,1) (2.5,3) (6.5,2.7) (6.9,2.65)};
			\draw [thin] plot [smooth] coordinates {(2.5, 0.4) (3,0.5) (4,1.5) (5.3,1.7) (5.8,1.6)};
			\draw [thin] plot [smooth] coordinates {(2.5, 1.4) (3,1.5) (4,2.5) (5.3,2.5) (5.8,2.4)};
			\draw [thin] plot [smooth] coordinates {(0.8, 1.8) (1,2) (2.5,4) (6.5,3.7) (6.8,3.65)};
				\node[] at (2.5,2) {$\Omega_{h}(t_n)$};
				\node[] at (4.8,3.4) {$\Omega_{h}(t_{n+1})$};
	\end{tikzpicture}
	\caption{}
    \end{subfigure}
    \caption{(a) Approximated volumetric increment between two time steps $t_{n-1}$ and $t_n$ : $\Omega_{h}(t_n)$ (b) Approximation of a volumetric increment as a sum of hexahedra}
    \label{fig_OmegaM_move}
\end{figure}

Now that the mathematical tools to compute the volumetric increments are introduced, the accuracy of the procedure presented in section \ref{subsec_modified_approach} can be established, we have the first lemma \ref{lemma_Gm0},
\begin{lemma}
In the context of theorem \ref{th_new_method}, and under the definitions \ref{def_vol_hex} and \ref{def_vol_inc}, for any face $m$ the temporal-order of accuracy of the zeroth Fourier coefficient of the integrated face mesh velocity $\hat{G}_{m,0}$ is limited to one.
\label{lemma_Gm0}
\end{lemma}

\begin{proof}
For each face $m$ the path of the four corresponding vertices between two time steps is linearly approximated, as shown in Figure \ref{fig_non_periodic_OmegaM_move}. Then, for any of these vertices $\mathbf{r}_i, \mbox{ }i=a,b,c,d$ at the $n^{th}$ time sample, we have the Taylor expansion :
\begin{equation}
\mathbf{r}_i(t_n) = \mathbf{r}_i(t_{n-1}) + \left( \frac{\partial \mathbf{r}_i}{\partial t}(t_{n-1})\right)(t_n-t_{n-1}) + \underset{\text{Truncation error on the vertex path}}{\underbrace{\frac{1}{2}\left(\frac{\partial^2\mathbf{r}_i}{\partial t^2}(t_{n-1})\right)(t_n-t_{n-1})^2 + \mathcal{O}((t_n-t_{n-1})^3)}}. 
\end{equation}
The volumetric increment is then defined by the vertices positions with the following indexation,
\begin{equation}
\left\{
\begin{array}{lclclclclclclcl}
\mathbf{r}_1 & = & \mathbf{r}_a(t_{n-1}), & \mathbf{r}_2 & = & \mathbf{r}_b(t_{n-1}), & \mathbf{r}_3 & = & \mathbf{r}_c(t_{n-1}), & \mathbf{r}_4 & = & \mathbf{r}_d(t_{n-1}), \\
\mathbf{r}_5 & = & \mathbf{r}_a(t_{n}), & \mathbf{r}_6 & = & \mathbf{r}_b(t_{n}), & \mathbf{r}_7 & = & \mathbf{r}_c(t_{n}), & \mathbf{r}_8 & = & \mathbf{r}_d(t_{n}). 
\end{array}
\right.
\end{equation}
By substitution of the Taylor expansion of the vertex positions at instant $t_n$ into equation (\ref{eq_tri_mapping_volume}), and exploiting the linearity of the function, the error committed during the estimation of the volumetric increment between two time steps is found to be of order two in $\tau=(t_n-t_{n-1})$ (see Appendix \ref{appendix_trunc_error}). Recalling that the number of time steps $N_{ts}=2N+1$ and using the definition of the time instance, the difference $(t_n-t_{n-1})$ is written as 
\begin{equation}
\tau=\frac{[n-(n-1)]T}{2N+1}=\frac{T}{N_{ts}}.
\end{equation}
Thus the truncation error during the estimation of the volumetric increment between two time steps at the $n^{th}$ instant and noted $\epsilon_{m,h}^T(t_n)$ is of order two in $\tau$ and can be expanded as :
\begin{equation}
\left\{
\begin{array}{l}
\epsilon^T_{m,h}(t_{0})=0, \\
\\
\epsilon^T_{m,h}(t_n) = \mathcal{E}_m^T(t_{n-1})\tau^2 + \mathcal{O}(\tau^3) \mbox{, for }1 \leq n \leq 2N+1,
\end{array}
\right.
\label{eq_truncation_error_h}
\end{equation}
where $\mathcal{E}_m^T(t)$ is a scalar periodic function depending on $\mathbf{r}_i(t)$, and $\displaystyle \frac{\partial^2 \mathbf{\mathbf{r}_i}}{\partial t^2}$, for $i=a,b,c,d$ (see Appendix \ref{appendix_trunc_error}).

In order to estimate the error committed on the exact volumetric increment approximated at the $n^{th}$ time sample, these errors have to be summed, and yields :  
\begin{equation}
\left\{
\begin{array}{l}
\epsilon^T_m(t_{0})=0, \\
\\
\epsilon^T_m(t_{n}) = \displaystyle \sum_{k=1}^n\epsilon^T_{m,h}(t_{k}) = \left( \displaystyle \sum_{k=1}^n\mathcal{E}_m^T(t_{k-1}) \right)\tau^2 + n\mathcal{O}(\tau^3) \mbox{, for }1 \leq n \leq 2N+1.
\end{array}
\right.
\label{eq_truncation_error}
\end{equation}
Hence, for $1 \leq n \leq 2N+1$ :
\begin{equation}
|\epsilon^T_m(t_{n})| \leq  n\left( \max_{1 \leq n \leq 2N+1} |\mathcal{E}^T_m(t_{n-1})| \right) \tau^2 + n\mathcal{O}(\tau^3) = n \mathcal{O}(\tau^2)
\end{equation}
Thus for $n=N_{ts}$ :
\begin{equation}
|\epsilon^T_m(t_{N_{ts}})| \leq N_{ts} \mathcal{O} \left( \left( \frac{T}{N_{ts}} \right) ^2 \right) = \cal{O}(\tau)
\end{equation}
The order of the error to approximate the exact volume of the volumetric increment may decrease over a period from 2 to 1 for the final value. Thus for any face $m$, the order of the truncation error $\epsilon^T_m(t_{N_{ts}})$ done to compute the zeroth Fourier coefficient of any integrated face mesh velocity $\hat{G}_{m,0}$ is one. 
\end{proof}

Recalling that the zeroth Fourier coefficient is then used to extract the periodic part of any volumetric increment see theorem \ref{th_new_method}, the error committed on the rest of the Fourier coefficients of the integrated face mesh velocities is given by the following lemma \ref{lemma_Gmk},
\begin{lemma}
In the context of theorem \ref{th_new_method}, and under the definitions \ref{def_vol_hex} and \ref{def_vol_inc}, for any face $m$ the temporal-order of accuracy of the Fourier coefficients $\hat{G}_{m,k}$, for $-N \leq k \leq N$ with $k \neq 0$, is limited to between one and two.
\label{lemma_Gmk}
\end{lemma}

\begin{proof}
For $1\leq n \leq N_{ts}$, the periodic part of the volumetric increment can be further expanded as :
$$
\begin{array}{lcl}
p_{m}(t_n) & = & \Omega_m(t_n)-\dfrac{\Omega_m(t_{N_{ts}})}{T}t_n \\
\\
& = & \displaystyle \sum_{k=1}^n \Omega_{m,h}(t_k) + \epsilon^T_m(t_{n}) - \left( \dfrac{\displaystyle \sum_{k=1}^{N_{ts}} \Omega_{m,h}(t_k) +\epsilon^T_m(t_{N_{ts}}) }{T} \left(\dfrac{nT}{N_{ts}} \right)\right)\\
\\
\end{array}
$$
From equation (\ref{eq_truncation_error}), we have :
$$
\begin{array}{lcl}
p_{m}(t_n) & = & \displaystyle \sum_{k=1}^n \Omega_{m,h}(t_k)  + \left( \sum_{k=1}^n \mathcal{E}_m^T(t_{k-1})\right)\tau^2 + n \mathcal{O}(\tau^3)
\\
& & - \displaystyle \left( \dfrac{ \displaystyle \sum_{k=1}^{N_{ts}} \Omega_{m,h}(t_k) +\left( \displaystyle \sum_{k=1}^{{N_{ts}}} \mathcal{E}^T_m(t_{k-1})\right)\tau^2 + ({N_{ts}}) \mathcal{O}(\tau^3) }{T} \left(\dfrac{nT}{N_{ts}} \right)\right)\\
\\
& = & \displaystyle \left( \sum_{k=1}^n \Omega_{m,h}(t_k)  -  \dfrac{n}{N_{ts}} \sum_{k=1}^{N_{ts}} \Omega_{m,h}(t_k) \right)
+ \displaystyle \left[ \sum_{k=1}^n \mathcal{E}^T_m(t_{k-1}) - \dfrac{n}{N_{ts}}\sum_{k=1}^{N_{ts}} \mathcal{E}^T_m(t_{k-1}) \right]\tau^2 + n\mathcal{O}(\tau^3). \\
\\
\end{array}
$$
Then the truncation error on the periodic part of any volumetric increment $p_{m}(t)$ is given for $1\leq n \leq N_{ts}$ by :
\begin{equation}
\epsilon^T_{p_{m}}(t_n) = \left[  \sum_{k=1}^n \mathcal{E}^T_m(t_{k-1}) - \dfrac{n}{N_{ts}}\sum_{k=1}^{N_{ts}} \mathcal{E}^T_m(t_{k-1}) \right]\tau^2 + n\mathcal{O}(\tau^3). 
\label{eq_truncation_error_pm}
\end{equation}

Since the bracketed term in equation (\ref{eq_truncation_error_pm}) is dependent of $n$, the order of accuracy for any $n$ is still unclear. To refine the determination of the order of accuracy during the computation of $p_{m}$, the approximation of an integral using the Riemann sum is exploited. 

For any $T$-periodic function $f$ at least three times continuous ($f\in \mathcal{C}^3([0;T])$, we have the following asymptotic development (\ref{eq_development_asymptotic}) where $f'=\frac{\partial f}{\partial t}$ :
\begin{equation}
\left\{
\begin{array}{lcl}
R_{N_{ts}} & = & \displaystyle \frac{T}{N_{ts}}\sum_{k=0}^{N_{ts}-1} f\left( t_k \right), \\
\\
R_{N_{ts}} & = & \displaystyle \int_0^T f(t)dt - \frac{T}{2N_{ts}}(f(T)-f(0))+\frac{T^2}{12(N_{ts})^2}(f'(T)-f'(0))+\mathcal{O} \left( \left(\frac{T}{N_{ts}}\right)^3 \right). \\
\end{array}
\right.
\label{eq_development_asymptotic}
\end{equation}
Then applying this result to the truncation error $\epsilon^T_{p_m}(t_n)$, for $1\leq n \leq N_{ts}$  :
\begin{equation}
\begin{array}{lcl}
\epsilon^T_{p_{m}}(t_n) 
& = & \displaystyle \left[  \sum_{k=0}^{n-1} \mathcal{E}^T_m(t_{k}) - \dfrac{n}{T} \dfrac{T}{N_{ts}}\sum_{k=0}^{N_{ts-1}} \mathcal{E}^T_m(t_{k}) \right]\tau^2 + n\mathcal{O}(\tau^3) \\
\\
& = & \displaystyle \left[  \sum_{k=0}^{n-1} \mathcal{E}^T_m (t_{k}) - \dfrac{n}{T} \left( \int_0^T \mathcal{E}^T_m (t)dt - \frac{\tau}{2}(\mathcal{E}^T_m (T)-\mathcal{E}^T_m (0))+\frac{\tau^2}{12}({\mathcal{E}^T_m}' (T)-{\mathcal{E}^T_m}' (0))+\mathcal{O}(\tau^3) \right) \right]\tau^2 \\
\\
& + & n\mathcal{O}(\tau^3) \\
\\
& = & \displaystyle \left[  \sum_{k=0}^{n-1} \mathcal{E}^T_m(t_{k}) - n \left( \langle \mathcal{E}^T_m \rangle_T - \frac{\tau}{2T}(\mathcal{E}^T_m(T)-\mathcal{E}^T_m(0))+\frac{\tau^2}{12T} ({\mathcal{E}^T_m}'(T)-{\mathcal{E}^T_m}'(0)) \right) \right]\tau^2 \\
\\
& + & n\mathcal{O}(\tau^3) \\
\\
& = & \displaystyle \left[  \sum_{k=0}^{n-1} \mathcal{E}^T_m(t_{k}) - n \langle \mathcal{E}^T_m \rangle_T  \right]\tau^2 + n\mathcal{O}(\tau^3) \\
\\
& = & \displaystyle \left[  \sum_{k=0}^{n-1} \left( \mathcal{E}^T_m(t_{k}) -  \langle \mathcal{E}^T_m \rangle_T \right) \right]\tau^2 + n\mathcal{O}(\tau^3), \\
\\
\end{array}
\label{eq_calculus_truncation_error_pm_1}
\end{equation}
where $\langle . \rangle_T$ represents the mean of a function on the segment $[0;T]$. Taking advantage of the fact that the function $ \Delta\mathcal{E}^T_m = \mathcal{E}^T_m -  \langle \mathcal{E}^T_m \rangle_T$ is $T$-periodic with zero mean value and exploiting a second time the expression (\ref{eq_development_asymptotic}),  yields for $n=N_{ts}$ :
\begin{equation}
\begin{array}{lcl}
\displaystyle \left\{  \sum_{k=0}^{N_{ts}-1} \left[ \left( \mathcal{E}^T_m(t_{k}) -  \langle \mathcal{E}^T_m \rangle_T \right) \right] \right\}\tau & = & \underset{=0}{\underbrace{\int_0^T \left( \Delta\mathcal{E}^T_m \right)dt}} - \frac{\tau}{2}(\Delta\mathcal{E}^T_m(T)-\Delta\mathcal{E}^T_m(0)) +\mathcal{O}(\tau^2) \\
\\
& = & \displaystyle -\frac{\tau}{2}(\mathcal{E}^T_m(T)-\mathcal{E}^T_m(0)) +\mathcal{O}(\tau^2). \\
\\
\end{array}
\label{eq_calculus_truncation_error_pm_2}
\end{equation}
Substituting back the expression (\ref{eq_calculus_truncation_error_pm_2}) into the final equation in (\ref{eq_calculus_truncation_error_pm_1}) for $n=N_{ts}$, leads to :
\begin{equation}
\epsilon^T_{p_{m}}(t_{N_{ts}}) = -\tau^2(\mathcal{E}^T_m(T)-\mathcal{E}^T_m(0)) + N_{ts}\mathcal{O}(\tau^3) = \mathcal{O}(\tau^2). 
\label{eq_calculus_truncation_error_pm_Nts}
\end{equation}

In summary the truncation error committed on the periodic part of any volumetric increment follows the equation :
\begin{equation}
\left\{
\begin{array}{l}
\epsilon^T_{p_{m}}(t_0)=0, \\
\\
\epsilon^T_{p_{m}}(t_n) = \displaystyle \left\{  \sum_{k=0}^{n-1} \left[ \left( \mathcal{E}^T_m(t_{k}) -  \langle \mathcal{E}^T_m \rangle_T \right) \right] \right\}\tau^2 + n\mathcal{O}(\tau^3)\mbox{, for }1\leq n \leq 2N. \\
\\
\epsilon^T_{p_{m}}(t_{N_{ts}}) = \mathcal{O}(\tau^2).
\end{array}
\right.
\label{eq_truncation_error_pm_order}
\end{equation}

In general, the order of the truncation error on the approximation of the periodic part of the volumetric increment used as input for the NLFD method is of order between one and two. Analytically, we observe that this order is determined by the sum $\displaystyle \sum_{k=0}^{n-1} \left[ \left( \mathcal{E}^T_m(t_{k}) -  \langle \mathcal{E}^T_m \rangle_T \right) \right]$, which is bounded for $1 \leq n \leq N_{ts}$ by $ \left\{ N_{ts} \displaystyle \max_{1 \leq k \leq N_{ts}} |\left( \mathcal{E}^T_m(t_{k}) -  \langle \mathcal{E}^T_m \rangle_T \right)| \right\}$. This upper bound ensures that in the worst case, the order of accuracy is 1. However asymptotically it is reasonable to assume that for small and high values of $n$, the term $\displaystyle\sum_{k=0}^{n-1} \left[ \left( \mathcal{E}^T_m(t_{k}) -  \langle \mathcal{E}^T_m \rangle_T \right) \right]$ is small enough to consider that the truncation error is of order 2 whereas for $n$ in the middle of the range $[1;N_{ts}]$, the order is greater than 1 but lesser than 2.
\end{proof}

Assuming that the spectral convergence of the Fourier transform is reached and taking advantage of its bijectivity, the truncation error on the Fourier coefficients $\hat{G}_{m,k}$ and finally on the integrated face mesh velocities is of order between 1 and 2. Therefore the accuracy of the procedure is given by the following corollary \ref{th_error} :
\begin{corollary}
In the context of theorem \ref{th_new_method}, and under the definitions \ref{def_vol_hex} and \ref{def_vol_inc}, for any face $m$ the temporal-order of accuracy of the integrated face mesh velocities is limited to between one and two.
\label{th_error}
\end{corollary}

Thus it is important to note that even if the method described in section \ref{subsec_DerivGCL} enforced the Geometric Conservation Law, the integrated face mesh velocities are determined within an accuracy of order 1 to 2. This is a disadvantage since the benefit of the spectral convergence of the NLFD method. 

\subsection{Alternative approach based on the exact integrated face mesh velocities}

\subsubsection{Trilinear mapping}
\label{subsec_tri_mapping}
The computation of the metrics of a grid is often easier in a Cartesian grid, for this reason a mapping between the curvilinear physical space and a Cartesian computational space can be performed. In this work, a trilinear mapping is already used to compute any hexahedron volume \citep{Dukowicz1988}, but it can also be used to compute the time derivative of any hexahedron, its surface vectors and the exact integrated face mesh velocities as long as the position and velocity vectors of the vertices are known. This section develops the derivation of these expressions.

\subsubsection*{Notation :}
\begin{equation}
\begin{array}{ccl}
\cal{T} & \leftrightarrow & \text{Trilinear mapping} \\
p & \leftrightarrow & \text{physical space : }(x,y,z) \\
r & \leftrightarrow & \text{reference space : }(\xi,\eta,\zeta) \\
m & \leftrightarrow & \text{any faces of an hexahedron} \\
\mathbf{n} & \leftrightarrow & \text{normal vector} \\
\hat{\mathbf{n}} & \leftrightarrow & \text{unit normal vector}
\end{array}
\label{eq_Notations_mapping}
\end{equation}

\subsubsection*{Derivation :}
The mapping $\cal{T}$ from the physical to the computational space is introduced :
\begin{equation}
 \cal{T} =
 \left\{
 \begin{array}{c}
 \cal(D_C) \rightarrow \cal(D_P) \\
 (\xi,\eta,\zeta) \rightarrow (x,y,z) = \cal(T(\xi,\eta,\zeta)),
 \end{array}
 \right.
 \label{eq_mapping}
\end{equation}
where $\cal(D_C)$ is the computational domain and $\cal(D_P)$ is the physical domain. The application is defined by considering a reference cube in the computational space which enables the mapping of any general hexahedron in the physical space. A necessary and sufficient condition to ensure the invertibility of the mapping is the strict positivity of the Jacobian for any point of the element. However no simple relations exist in order to verify the positivity of the Jacobian in 3D \cite{Lopez2017,Knabner2003}.

In this work, the position vector $\mathbf{r}_p=(x,y,z)$ in the physical space is mapped through $\mathbf{r}_r=(x(\xi,\eta,\zeta),y(\xi,\eta,\zeta),z(\xi,\eta,\zeta))$ based on the location vectors in the physical space $\mathbf{r_i}_{,p}=[x_i,y_i,z_i]$ $i=1,...,8$ of the vertices with the following convention derived from Figure \ref{fig_tri_mapping} :
\begin{equation}
\begin{array}{c}
\mathbf{r}_r=(1-\xi)(1-\eta)(1-\zeta)\mathbf{r_1}_{,p}+\xi(1-\eta)(1-\zeta)\mathbf{r_2}_{,p}+\xi\eta(1-\zeta)\mathbf{r_3}_{,p}+(1-\xi)\eta(1-\zeta)\mathbf{r_4}_{,p}\\
+(1-\xi)(1-\eta)\zeta \mathbf{r_5}_{,p}+\xi(1-\eta)\zeta \mathbf{r_6}_{,p}+\xi\eta\zeta \mathbf{r_7}_{,p}+(1-\xi)\eta\zeta \mathbf{r_8}_{,p}, \\
\end{array}
\label{eq_position_mapping}
\end{equation}
where $0 \leq \xi,\eta,\zeta \leq 1$. 

The velocity vector $\mathbf{v}_p=(v_x,v_y,v_z)$ in the physical domain is mapped in the same way $\mathbf{v}_r=(v_x(\xi,\eta,\zeta),v_y(\xi,\eta,\zeta),v_z(\xi,\eta,\zeta))$ based on the velocity vectors of the vertices $\mathbf{v_i}_{,p}=[v_{xi},v_{yi},v_{zi}]$ $i=1,...,8$ :
\begin{equation}
\begin{array}{c}
\mathbf{v}_r=(1-\xi)(1-\eta)(1-\zeta)\mathbf{v_1}_{,p}+\xi(1-\eta)(1-\zeta)\mathbf{v_2}_{,p}+\xi\eta(1-\zeta)\mathbf{v_3}_{,p}+(1-\xi)\eta(1-\zeta)\mathbf{v_4}_{,p}\\
+(1-\xi)(1-\eta)\zeta \mathbf{v_5}_{,p}+\xi(1-\eta)\zeta \mathbf{v_6}_{,p}+\xi\eta\zeta \mathbf{v_7}_{,p}+(1-\xi)\eta\zeta \mathbf{v_8}_{,p}. \\
\end{array}
\label{eq_velo_mapping}
\end{equation}

For any face $m$ of a cell, the normal vector is given by one of the following expressions :
\begin{equation}
\begin{array}{cc}
\mathbf{n}_{r,\zeta =0}= -\left(\dfrac{\partial \mathbf{r}_r}{\partial \xi}\right)\text{x}\left(\dfrac{\partial \mathbf{r}_r}{\partial \eta} \right),&
\mathbf{n}_{r,\zeta =1}= +\left(\dfrac{\partial \mathbf{r}_r}{\partial \xi}\right)\text{x}\left(\dfrac{\partial \mathbf{r}_r}{\partial \eta} \right),
\\
\\
\mathbf{n}_{r,\xi =0}=- \left( \dfrac{\partial \mathbf{r}_r}{\partial \eta} \right) \text{x}\left(\dfrac{\partial \mathbf{r}_r}{\partial \zeta}\right),& 
\mathbf{n}_{r,\xi =1}=+ \left( \dfrac{\partial \mathbf{r}_r}{\partial \eta} \right) \text{x}\left(\dfrac{\partial \mathbf{r}_r}{\partial \zeta}\right),
\\
\\
\mathbf{n}_{r,\eta =0}=- \left(\dfrac{\partial \mathbf{r}_r}{\partial \zeta}\right)\text{x}\left(\dfrac{\partial \mathbf{r}_r}{\partial \xi} \right), &
\mathbf{n}_{r,\eta =1}=+ \left(\dfrac{\partial \mathbf{r}_r}{\partial \zeta}\right)\text{x}\left(\dfrac{\partial \mathbf{r}_r}{\partial \xi} \right),
\\
\end{array}
\label{eq_normal_mapping}
\end{equation}
where the sign is determined in order to have the normal pointing outward of the cell volume.
The Jacobian matrix $J(\xi,\eta,\zeta)$ is expressed as :
\begin{equation}
J(\xi,\eta,\zeta)=
\left(
\begin{array}{ccc}
\dfrac{\partial \mathbf{r}_r}{\partial \xi} & \dfrac{\partial \mathbf{r}_r}{\partial \eta} & \dfrac{\partial \mathbf{r}_r}{\partial \zeta}
\end{array}
\right)
\end{equation}
and its determinant can be calculated with one of the following expressions :
\begin{equation}
|J| = \left(\dfrac{\partial \mathbf{r}_r}{\partial \xi}\right)\cdot \left[\left(\dfrac{\partial \mathbf{r}_r}{\partial \eta}\right) \mbox{x} \left(\dfrac{\partial \mathbf{r}_r}{\partial \zeta}\right)\right]
=\left(\dfrac{\partial \mathbf{r}_r}{\partial \eta}\right)\cdot \left[\left(\dfrac{\partial \mathbf{r}_r}{\partial \zeta}\right) \mbox{x} \left(\dfrac{\partial \mathbf{r}_r}{\partial \xi}\right)\right]
=\left(\dfrac{\partial \mathbf{r}_r}{\partial \zeta}\right)\cdot \left[\left(\dfrac{\partial \mathbf{r}_r}{\partial \xi}\right) \mbox{x} \left(\dfrac{\partial \mathbf{r}_r}{\partial \eta}\right)\right].
\end{equation}
Once the position vector, velocity vector, normal vectors and Jacobian are known, these quantities are used to compute the integrals of the volume and mesh velocity through a change of variables.

\subsubsection*{Volume integral}
Trough the application of the divergence theorem, the volume of the hexahedron can be evaluated as such,
$$
V_{p} = \int_{\Omega_{p}} dV_{p} = \oiint_{\partial \Omega_{p}} \frac{1}{3} \mathbf{r}_{p}\cdot \mathbf{dS}_{p}
= \frac{1}{3} \oiint_{\partial \Omega_{p}} (\mathbf{r}_{p}\cdot \mathbf{\hat{n}}_{p} )dS_{p}
$$
We can then write the integral for the computational domain through the trilinear mapping to acquire,
$$
\begin{array}{cl}
V_{p} = \displaystyle \frac{1}{3} \oiint_{\partial \Omega_{p}} (\mathbf{r}_{p}\cdot \mathbf{\hat{n}}_{p} )dS_{p}
&\displaystyle = \dfrac{1}{3} \oiint_{\partial \Omega_{r}} (\mathbf{r}_{r}\cdot \mathbf{\hat{n}}_{r} )|J_{r}| dS_{r} \\
&\\
&\displaystyle = \dfrac{1}{3} \sum_{m=1}^{N_f} \iint_{\partial \Omega_{r,m}} (\mathbf{r}_{r,m}\cdot \mathbf{\hat{n}}_{r,m} )|J_{r,m}| dS_{r,m},
\end{array}
$$
where $dS_{c,r,m}$ is either $d\xi d\eta$, $d\eta d\zeta$ or $d\zeta d\xi$ and the integral boundaries are $[0\mbox{ }1]^2$. \textbf{N.B. :} On any face of the hexahedron only one of the variables in the reference space $\xi$, $\eta$ or $\zeta$ has a fixed value. Thus the quantity $(\mathbf{r}_{c,r,m}.\mathbf{\hat{n}}_{r,m} )|J_{r,m}|$ is still a function of two variables which has to be integrated over the face.

For each face, the computation of the integral over the surface under this form is not straightforward (the difficulty comes from the unit normal vector) and needs to be simplified \textit{a priori}. This is done by exploiting the relation (\ref{eq_simplification_rational}), for the derivation of this expression see Appendix B in \cite{Zwanenburg2016}  :
\begin{equation}
\mathbf{\hat{n}}_{r,m} |J_{r,m}|=C_{r,m}\mathbf{\hat{N}}_{r,m},
\label{eq_simplification_rational}
\end{equation}
where $C_{r,m}=C(\xi,\eta,\zeta)$ is the cofactor matrix of the Jacobian matrix $J$ for the trilinear mapping and $\mathbf{\hat{N}}_{r,m}$ is the constant unit normal vector to the corresponding face in the reference space :
\begin{equation}
\begin{array}{cc}
\mathbf{\hat{N}}_{\zeta =0}= [ \begin{array}{ccc} 0 & 0 & -1\end{array} ]^T, &
\mathbf{\hat{N}}_{\zeta =1}= [ \begin{array}{ccc} 0 & 0 & +1\end{array} ]^T,
\\
\\
\mathbf{\hat{N}}_{\eta =0}= [ \begin{array}{ccc} 0 & -1 & 0\end{array} ]^T, &
\mathbf{\hat{N}}_{\eta =1}= [ \begin{array}{ccc} 0 & +1 & 0\end{array} ]^T,
\\
\\
\mathbf{\hat{N}}_{\xi =0}= [ \begin{array}{ccc} -1 & 0 & 0\end{array} ]^T, &
\mathbf{\hat{N}}_{\xi =1}= [ \begin{array}{ccc} +1 & 0 & 0\end{array} ]^T.
\\
\end{array}
\label{eq_unormal_mapping}
\end{equation}
Once the equation (\ref{eq_simplification_rational}) is substituted into the integrals over the surfaces, an explicit expression of the volume as a function of $\mathbf{r}_i$, $i=1,...,8$ is obtained :
\begin{equation}
\begin{array}{c}
( V_{c,p} )_{\cal{T}}=(V_{4321}+V_{5678}+V_{3487}+V_{1256}+V_{4158}+V_{2376})_{\cal{T}},
\end{array}
\label{eq_any_volume}
\end{equation}
where for any set $i,j,k,l$, the volumetric contribution of the face $S_{ijkl}$ is given by (\ref{eq_surface_increment}) :
\begin{equation}
(V_{ijkl})_{\cal{T}}
= \frac{1}{3} \iint_{\partial \Omega_{r,ijkl}} \mathbf{r}_{r,ijkl}\cdot (C_{r,ijkl}\mathbf{\hat{N}}_{r,ijkl}) dS_{r,ijkl} 
=\frac{1}{12}(\mathbf{r}_j+\mathbf{r}_k)\cdot ((\mathbf{r}_i+\mathbf{r}_j)\text{x}(\mathbf{r}_i+\mathbf{r}_l)).
\label{eq_surface_increment}
\end{equation}

\subsubsection*{Time derivative of the volumetric integral}
The temporal derivative of the volumetric integral can be expressed as,
\begin{equation}
\frac{\partial V_{p}}{\partial t} = \frac{\partial}{\partial t} \int_{\Omega_{p}} dV_{p}.
\end{equation}
By substituting the results of the previous section, primarily equations (\ref{eq_any_volume}) and (\ref{eq_surface_increment}), an explicit expression of the temporal derivative of the volume as a function of $\mathbf{r}_i$ and $\mathbf{v}_i$, $i=1,...,8$  is derived :
\begin{equation}
\left( \frac{\partial V_{p}}{\partial t} \right)_{\cal{T}} = \left(\frac{\partial V_{4321}}{\partial t}
+\frac{\partial V_{5678}}{\partial t}
+\frac{\partial V_{3487}}{\partial t}
+\frac{\partial V_{1256}}{\partial t}
+\frac{\partial V_{4158}}{\partial t}
+\frac{\partial V_{2376}}{\partial t} \right)_{\cal{T}},
\label{eq_dvol}
\end{equation}
where for any set $i,j,k,l$, the time derivative volumetric contribution of the face $S_{ijkl}$ is given by (\ref{eq_dsurface_increment}) :
\begin{equation}
\begin{array}{r}
\left( \dfrac{\partial V_{ijkl}}{\partial t} \right)_{\cal{T}}=
 \dfrac{1}{12}(\mathbf{v}_j+\mathbf{v}_k) \cdot ((\mathbf{r}_i+\mathbf{r}_j)\text{x}(\mathbf{r}_i+\mathbf{r}_l)) \\
 \\
+ \dfrac{1}{12}(\mathbf{r}_j+\mathbf{r}_k) \cdot ((\mathbf{v}_i+\mathbf{v}_j)\text{x}(\mathbf{r}_i+\mathbf{r}_l)) \\
\\
+ \dfrac{1}{12}(\mathbf{r}_j+\mathbf{r}_k) \cdot ((\mathbf{r}_i+\mathbf{r}_j)\text{x}(\mathbf{v}_i+\mathbf{v}_l)).
\end{array}
\label{eq_dsurface_increment}
\end{equation}

\subsubsection*{Integrated face mesh velocities}
The integral of the face mesh velocity in the physical domain for a face $m$ is given by,
$$
G_{p,m} =\iint_{\partial\Omega_{c,p,m}} ({\mathbf{v}}\cdot\mathbf{\hat{n}_{p,m}} ) dS_{c,p,m}.
$$
By introducing the trilinear mapping, we can express the integrated face mesh velocities as,
$$
G_{p,m} =\iint_{\partial \Omega_{r,m}} \mathbf{v}_{r,m} \cdot (\mathbf{\hat{n}}_{r,m} |J_{r,m}|) dS_{r,m}
=\iint_{\partial \Omega_{r,m}} \mathbf{v}_{r,m} \cdot (C_{r,m}\mathbf{\hat{N}}_{r,m}) dS_{r,m}.
$$
Once the integration is performed, the explicit expressions of the integrated face mesh velocities are obtained as a function of $\mathbf{r}_i$ and $\mathbf{v}_i$, $i=1,...,8$. For a face with the set $(i,j,k,l)\in \left\{4321;5678;3487;1256;4158;2376 \right\}$ :
\begin{equation}
\left\{
\begin{array}{rl}
\mathbf{v}_t= & (\mathbf{v}_i+\mathbf{v}_j+\mathbf{v}_k+\mathbf{v}_l) \\
&\\
\mathbf{S}_{i,j,k,l}= & \left( (\mathbf{r}_i \mbox{ x } \mathbf{r}_j) + (\mathbf{r}_j \mbox{ x } \mathbf{r}_k) + (\mathbf{r}_k \mbox{ x } \mathbf{r}_l) + (\mathbf{r}_l \mbox{ x } \mathbf{r}_i) \right)\\
&\\
\mathbf{S}_{a,b,c}= & \left( (\mathbf{r}_a\mbox{ x } \mathbf{r}_b) + (\mathbf{r}_b \mbox{ x } \mathbf{r}_c) + (\mathbf{r}_c \mbox{ x } \mathbf{r}_a) \right) \text{ for any set } a,b,c\\
&\\
(G_{i,j,k,l})_{\cal{T}} =& \dfrac{1}{12} \left( \mathbf{v}_t \cdot \mathbf{S}_{i,j,k,l} + \mathbf{v}_j \cdot \mathbf{S}_{i,j,k} + \mathbf{v}_k \cdot \mathbf{S}_{j,k,l} 
                                                       + \mathbf{v}_l \cdot \mathbf{S}_{k,l,i} + \mathbf{v}_i \cdot\mathbf{S}_{l,i,j} \right)
\end{array}
\right.
\label{eq_tri_mapping_IFMV}
\end{equation}

It was checked that with these expressions for the IFMV and the time derivative of the volume as functions of velocity and position vectors of the vertices, the semi-discrete equation of the GCL (\ref{eq_semi_discrete_GCL_2}) is analytically retrieved. In other words, the sum of equation (\ref{eq_tri_mapping_IFMV}) applied to the 6 sets $\left\{4321;5678;3487;1256;4158;2376 \right\}$ is equal to expression (\ref{eq_dvol}).

\subsubsection{Derivation of the GCL in the NLFD framework}
The methods presented in sections \ref{subsec_Tardiff_approach} and \ref{subsec_modified_approach} to enforce the GCL are based on equation (\ref{eq_GCL_split}), and the integrated face mesh velocities are deduced from the calculation of the volumetric increments as input. The approach presented in this section using the trilinear mapping is quite different because no volumetric increments are computed, the integrated face mesh velocities are directly evaluated in physical time using equation (\ref{eq_tri_mapping_IFMV}). In addition the cell volumes are computed using equation (\ref{eq_tri_mapping_volume}). Hence in the GCL equation as established in  (\ref{eq_semi_discrete_GCL_1}), both $\Omega$ and $G = \sum_{m=1}^{N_f} G_m$ are exactly calculated, the only degree of freedom remaining to enforce the equation is the discretization of the temporal derivative operator $\displaystyle \left( \frac{\partial}{\partial t} \right)$. In the NLFD framework, this operator is discretized in the Fourier domain and is a function of the number of harmonics $N$ employed in the temporal discretization. Therefore the GCL equation will be satisfied if and only if the time derivative of the cell volume expressed in Fourier space converge to the Fourier time differentiation applied to the cell volume, 
\begin{equation}
\text{DFT} \left\{ \left( \frac{\partial \Omega}{\partial t} \right)_{\cal{T}} \right\} = \left( \frac{\partial}{\partial t} \right)_{Fourier} \left( \text{DFT}\left\{ (\Omega_{\cal{T}}) \right\} \right),
\label{eq_conv_volume_DFT}
\end{equation}
where $\cal{T}$ refers to the trilinear mapping. Hence, this method will not enforce the GCL for any number of time steps contrary to the method presented in section \ref{subsec_modified_approach}, but for a sufficient number of harmonics ensuring the convergence of the equations (\ref{eq_conv_volume_DFT}). Since this approach is based on the exact integrated face mesh velocities and ensures the GCL with a spectral rate of convergence depending on the mesh motion, it provides a good alternative to the method exploiting the volumetric increments with an order of accuracy comprised between one and two. In the section \ref{sec_num}, we will present the numerical results of these different methods for several test cases.

\subsection{Extension of the results to Time-Spectral method}
In this section, we extent the previous results to Time-Spectral (TS) method presented by Gopinath and Jameson \cite{Gopinath2005,Gopinath2006}. Compared to the NLFD method which solves the governing equations in the frequency domain, the Time-Spectral method solves the governing equations in the time domain but exploits the features of a spectral approach.

\subsubsection{Time-Spectral method}
Assuming a periodic flow and a periodic deformation of the mesh, we recall the temporal discretization of the modified state vector $\bar{\mathbf{w}}=\Omega\mathbf{w}$ equations (\ref{eq_DFT_w}) and (\ref{eq_Fourier_coeff_w}),
$$
\bar{\mathbf{w}}(t)=\sum_{k=-N}^N \mathbf{\hat{w}}_k \text{e}^{i(2\pi k/T)t},
$$
with :
$$
\mathbf{\hat{w}}_k = \frac{1}{2N+1}\sum_{n=0}^{2N} \Omega(t_n)\mathbf{w}(t_n)\text{e}^{-i(2\pi k/T)t_n},
$$
where $T$ is the time period, $N$ is the number of modes considered in the DFT and $t_n$ the equally spaced time instances given by,
$$
t_n = \frac{n}{2N+1}T \mbox{, for }n=0,..,2N.
$$
In Fourier space, the time discretization operator leads to,
\begin{equation}
\frac{\partial \mathbf{\bar{w}}}{\partial t}(t) = \sum_{k=-N}^N \frac{i2 \pi k}{T} \mathbf{\hat{w}}_k \text{e}^{i(2\pi k/T)t},
\end{equation}
\begin{equation}
\Leftrightarrow \frac{\partial \mathbf{\bar{w}}}{\partial t}(t) = \frac{2\pi}{T} \sum_{k=-N}^N ik \left( \frac{1}{2N+1}\sum_{K=0}^{2N} \Omega(t_K)\mathbf{w}(t_K)\text{e}^{-i(2\pi k/T)t_K} \right) \text{e}^{i(2\pi k/T)t}.
\end{equation}
By evaluating this expression for each time instance $t_n$, we have for $n=0,...,2N$,
\begin{equation}
\frac{\partial \mathbf{\bar{w}}}{\partial t}(t_n) = \frac{2\pi}{T} \sum_{k=-N}^N ik \left( \frac{1}{2N+1}\sum_{K=0}^{2N} \Omega(t_K)\mathbf{w}(t_K)\text{e}^{-i(2\pi k/T)t_K} \right) \text{e}^{i(2\pi k/T)t_n},
\end{equation}
\begin{equation}
\Leftrightarrow \frac{\partial \mathbf{\bar{w}}}{\partial t}(t_n) = \sum_{K=0}^{2N} \left[ \Omega(t_K)\mathbf{w}(t_K) \left( \frac{2\pi}{T} \frac{1}{2N+1} \sum_{k=-N}^N ik   \text{e}^{i(2\pi k)(n-K)/(2N+1)} \right)  \right].
\end{equation}
We introduce, the coefficients $d_{n,K}$, defined for $n=0,...,2N$ by,
\begin{equation}
d_{n,K} = \frac{2\pi}{T} \frac{1}{2N+1} \sum_{k=-N}^N ik   \text{e}^{i(2\pi k)(n-K)/(2N+1)},
\end{equation}
the compact form of the coefficients for an odd number of time steps is written as follows (for the derivation see Reference \cite{Gopinath2006}), 
\begin{equation}
d_{n,K} = 
\left\{
\begin{array}{lcl}
\displaystyle \frac{2\pi}{T} \frac{1}{2}(-1)^{n-K} \csc \left(\frac{\pi(n-K)}{2N+1}\right), &\mbox{ if }& K \neq n \\
0 , &\mbox{ if }& K = n,
\end{array}
\right.
\end{equation}
and,
\begin{equation}
\frac{\partial \mathbf{\bar{w}}}{\partial t}(t_n) = \sum_{K=0}^{2N} d_{n,K} \bar{\mathbf{w}}(t_K).
\end{equation}
The temporal-derivation operator appears as the multiplication of a matrix $\mathcal{D}=(d_{n,K})_{0\leq n,K \leq 2N}$ with each vector $(\mathbf{\bar{w}}_i(t_K))_{0\leq K \leq 2N}$, for $i=1,...,5$ where the index $i$ refers to the component of the modified state vector in the governing equations. In addition, this matrix is skew-symmetric, independent of any state variables and completely determined by the number of harmonics used in the DFT and the temporal period. Then a pseudo-time $t^{*}$ is introduced and the equations are solved in the time domain through,
\begin{equation}
\frac{\partial \mathbf{\bar{w}}}{\partial t^{*}}(t_n) + \frac{\partial \mathbf{\bar{w}}}{\partial t}(t_n) + \mathbf{R(w(}t_n\mathbf{))} = 0, \mbox{ for }n=0,...,2N.
\end{equation}

\subsubsection{Derivation and enforcement of the GCL}
Recall that in order to obtain a consistent solution method, the GCL must be discretized using the same numerical scheme employed to discretize the governing equations. In the case of Time-Spectral method, it leads to the following theorem \ref{th_new_method_TS},

\begin{theorem}
Let $\Omega$ be a discretized control volume, enclosed by $N_f$ faces, and subjected to a periodic motion of its vertices. Then given the knowledge of the exact volumetric {increments} $\Omega_m$ for $m=1,...,N_f$, a sufficient condition to ensure the satisfaction of GCL in the TS framework is the computation of the integrated face mesh velocities through the following relations,
\begin{equation}
\mathbf{G}_{m} =  (\mathcal{D}) \boldsymbol{p}_{m} + \left\langle G_{m} \right\rangle_T (\mathcal{I_{N}}),
\label{eq_GCL_TS_Gmk}
\end{equation} 
where for all $m$, $\left\langle G_{m} \right\rangle_T$ are the temporal mean values of the integrated face mesh velocities, $\mathcal{I_{N}}$ is the identity matrix of dimension $2N+1$, $\mathbf{G}_{m}=(G_m(t_n))_{0 \leq n \leq 2N}$ and $\boldsymbol{p}_{m}=(p_{m}(t_n))_{0 \leq n \leq 2N}$ are the vectors grouping the time instances of respectively the integrated face mesh velocities and the periodic part of the exact volumetric {increments} given by,
\begin{equation}
p_{m}(t) =  \Omega_{m}(t)-\left( \frac{\Omega_{m}(T)}{T}\right) t,
\label{eq_pm_periodic_simplified_TS}
\end{equation}
and $\mathcal{D}=(d_{n,K})_{0\leq n,K \leq 2N}$ is the matrix representing the temporal derivation operator of the Time-Spectral method, defined by its coefficients $d_{n,K}$ for $0\leq n,K \leq 2N$,
\begin{equation}
d_{n,K} = 
\left\{
\begin{array}{lcl}
\displaystyle \frac{\pi}{T}(-1)^{n-K} \csc \left(\frac{\pi(n-K)}{2N+1}\right), &\mbox{ if }& K \neq n \\
0 , &\mbox{ if }& K = n.
\end{array}
\right.
\label{eq_D_matrix_TS}
\end{equation}
\label{th_new_method_TS}
\end{theorem}

\begin{proof}
Under the assumption that the motion of the vertices is periodic, the temporal rate of change of the algebraic volume swept by each face through time is periodic. Thus the temporal derivative of the volumetric {increments} and the integrated face mesh velocities are periodic, the DFT is applied to equation (\ref{eq_GCL_split}) leading to :
\begin{equation}
G_m(t) = \frac{\partial \Omega_m}{\partial t} = \hat{G}_{m,0} + \sum_{k=-N,k\neq 0}^{N} \hat{G}_{m,k} \text{e}^{i\frac{2\pi}{T}kt},
\label{eq_DFT_Gm_TS}
\end{equation}
where for any face $m$, $\hat{G}_{m,k}$ are the Fourier coefficients of both the derivative of the volumetric {increment} and the integrated face mesh velocity. The mean of a function expandable in Fourier serie is given by its zeroth Fourier coefficient thus,
\begin{equation}
\left\langle G_{m} \right\rangle_T = \hat{G}_{m,0}.
\label{eq_Gm0_mean_Gm_TS}
\end{equation}
From the proof of theorem \ref{th_new_method}, we have the following relations,
\begin{equation}
\Omega_{m}(t) = l_{m}(t)+p_{m}(t),
\label{eq_int_dOmegam_TS}
\end{equation}
\begin{equation}
l_{m}(t) = \hat{G}_{m,0} t = \frac{\Omega_{m}(T)}{T}t,
\label{eq_Omegam_lm_TS}
\end{equation} 
\begin{equation}
p_{m}(t) = \hat{\Omega}_{m,0}+\sum_{k=-N,k \neq 0}^{N} \frac{T}{i2\pi k}\hat{G}_{m,k} e^{i\frac{2\pi}{T}kt},
\label{eq_Omegam_pm_TS}
\end{equation}
\begin{equation}
p_{m}(t) =  \Omega_{m}(t)-\left( \frac{\Omega_{m}(T)}{T}\right) t.
\label{eq_pm_periodic_simplified_proof_TS}
\end{equation}
Then, by exploiting these results and applying the Time-Spectral temporal derivation to the periodic part of the volumetric increments, we can write for each time instance $t_n$, with $n=0,...,2N$,
\begin{equation}
\begin{array}{lcl}
G_m(t_n) & = & \displaystyle \frac{\partial \Omega_m}{\partial t}(t_n) \\
\\
& = & \displaystyle  \frac{\partial (l_{m}+p_{m})}{\partial t}(t_n) \\
\\
& = & \displaystyle \left\langle G_{m} \right\rangle_T + \frac{\partial p_{m}}{\partial t}(t_n) \\
\\
& = & \displaystyle \left\langle G_{m} \right\rangle_T + \sum_{K=0}^{2N} d_{n,K} p_{m}(t_K), \\
\\
\end{array}
\end{equation}
Finally, if we group all the time instances in a vector $\mathbf{G}_m$, we obtain,
\begin{equation}
\mathbf{G}_{m} =  (\mathcal{D}) \boldsymbol{p}_{m} + \left\langle G_{m} \right\rangle_T (\mathcal{I_{N}}),
\label{eq_GCL_TS_proof}
\end{equation} 
where $\mathcal{I_{N}}$ is the identity matrix of dimension $2N+1$ and $\mathcal{D}=(d_{n,K})_{0\leq n,K \leq 2N}$ is the matrix representing the temporal derivation operator of the Time-Spectral method. The condition given by equation (\ref{eq_GCL_TS_proof}) is a criteria to ensure that the GCL are enforced in the Time-Spectral framework.
\end{proof}

\section{Numerical results}
\label{sec_num}
The new approaches to enforce the Geometric Conservation Law developed in section \ref{sec_GCL} are numerically tested in order to validate their procedures. The protocol, test cases and results are presented in the following sections.

\subsection{Protocol}
The physical interpretation of the GCL is that any uniform flow must be preserved by the numerical scheme employed for the flow solver and independently of the mesh movements. This law imposes constraints on the manner to compute some geometrical quantities such as the volume and the integrated face mesh velocities. Thus the first step of our test is to ensure the preservation of uniform flow by computing the relative error between the initially defined uniform state vector $\mathbf{W}_0$ and the computed state vector $\mathbf{W}$ by the flow solver,
\begin{equation}
RelErr= \max_{0\leq n \leq 2N} \left\{ \max_{1\leq n_v \leq N_{cell}} \left( \max_{1\leq j \leq 5} \left|  \frac{W_j(n_v,t_n)-W_{0,j}(n_v,t_n)}{W_{0,j}(n_v,t_n)} \right| \right) \right\}, 
\label{eq_results_rel_error_W}
\end{equation}
where $W_1=\rho$, $W_2=\rho u$, $W_3=\rho v$, $W_4=\rho w$ and $W_5=\rho E$ and $n_v$ is the index pointing to the grid cell with $N_{cell}$ the number of cells in the mesh.

However the verification of uniform flow preservation only guarantees that the GCL are satisfied "by summing over the faces", but not that the computed integrated face mesh velocities are correct. Indeed as long as the sum of the time derivative of the volumetric increments is equal to the time derivative of the cell volume,
\begin{equation}
\frac{\partial \Omega}{\partial t} = \sum_{m=1}^{N_f} \frac{\partial \Omega_m}{\partial t},
\label{eq_dOmega_sum_dOmegam}
\end{equation}
the deduced integrated face mesh velocities from the time derivative of the volumetric increments from equation (\ref{eq_GCL_split}) enforce the GCL after the summation through the faces (see equation (\ref{eq_semi_discrete_GCL_2})) but the integrated face mesh velocities themselves may not converge to the correct value.

Thus in order to verify that the GCL are enforced with a correct evaluation of the integrated face mesh velocities, the values derived from the trilinear mapping equations (\ref{eq_dvol}) and (\ref{eq_tri_mapping_IFMV}) based on the location and velocity vectors of the grid points retrieved from the dynamic mesh deformation, are considered as reference. Therefore for each motion of the mesh and for various number of harmonics $N$, four different implementations of the integrated face mesh velocities are compared :
\begin{enumerate}
\item the IFMV deduced from the linear volumetric increments from Tradiff et al. \cite{Tardif2017} see Figure \ref{fig_linear_Omegam} noted as "NLFD-LVI" ;
\item the IFMV calculated with the new method based on the exact volumetric increments approximated as a sum of hexahedron see Figure \ref{fig_OmegaM_move} noted as "NLFD-AEVI" ;
\item the previous approximation obtained by taking the average of the velocity of the four vertices defining a face and projected along the surface normal vector noted as "AVG";
\item the method based on the trilinear mapping noted as "TRI-MAP" and used as reference for the exact values of the IFMV.
\end{enumerate}
For each of these approaches, the preservation of uniform flow is tested. Then different quantities are compared by computing the maximum absolute error :
\begin{itemize}
\item comparison of the sum of the IFMV to the NLFD time derivative of the cell volume computed using the numerical scheme of the flow solver. This comparison is similar to a demonstration of the preservation of uniform flow:
\begin{equation}
AbsErr_1 = \max_{0\leq n \leq 2N} \left\{ \max_{1\leq n_v \leq N_{cell}} \left| \left( \sum_{m=1}^{N_f} G_m(n_v,t_n) \right)_{METHOD} - \left(\frac{\partial \Omega}{\partial t}(n_v,t_n) \right)_{NLFD} \right| \right\};
\label{eq_comparison_abserr_1}
\end{equation}
\item comparison of the IFMV to the reference integrated face mesh velocities (TRI-MAP) in each direction $dir=$ $x$, $y$ or $z$ :
\begin{equation}
AbsErr_2 = \max_{0\leq n \leq 2N} \left\{ \max_{1\leq n_v \leq N_{cell}} \left|  \left(G_{m,dir}(n_v,t_n) \right)_{METHOD} - \left( G_{m,dir}(n_v,t_n) \right)_{TRI-MAP} \right|  \right\}.
\label{eq_comparison_abserr_2}
\end{equation}
\end{itemize} 

\subsection{Test cases}
This section presents the different mesh motions impose as test cases. The time period is always taken to be unity. All tests are performed on a square mesh of size $10\times10\times10$, and of lengths $L_x=3.2$, $L_y=2.8$, and $L_z=2.4$. The undeformed positions of the mesh are indexed with the subscript $0$, if needed the RBF points are indexed with the subscript $r$.  The two parameters in the JST scheme are $\kappa^{(2)}=1$ and $\kappa^{(4)}=1/32$. The simulations are run for a number of harmonics, $N$ from 1 to 20. The mesh deformations for cases 2, 4 and 5 at an arbitrary time instant are presented on Figure \ref{fig_def_mesh}.

\subsubsection{Without RBF}
Three test cases are performed by directly imposing the mesh deformation to the entire mesh. The velocity of the vertices is computed based on the analytic time derivation of the vector position of the vertices. For any vertex, its initial position is noted $(x_0,y_0,z_0)$. The parameters $A_x$, $A_y$, $A_z$, $R$, and $\alpha_0$ can be arbitrarily chosen as long as no degenerative cells (cells with negative volume) appear during the motion. The analytic functions employed for the motions are as follows :
\begin{description}
\item[Case 1 :] 1-harmonic sinusoidal perturbation of the mesh with a linear motion, the direction is held fixed while each point has its own motion amplitude based on its initial position :
\begin{equation}
\left\{
\begin{array}{l}
x(t)=x_{0}+A_x \sin \left( \frac{\pi x_{0}}{L_x} \right) \sin \left( \frac{\pi y_{0}}{L_y} \right) \sin \left( \frac{\pi z_{0}}{L_z} \right)  \sin \left( 2 \pi t \right)\\
y(t)=y_{0}+A_y \sin \left( \frac{\pi x_{0}}{L_x} \right) \sin \left( \frac{\pi y_{0}}{L_y} \right)  \sin \left( \frac{\pi z_{0}}{L_z} \right)  \sin \left( 2 \pi t \right)\\
z(t)=z_{0}+A_z \sin \left( \frac{\pi x_{0}}{L_x} \right) \sin \left( \frac{\pi y_{0}}{L_y} \right)  \sin \left( \frac{\pi z_{0}}{L_z} \right)  \sin \left( 2 \pi t \right)
\end{array}
\right.
\label{eq_3D_pert_1}
\end{equation}
\item[Case 2 :] 2D perturbation of the mesh with a non linear motion; however the time-average volume swept by a face, $\hat{G}_{m,0} = 0$ in (\ref{eq_DFT_Gm}). For any cell the projection of the motion along a plane $z=constant$ is shown in Figure \ref{fig_motion_case2}  :
\begin{equation}
\left\{
\begin{array}{l}
\alpha(t) = \alpha_0 \sin(2\pi t) \\
x(t)=x_{0}+y_{0}\cos (\frac{\pi}{2}-\alpha(t))\\
y(t)=y_{0}\sin (\frac{\pi}{2}-\alpha(t))\\
z(t)=z_{0}
\end{array}
\right.
\label{eq_3D_pert_2}
\end{equation}
\item[Case 3 :] 2D perturbation of the mesh with a non linear motion and $\hat{G}_{m,0} \neq 0$ in (\ref{eq_DFT_Gm}), the deformation is prescribed only for the interior grid points while the boundary points are fixed. The projection of the motion along a plane $z = constant$ is identical to the movement of the 3rd node on the Figure \ref{fig_non_periodic_OmegaM_move} presented in section \ref{subsec_Tardiff_approach} :
\begin{equation}
\left\{
\begin{array}{l}
\alpha(t) = 2 \pi t\\
x(t)=x_{0}+R(1-\cos (\alpha(t))\\
y(t)=y_{0}+R \sin(\alpha(t))\\
z(t)=z_{0}
\end{array}
\right.
\label{eq_3D_pert_3}
\end{equation}
\end{description}

\subsubsection{With RBF}
Two test cases are performed by deforming the mesh through the RBF. The analytic functions employed for the RBF motions are as follows :
\begin{description}
\item[Case 4 :] 3D perturbation of the mesh using the RBF, with each point having its own linear motion (amplitude and direction) :
\begin{equation}
\left\{
\begin{array}{c}
s_x(t)=r_x (x_{0,r},y_{0,r},z_{0,r})  \sin(2\pi y_{0,r})  \sin(2\pi z_{0,r}) \sin(2 \pi t)\\
s_y(t)=r_y (x_{0,r},y_{0,r},z_{0,r})  \sin(2\pi x_{0,r})  \sin(2\pi z_{0,r}) \sin(2 \pi t)\\
s_z(t)=r_z (x_{0,r},y_{0,r},z_{0,r})  \sin(2\pi y_{0,r})  \sin(2\pi z_{0,r}) \sin(2 \pi t)\\
\text{where }r_x (x_{0,r},y_{0,r},z_{0,r}); r_y (x_{0,r},y_{0,r},z_{0,r})\text{ and } \\
r_z (x_{0,r},y_{0,r},z_{0,r}) \text{ are randomly generated }
\end{array}
\right.
\label{eq_3D_pert_5}
\end{equation}
\item[Case 5 :] simulation of a sinusoidal pitching motion :
\begin{equation}
\left\{
\begin{array}{l}
\alpha(t) = \alpha_0 \cos(2\pi t) \\
x_p = 0.621 L_x\\ 
s_x(t)=(x_{0,r}-x_p)[\cos (\alpha(t))-1] + y_{0,r}\sin (\alpha(t))\\
s_y(t)=-(x_{0,r}-x_p)\sin (\alpha(t)) + y_{0,r}[\cos (\alpha(t))-1] \\
s_z(t)=z_{0,r}
\end{array}
\right.
\label{eq_3D_pert_8}
\end{equation}
\end{description}

\begin{figure}[!htbp]
\centering
	\begin{subfigure}[b]{0.48\textwidth}
	\includegraphics[width=7.5cm]{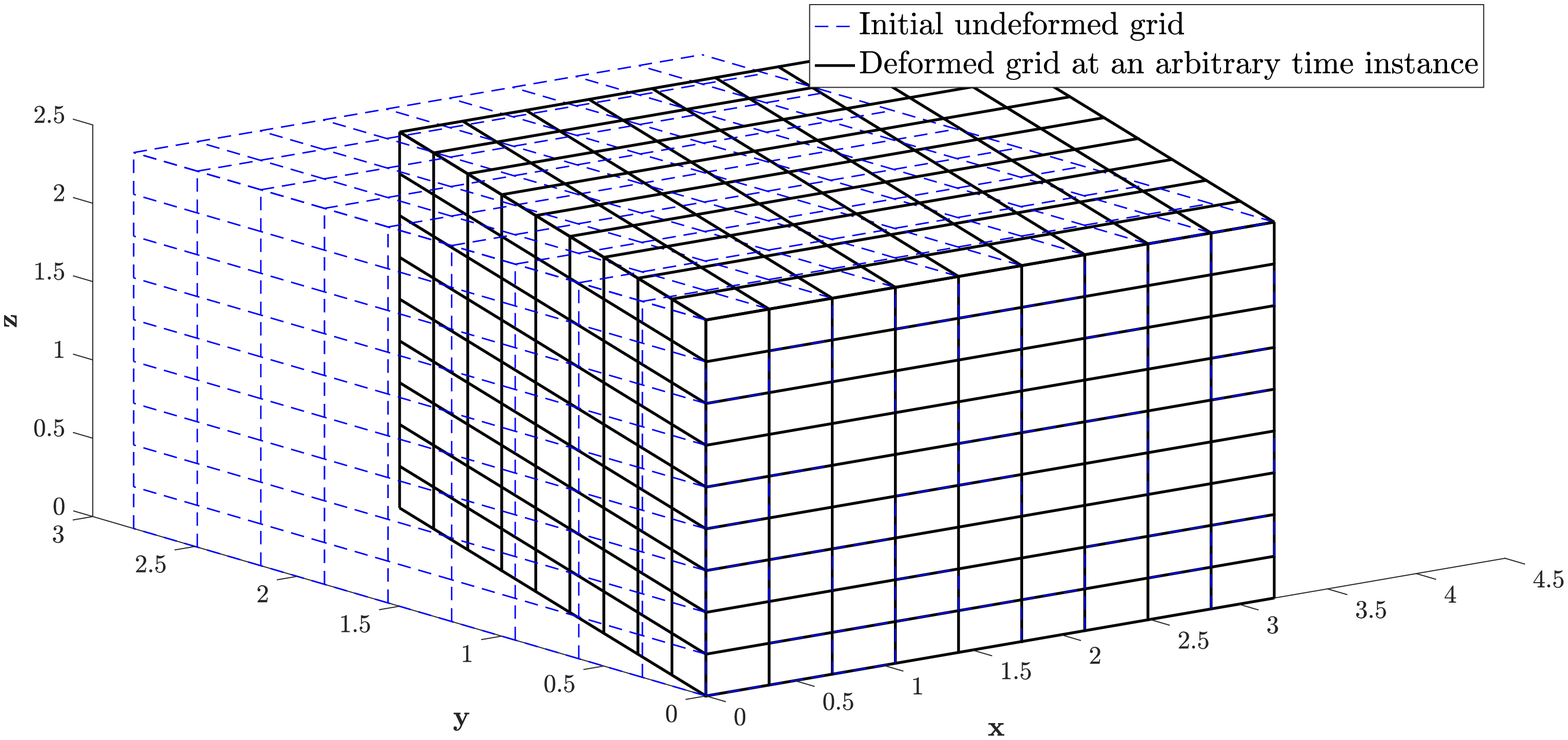}
	\caption{Case 2}
	\end{subfigure}
	\begin{subfigure}[b]{0.48\textwidth}
	\includegraphics[width=7.5cm]{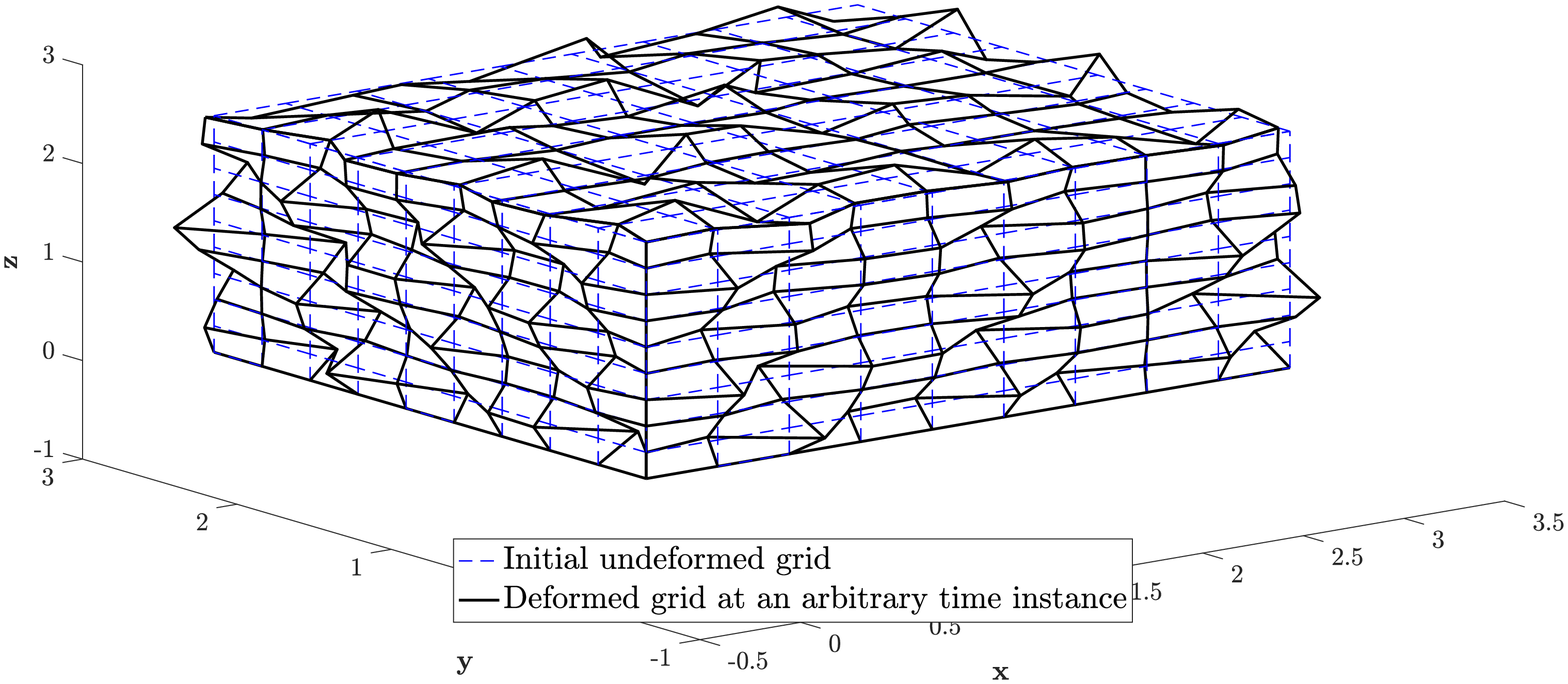}
	\caption{Case 4}
	\end{subfigure}
	
	\begin{subfigure}[b]{0.48\textwidth}
	\includegraphics[width=7.5cm]{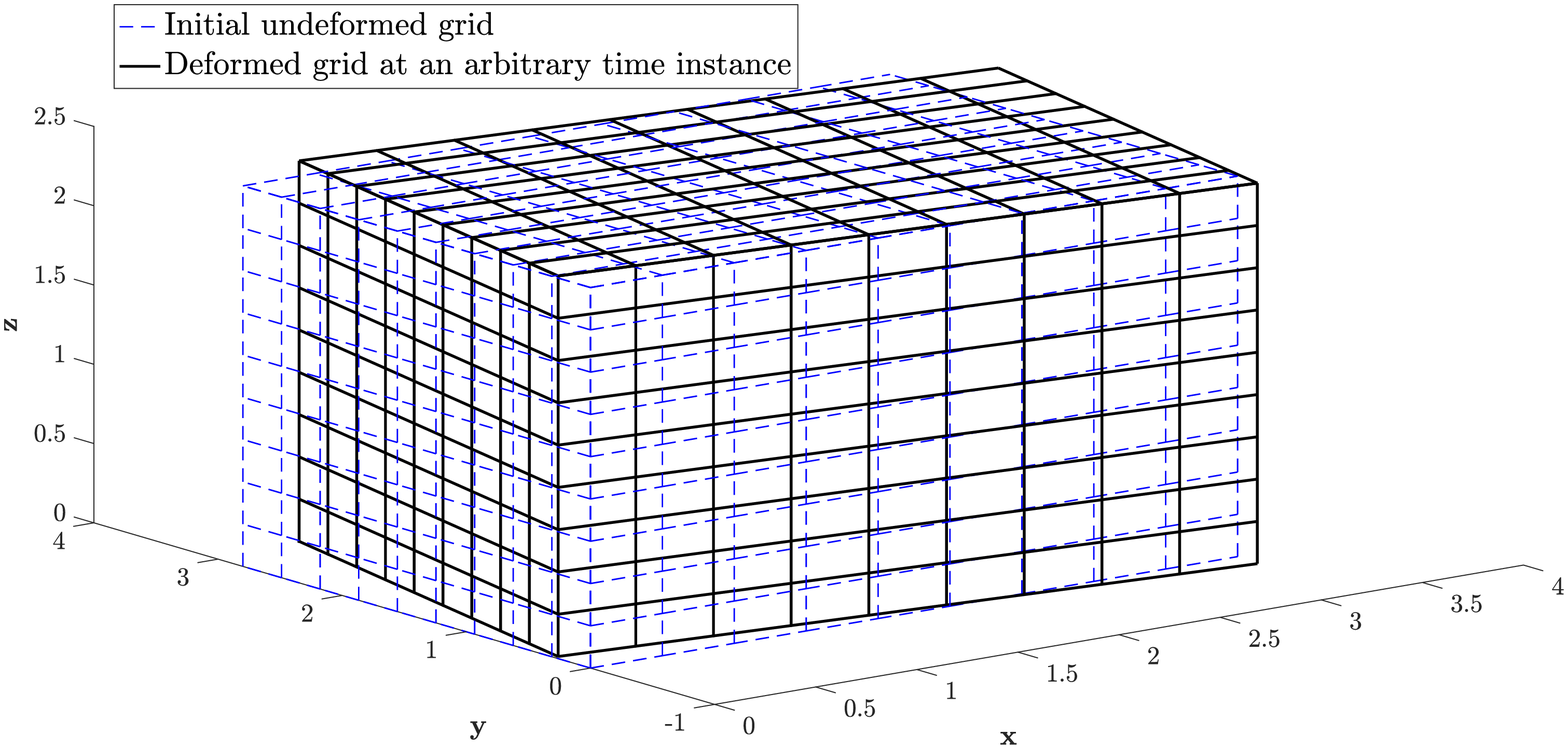}
	\caption{Case 5}
	\end{subfigure}
\caption{Mesh deformations of the exterior grid points for cases 2, 4 and 5 at an arbitrary chosen time step}
\label{fig_def_mesh}
\end{figure}

\subsection{Freestream preservation}
The results demonstrating uniform flow preservation are shown for all test cases in Figure \ref{fig_results_uniform_flow}. The evolution of the relative error defined by equation (\ref{eq_results_rel_error_W}) is presented as a function of the number of time steps $N_{ts}$.

\begin{figure}[!htbp]
\centering
	\begin{subfigure}[b]{0.48\textwidth}
		\centering
		\includegraphics[width=7.5cm]{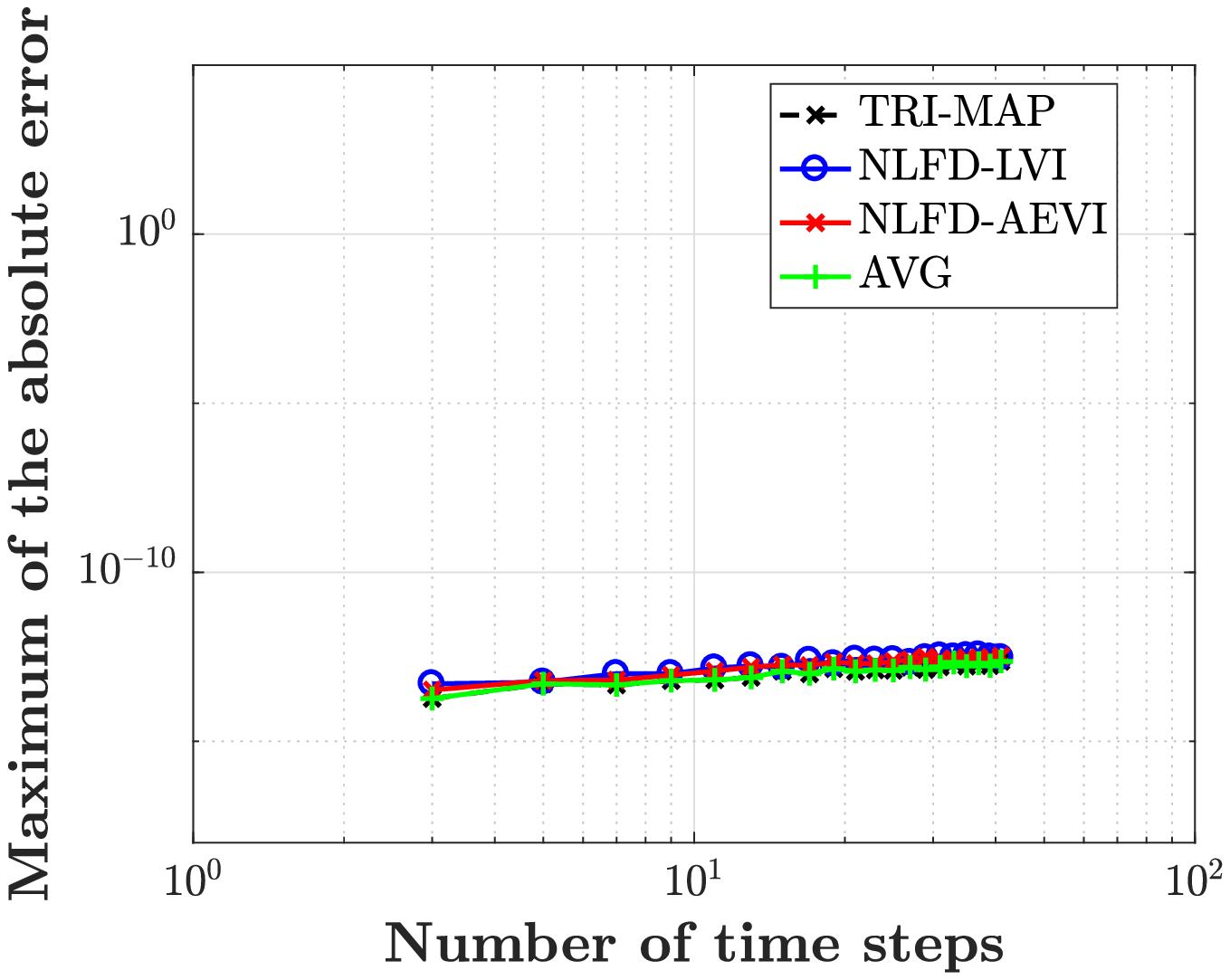}
		\caption{Case 1}
	\end{subfigure}
	\begin{subfigure}[b]{0.48\textwidth}
		\centering
		\includegraphics[width=7.5cm]{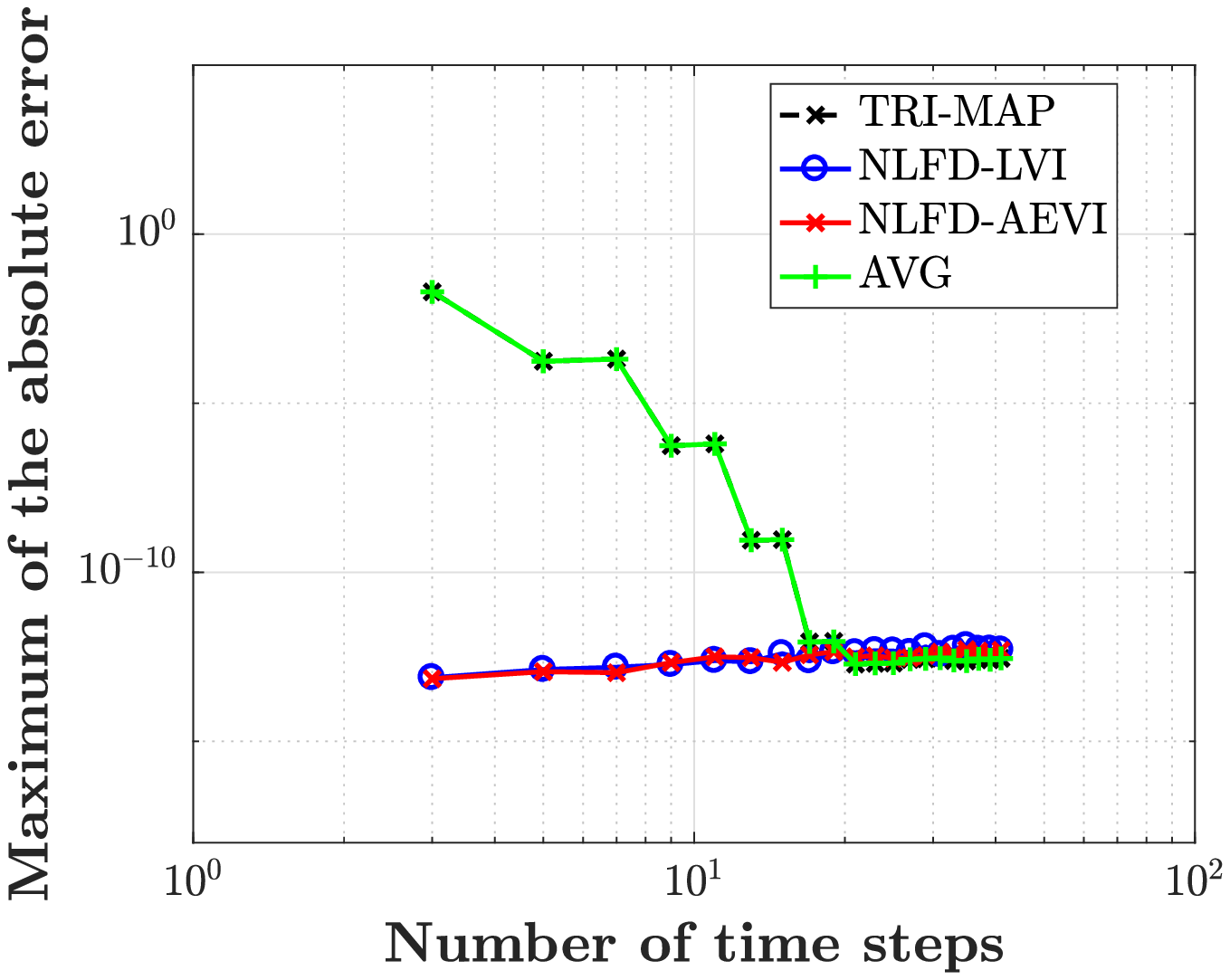}
		\caption{Case 2}
	\end{subfigure}
	
	\begin{subfigure}[b]{0.48\textwidth}
		\centering
		\includegraphics[width=7.5cm]{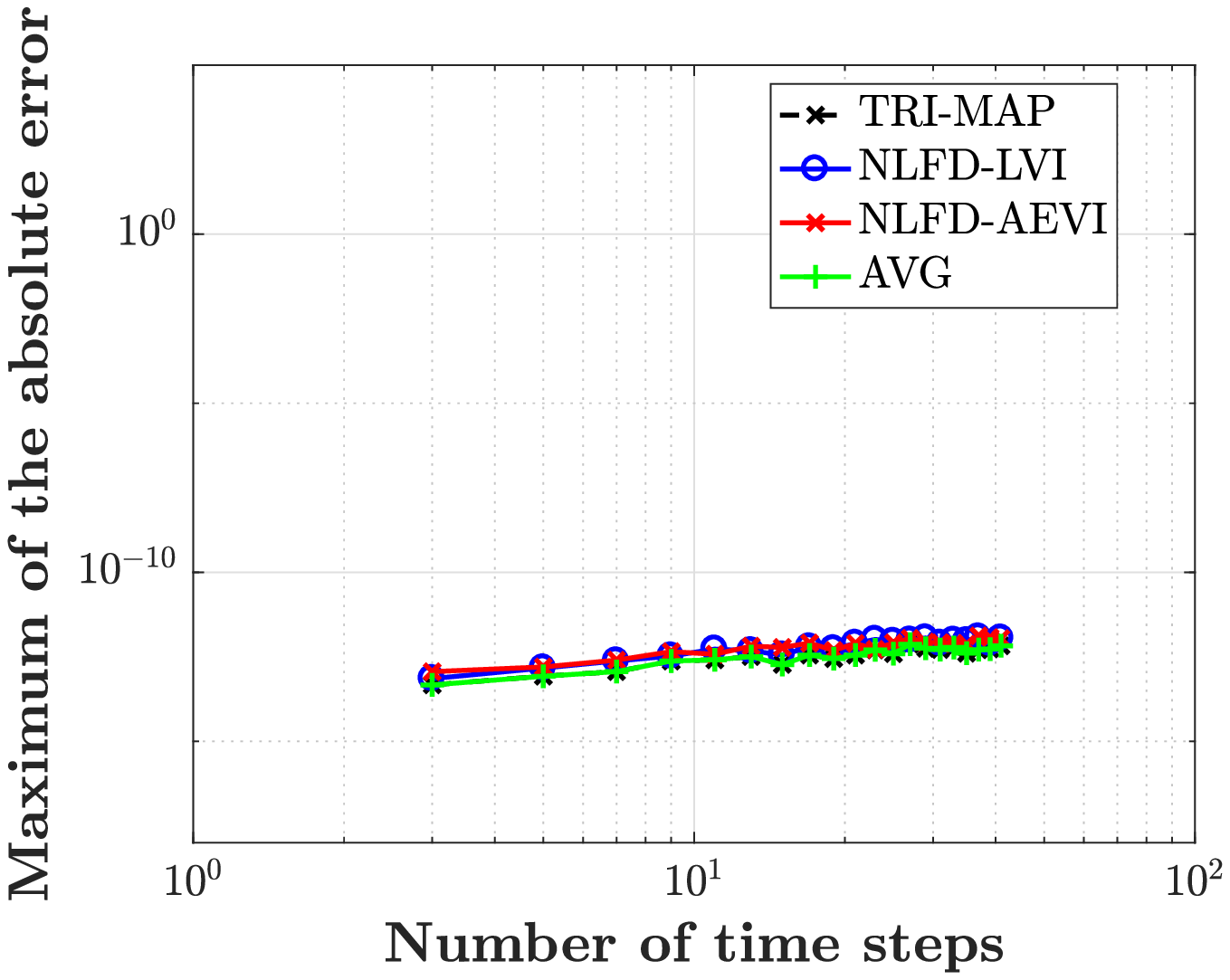}
		\caption{Case 3}
	\end{subfigure}
	\begin{subfigure}[b]{0.48\textwidth}
		\centering
		\includegraphics[width=7.5cm]{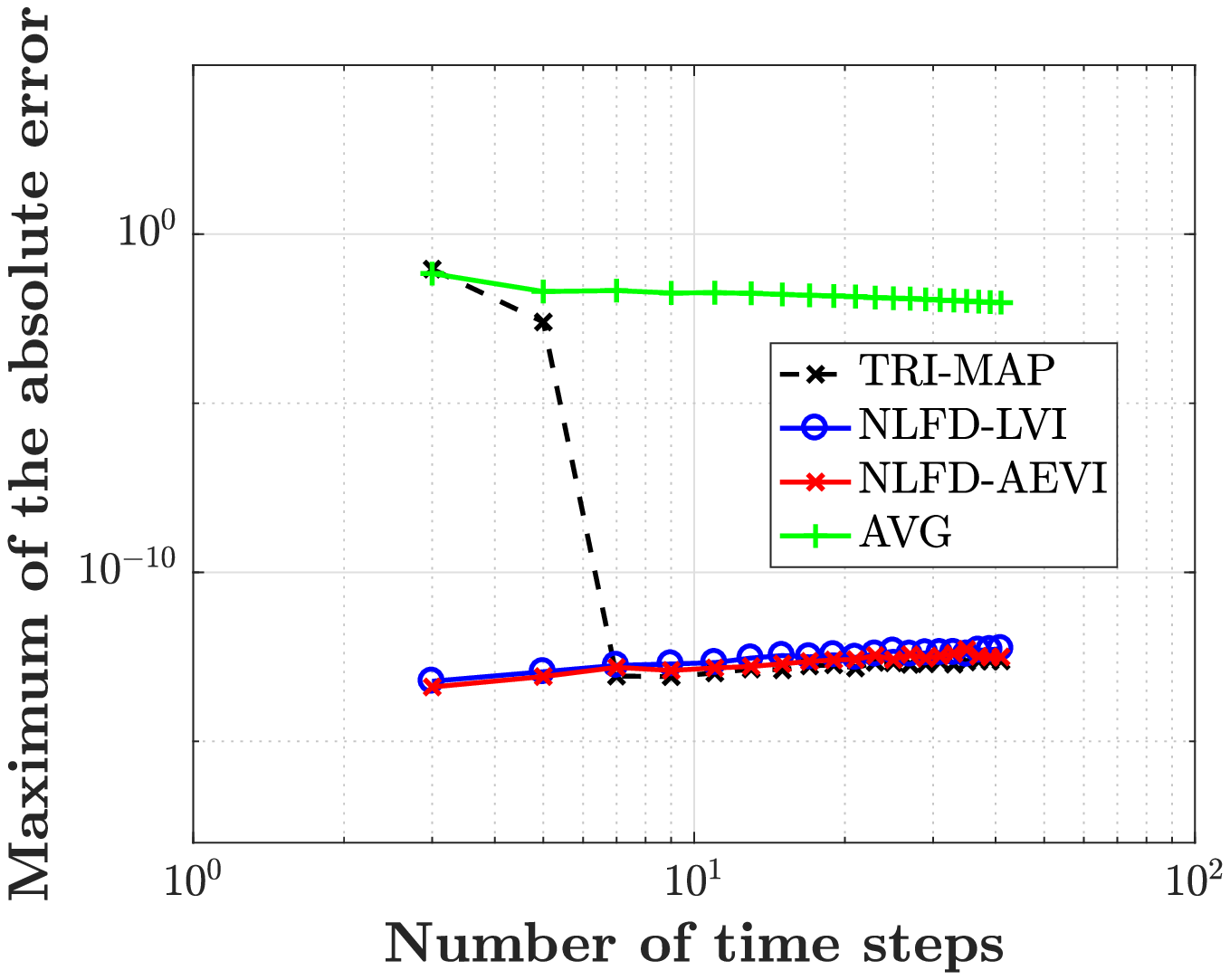}
		\caption{Case 4}
	\end{subfigure}
	
	\begin{subfigure}[b]{0.48\textwidth}
		\centering
		\includegraphics[width=7.5cm]{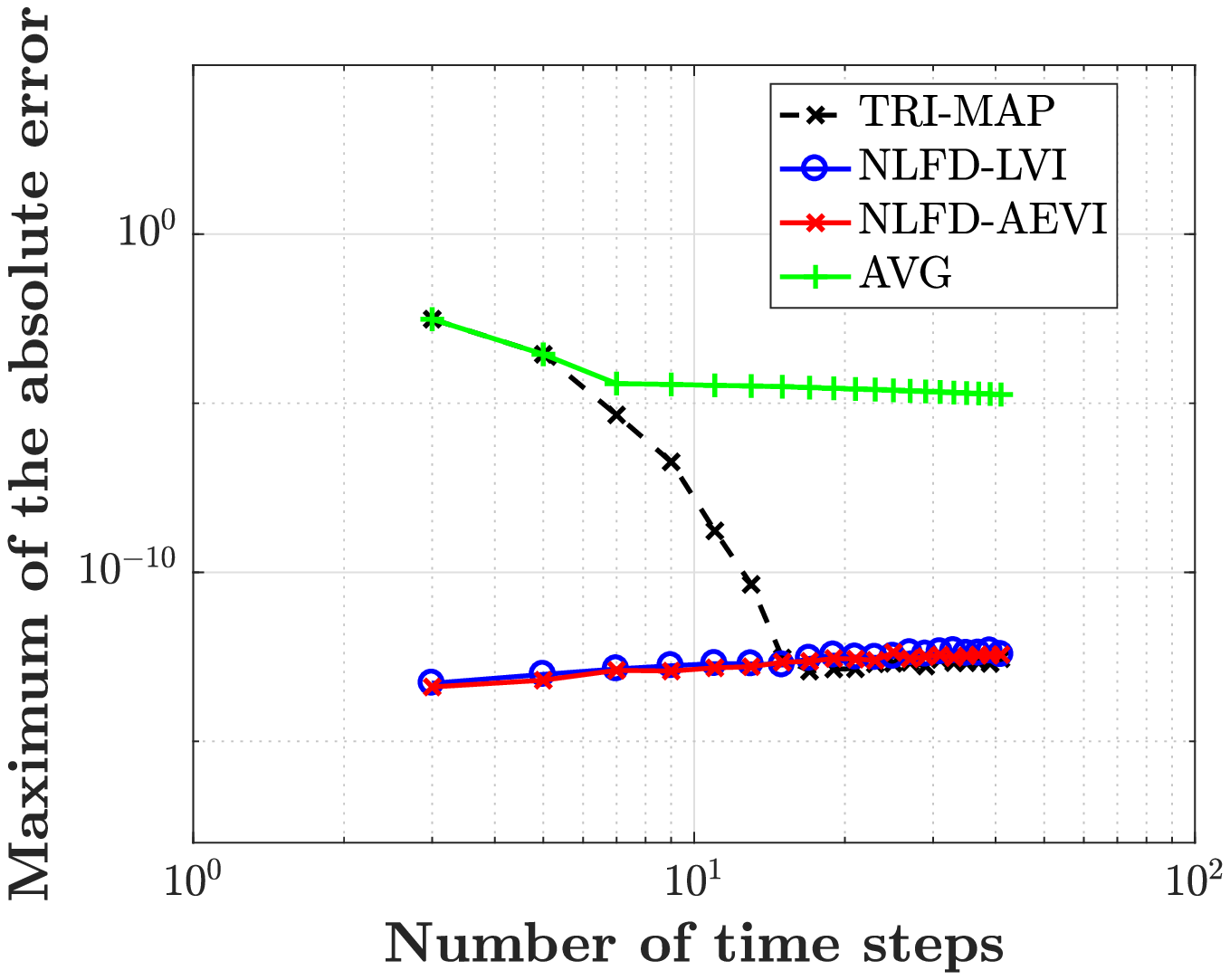}
		\caption{Case 5}
	\end{subfigure}
\caption{Relative error regarding the uniform flow preservation for each test case}
\label{fig_results_uniform_flow}
\end{figure}

The results show that the two methods employing the IFMV deduced from the Fourier discretization preserve uniform flow, while the approximation derived from the AVG yields the least accurate results. This is consistent since for both methods NLFD-LVI and NLFD-AEVI, despite different definitions of the volumetric increments, they still ensure that the sum of the temporal derivative of the volumetric increments is equal to the temporal derivative of the cell volume evaluated in the frequency domain (equation~(\ref{eq_dOmega_sum_dOmegam})).

It is also observed that using the (TRI-MAP) integrated face mesh velocities preserves uniform flow and thus satisfies the GCL given a sufficient number of harmonics (see cases 2, 4 and 5) which is expected. Its rate of convergence should be exactly the same as the rate of convergence of the time derivative of the cell volume in the Fourier space. This is verified in the next section.   

\subsection{Comparison of the integrated face mesh velocities to the reference value}
The results are shown on Figures \ref{fig_results_comparison_case_1} through \ref{fig_results_comparison_case_5}. It is important to note that for all figures, the graph (a) refers to equation (\ref{eq_comparison_abserr_1}) as the function of the number of time steps and is not the sum of the graphs from (b), (c) and (d) which refer to equation (\ref{eq_comparison_abserr_2}). The errors that appear on the $y$-axis of the figures are the max norm between the investigated approaches, both NLFD-based and AVG and the reference approach (TRI-MAP). 

\begin{figure}[!htbp]
\centering
	\begin{subfigure}[b]{0.48\textwidth}
		\centering
		\includegraphics[width=7.5cm]{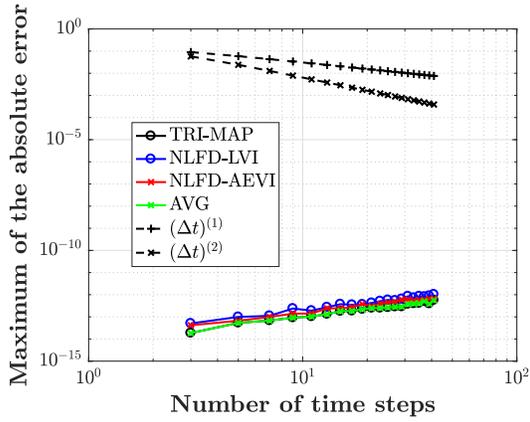}
		\caption{}
	\end{subfigure}
	\begin{subfigure}[b]{0.48\textwidth}
		\centering
		\includegraphics[width=7.5cm]{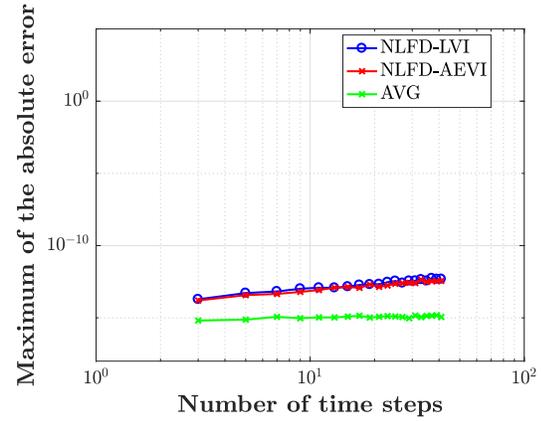}
		\caption{}
	\end{subfigure}
	
	\begin{subfigure}[b]{0.48\textwidth}
		\centering
		\includegraphics[width=7.5cm]{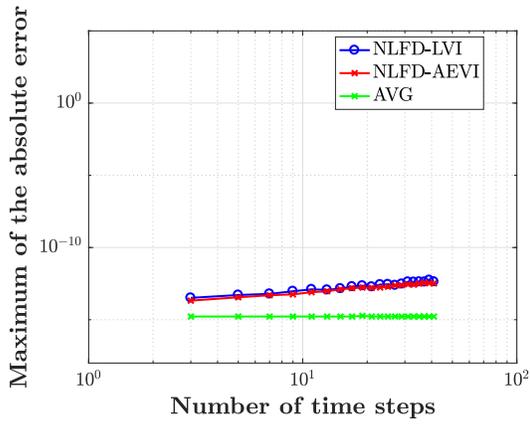}
		\caption{}
	\end{subfigure}
	\begin{subfigure}[b]{0.48\textwidth}
		\centering
		\includegraphics[width=7.5cm]{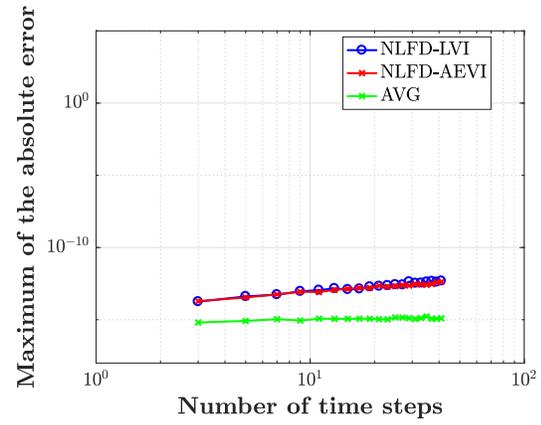}
		\caption{}
	\end{subfigure}		
\caption{Case 1 : (a) Comparison of the sum of the integrated face mesh velocities to the NLFD time derivative of the volume (b) Comparison of the individual integrated face mesh velocity to the reference values (TRI-MAP) in the $x$ direction (c) in the $y$ direction (d) in the $z$ direction}	
\label{fig_results_comparison_case_1}
\end{figure}

Regarding the comparison of the {\it sum of the} integrated face mesh velocities (IFMV) to the NLFD temporal derivative of the volume from Figures~\ref{fig_results_comparison_case_1}(a) through Figure~\ref{fig_results_comparison_case_5}(a), the results show that the sum of the IFMV computed with the methods NLFD-LVI, NLFD-AEVI and TRI-MAP converge to the expected value for all cases while the AVG method provides the correct value only for cases 1 through 3 and yields a constant absolute error above $10^{-5}$ for cases 4 and 5. Recall that the maximum error in the sum of the IFMV is a measure of the level to which GCL is satisfied as given in the semi-discrete GCL equation (\ref{eq_semi_discrete_GCL_2}). Hence the NLFD-based approaches prove to satisfy the GCL for all considered grid deformation and for any number of harmonics which is expected by design. The reference approach (TRI-MAP) satisfies this requirement exactly for linear deformation cases as shown for Cases 1 (Figure~\ref{fig_results_comparison_case_1}(a)) and 4 (Figure~\ref{fig_results_comparison_case_4}(a)) for any number of harmonics and converged spectrally for nonlinear deformation cases (Cases 2, 3, and 5). The spectral rate of convergence is observed compared to the first-order backward finite-difference ($\Delta t^{(1)}$) and second-order centered finite-difference ($\Delta t^{(2)}$) approximating the time derivative of the cell volume. As expected, this rate of convergence is found to be similar for the preservation of uniform flow using the reference TRI-MAP method. However, the AVG approach is not designed to enforce the GCL, it is only an approximation based on the mesh velocities and face metrics and hence for the cases considered in this article, the method proved to ensure the GCL with an accuracy up to $10^{-5}$.

A comparison of the {\it individual} integrated face mesh velocities for each direction reveals the limits and provides interesting insights of the investigated approaches. Two primary observations can be made. First, the NLFD based approaches converge at most at second order as expected based on corollary \ref{th_error}, if the mesh deformation along the observed direction is nonlinear. For Cases 2, 3, and 5, the mesh deformation in both the $x$-and $y$-directions are nonlinear as shown in sub-figures (b) and (c) of Figures~\ref{fig_results_comparison_case_2}, ~\ref{fig_results_comparison_case_3}, and~\ref{fig_results_comparison_case_5}. One exception is the spectral rate of convergence for the $y$-direction in Case 2. These results can be explained by analyzing in details the mesh movement. Since the motion is in two dimensions, let us consider a constant $z$ plane, then the deformation of any cell can be represented as shown in Figure \ref{fig_motion_case2}.
\begin{figure}[!htbp]
\centering
	\begin{tikzpicture}[scale=1]

   		\draw[very thick,->] (0,0) -- (1,0) node[above]{$x$};
   		\draw[very thick,->] (0,0) -- (0,1) node[above,right]{$y$};
		\draw[dash dot] (-2,0)rectangle(9,6);
		\draw[very thick] (1.5,2)rectangle(5,5);	
		\fill[blue,opacity=0.6,domain=60:90] plot ({1.5+6*cos(\x)}, {-1+6*sin(\x)}) -- plot ({5+6*cos(90-30*((\x-60)/30))}, {-1+6*sin(90-30*((\x-60)/30))}) -- cycle;  		 
  		\fill[green,opacity=0.2,domain=60:90]	plot ({1.5+3*cos(\x)}, {-1+3*sin(\x)}) -- plot ({1.5+6*cos(90-30*((\x-60)/30))}, {-1+6*sin(90-30*((\x-60)/30))}) -- cycle;
   		\draw[blue,dash dot,domain=60:120,<->,>=latex] plot ({1.5+3*cos(\x)}, {-1+3*sin(\x)});
   		\draw[blue,dash dot,domain=60:120,<->,>=latex] plot ({5+3*cos(\x)}, {-1+3*sin(\x)});
   		\draw[blue,dash dot,domain=60:120,<->,>=latex] plot ({1.5+6*cos(\x)}, {-1+6*sin(\x)});
   		\draw[blue,dash dot,domain=60:120,<->,>=latex] plot ({5+6*cos(\x)}, {-1+6*sin(\x)});
		\draw[red,thick] (3,1.6)--(6.5,1.6)--(8,4.2)--(4.5,4.2)--cycle;  			   				
	\end{tikzpicture}
\caption{Two dimensional projection of the motion in case 2 for one cell : the exact volumetric increment in the $x$ direction is filled in clear green and in dark blue in the $y$ direction. The blue dashed dot arrows show the paths of the vertices.}
\label{fig_motion_case2}
\end{figure}
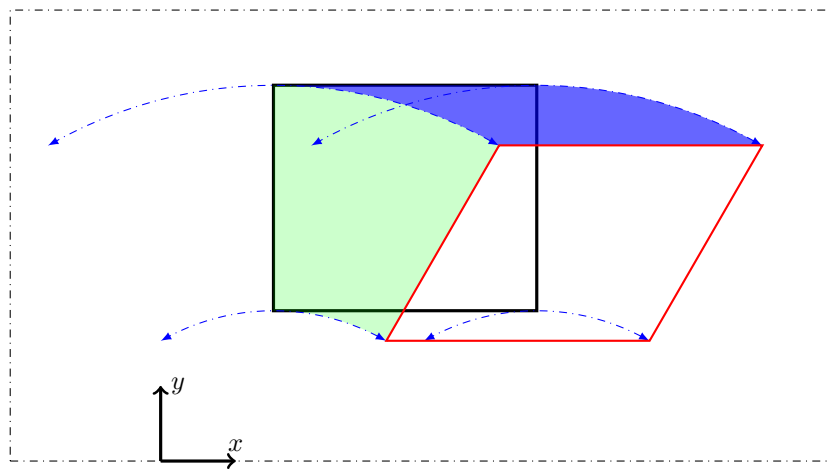
We observe that in the $y$-direction, the area swept by the faces can be exactly evaluated using a linear approximation of the curved boundaries shown in blue. Therefore in the $y$-direction, the volumetric increments are exactly computed and the {\it individual} IFMV are correctly computed using either the LVI or AEVI methods once the temporal derivative operator is converged in Fourier space. In the $x$-direction, a linear approximation is insufficient to compute exactly the volumetric increments thus the AEVI method converges at an order between one and two as stated in corollary \ref{th_error}.

Second, even if the numerical scheme enforces the GCL by preserving uniform flow, the employed method may not converge to the correct integrated face mesh velocities. The method based on the approximation of the exact volumetric increment (AEVI) is found to be converging toward the reference values at an order between one and two in the worst test cases considered here (4 \& 5). This is consistent with the derivation of the error from section \ref{subsubsec_prac_enf} and the resulting corollary (see corollary \ref{th_error}). The NLFD-LVI and AVG methods may present significant inaccuracies depending on the mesh deformation.

\begin{figure}[!htbp]
\centering
	\begin{subfigure}[b]{0.48\textwidth}
		\centering
		\includegraphics[width=7.5cm]{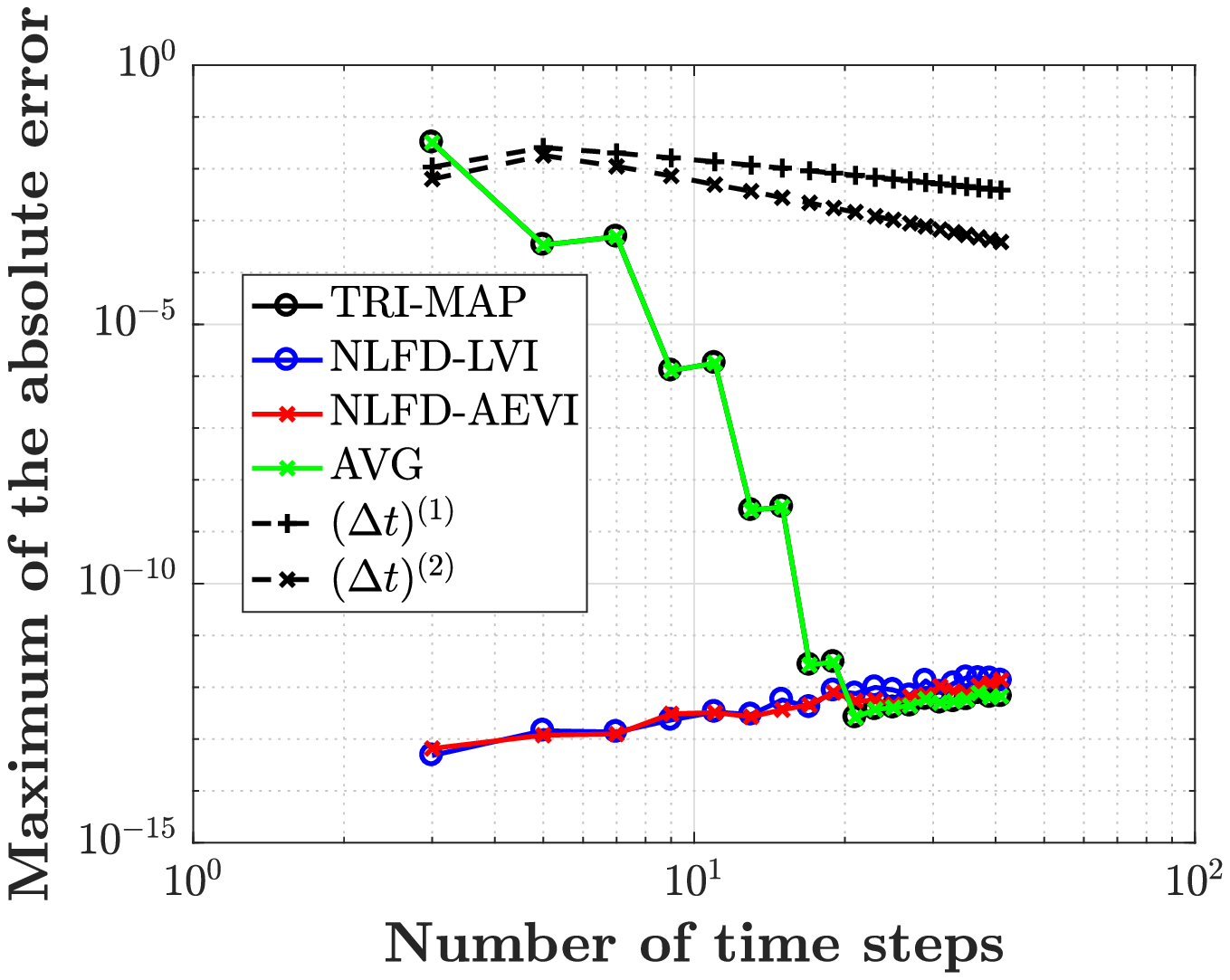}
		\caption{}
	\end{subfigure}
	\begin{subfigure}[b]{0.48\textwidth}
		\centering
		\includegraphics[width=7.5cm]{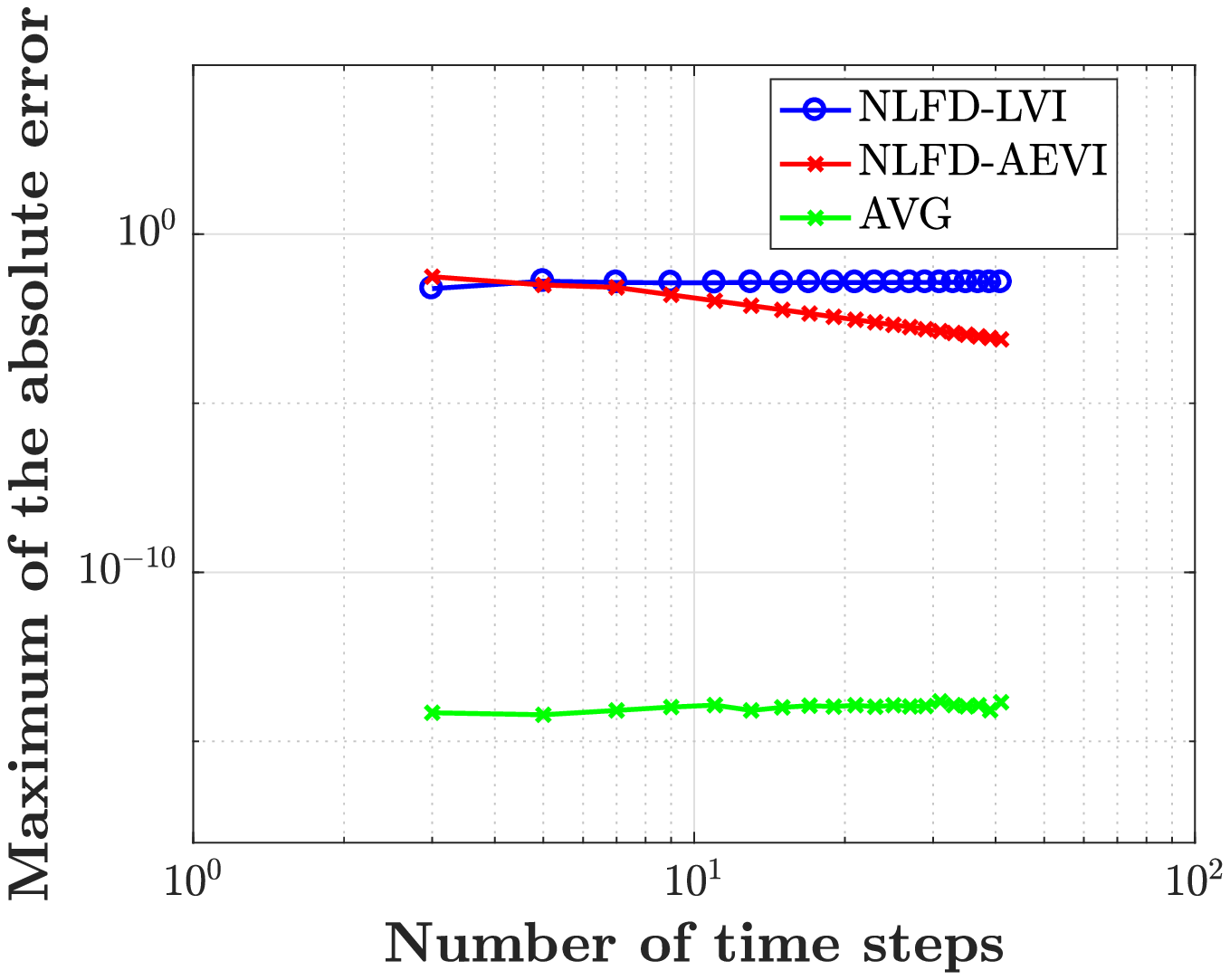}
		\caption{}
	\end{subfigure}
	
	\begin{subfigure}[b]{0.48\textwidth}
		\centering
		\includegraphics[width=7.5cm]{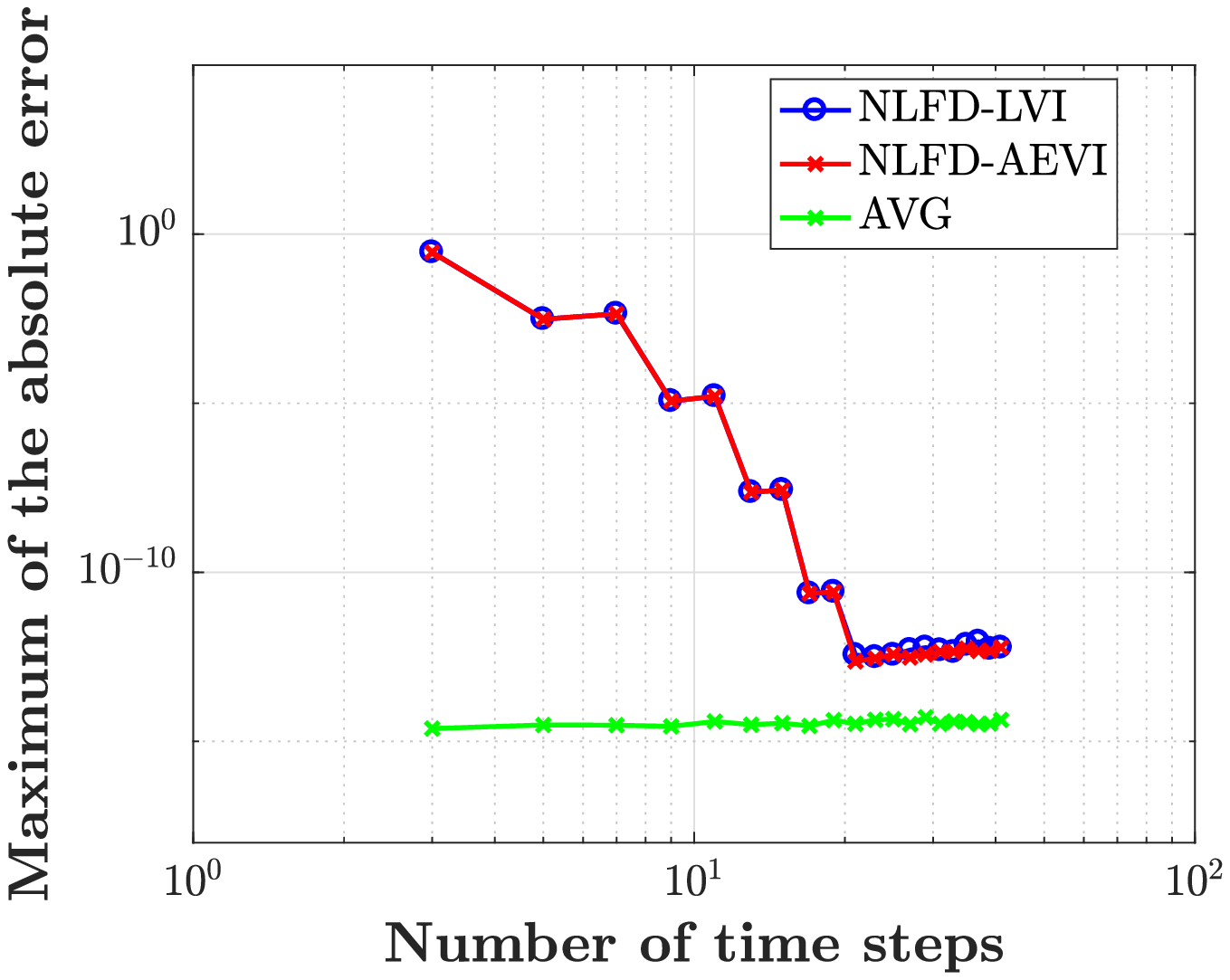}
		\caption{}
	\end{subfigure}
	\begin{subfigure}[b]{0.48\textwidth}
		\centering
		\includegraphics[width=7.5cm]{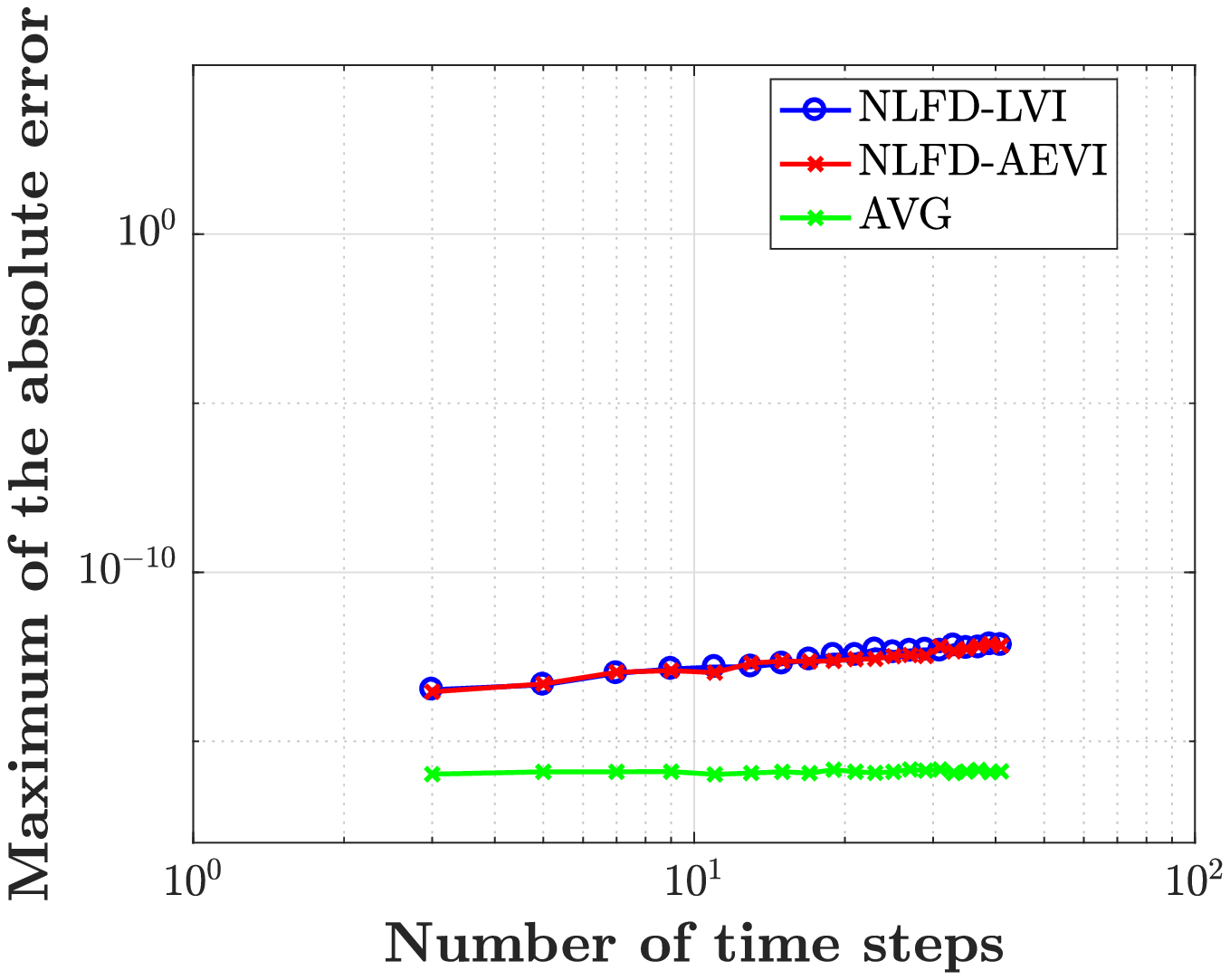}
		\caption{}
	\end{subfigure}		
\caption{Case 2 : (a) Comparison of the sum of the integrated face mesh velocities to the NLFD time derivative of the volume (b) Comparison of the individual integrated face mesh velocity to the values (TRI-MAP) in the $x$ direction (c) in the $y$ direction (d) in the $z$ direction}	
\label{fig_results_comparison_case_2}
\end{figure}

\begin{figure}[!htbp]
\centering
	\begin{subfigure}[b]{0.48\textwidth}
		\centering
		\includegraphics[width=7.5cm]{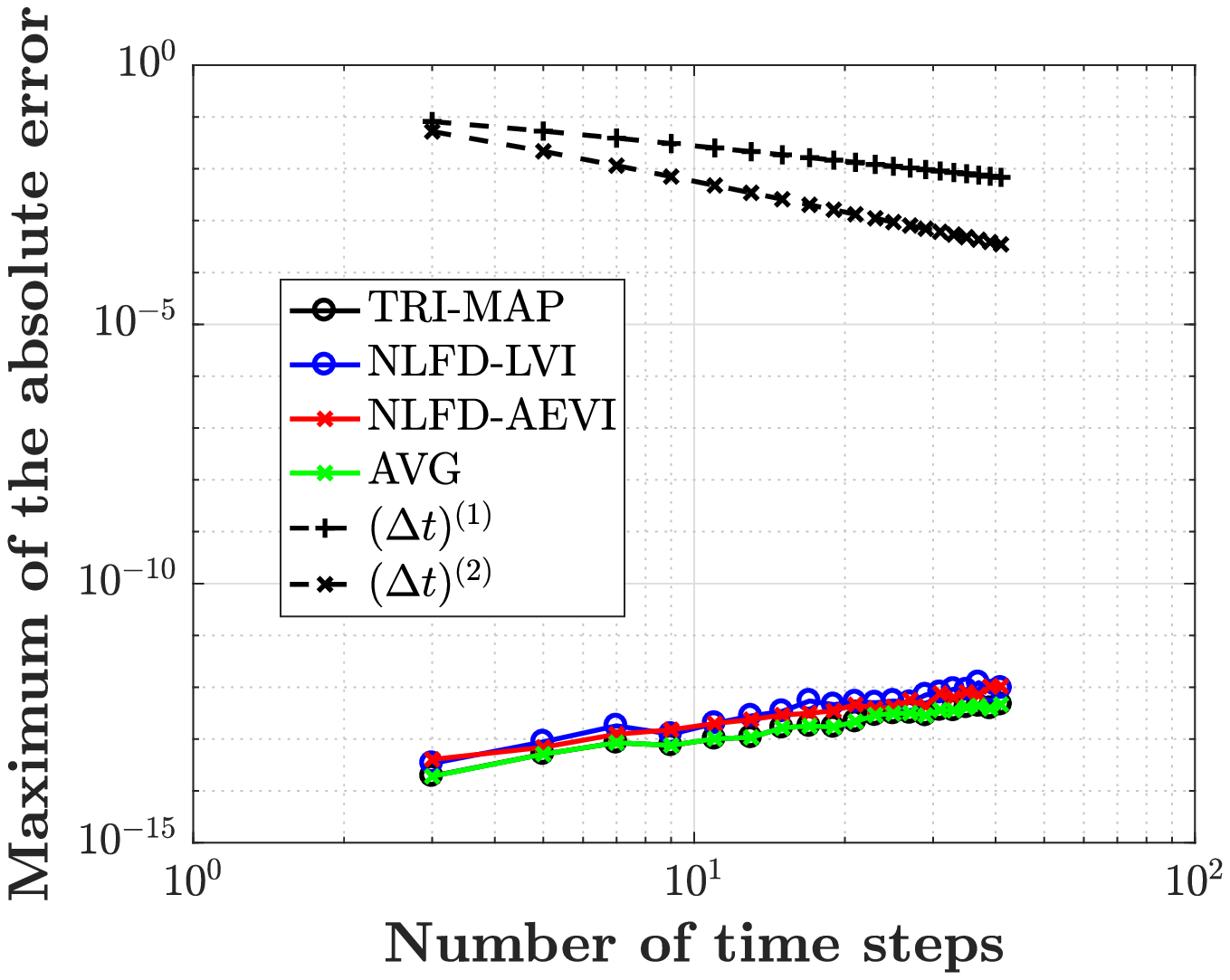}
		\caption{}
	\end{subfigure}
	\begin{subfigure}[b]{0.48\textwidth}
		\centering
		\includegraphics[width=7.5cm]{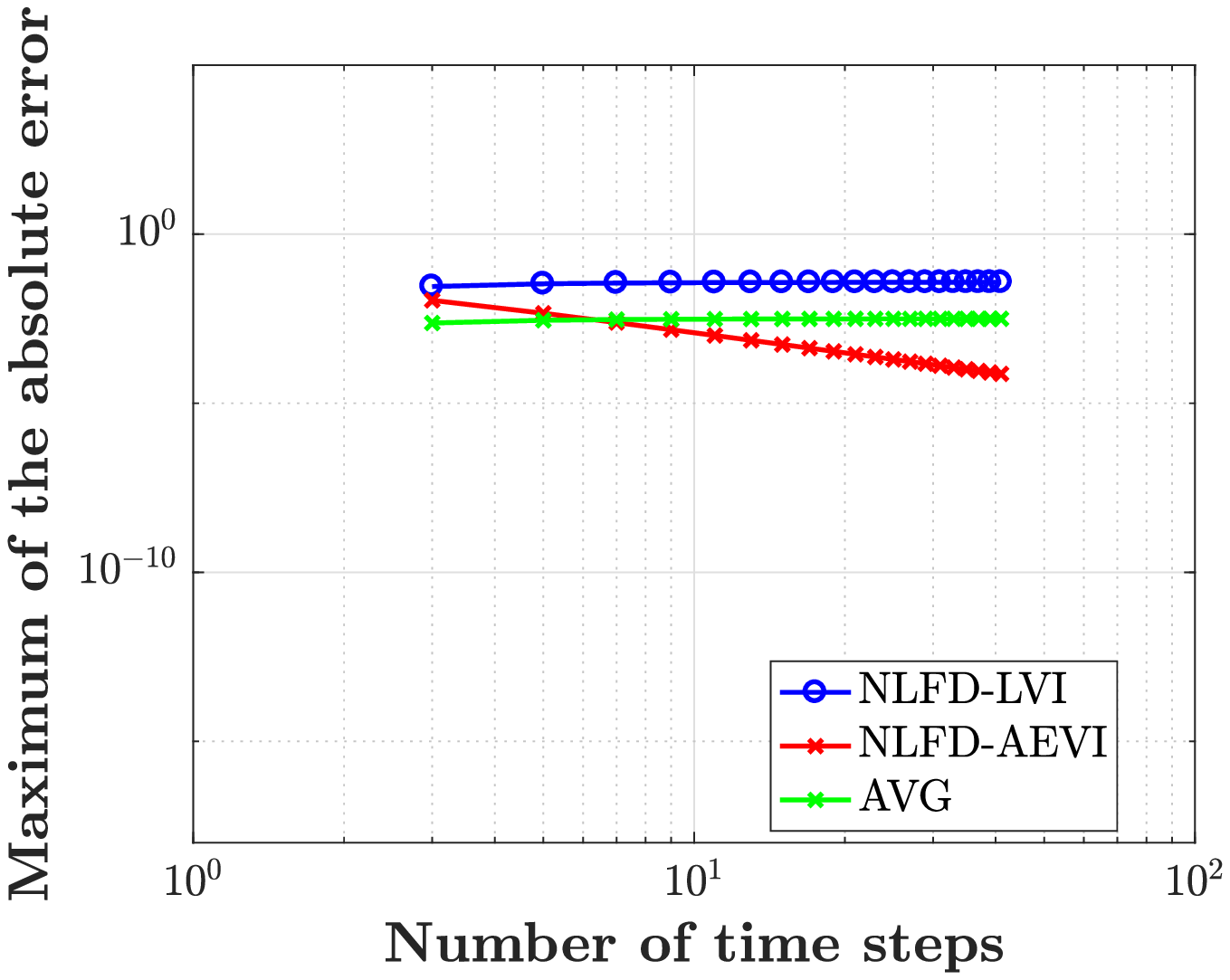}
		\caption{}
	\end{subfigure}
	
	\begin{subfigure}[b]{0.48\textwidth}
		\centering
		\includegraphics[width=7.5cm]{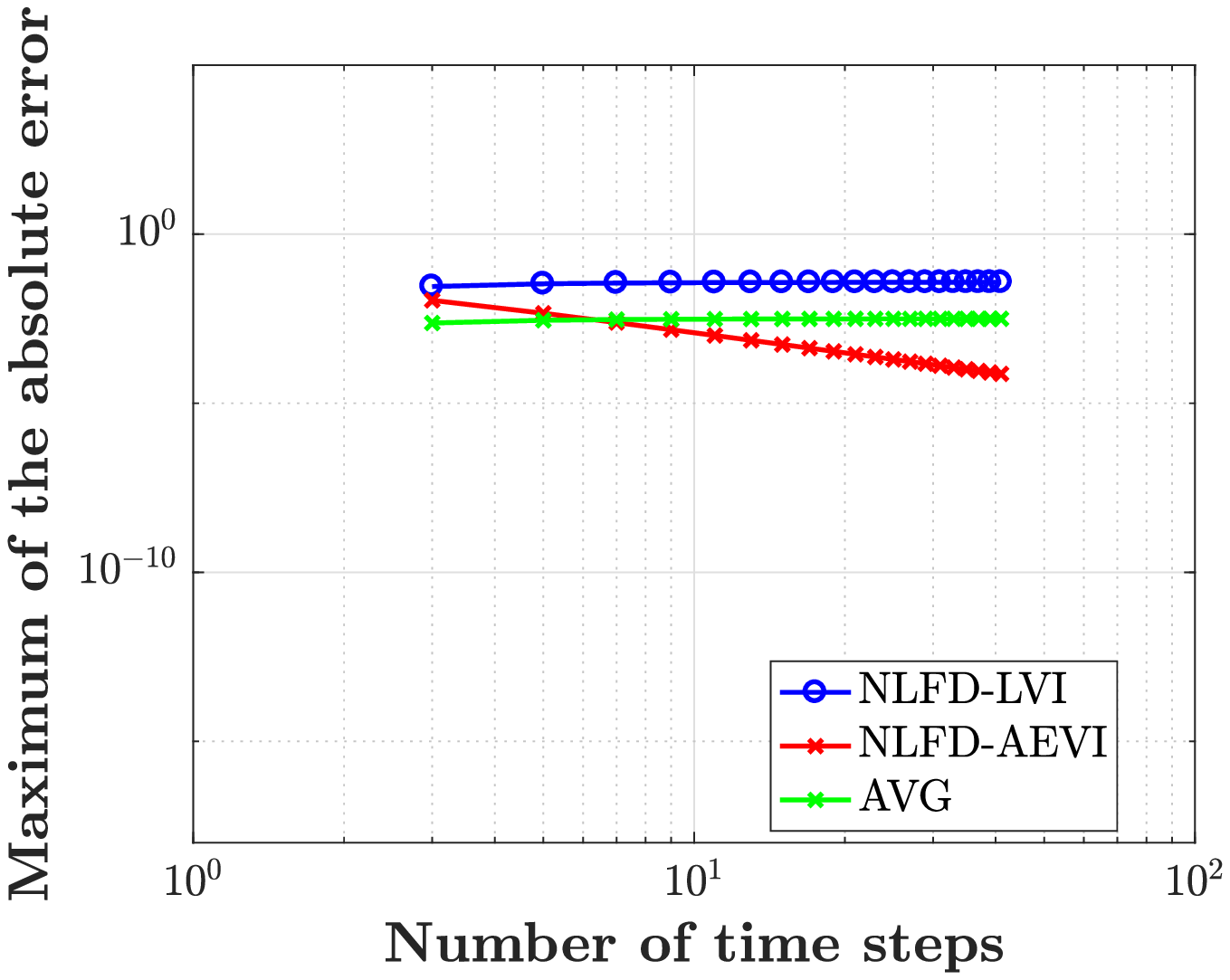}
		\caption{}
	\end{subfigure}
	\begin{subfigure}[b]{0.48\textwidth}
		\centering
		\includegraphics[width=7.5cm]{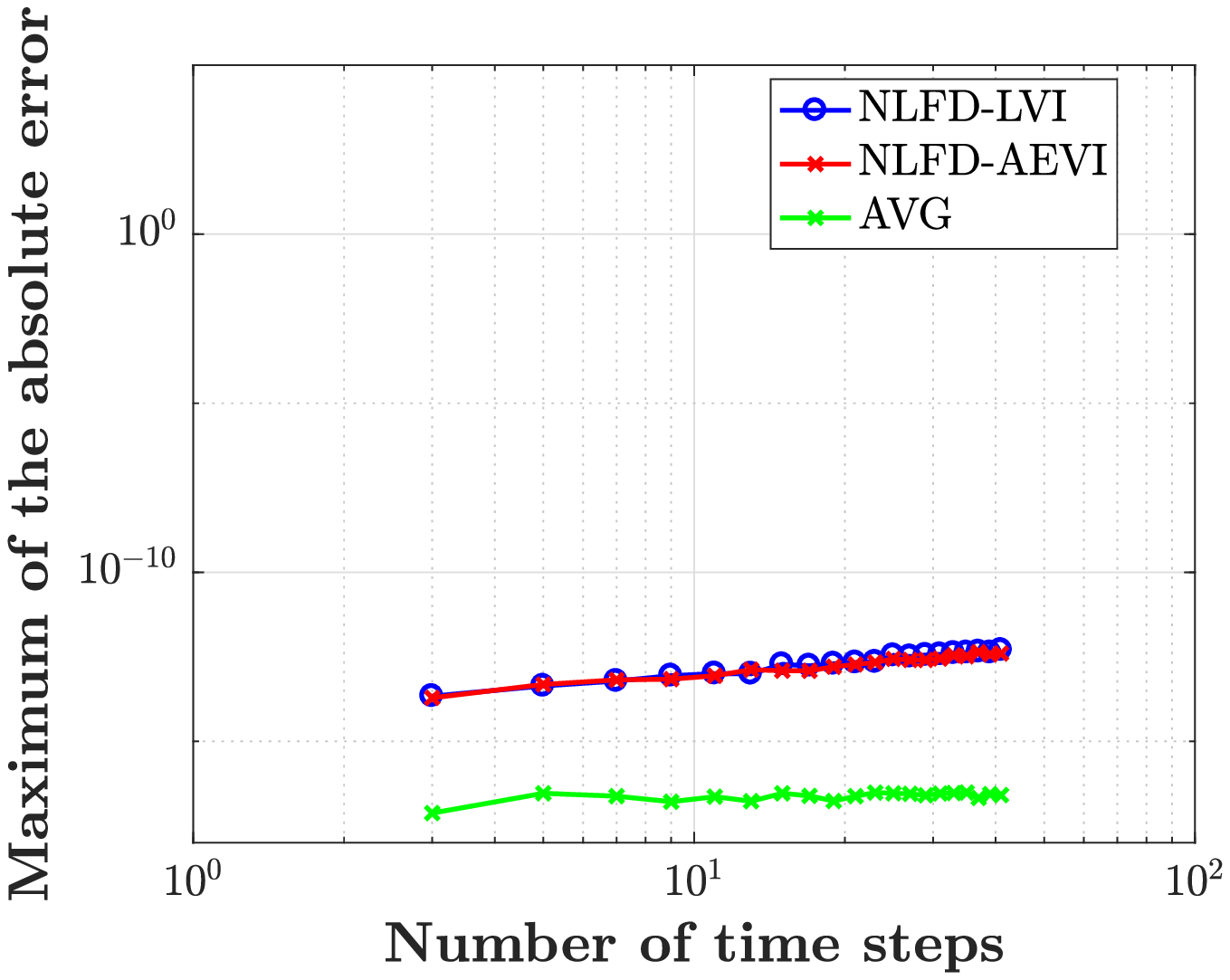}
		\caption{}
	\end{subfigure}		
\caption{Case 3 : (a) Comparison of the sum of the integrated face mesh velocities to the NLFD time derivative of the volume (b) Comparison of the individual integrated face mesh velocity to the values (TRI-MAP) in the $x$ direction (c) in the $y$ direction (d) in the $z$ direction}	
\label{fig_results_comparison_case_3}
\end{figure}

\begin{figure}[!htbp]
\centering
	\begin{subfigure}[b]{0.48\textwidth}
		\centering
		\includegraphics[width=7.5cm]{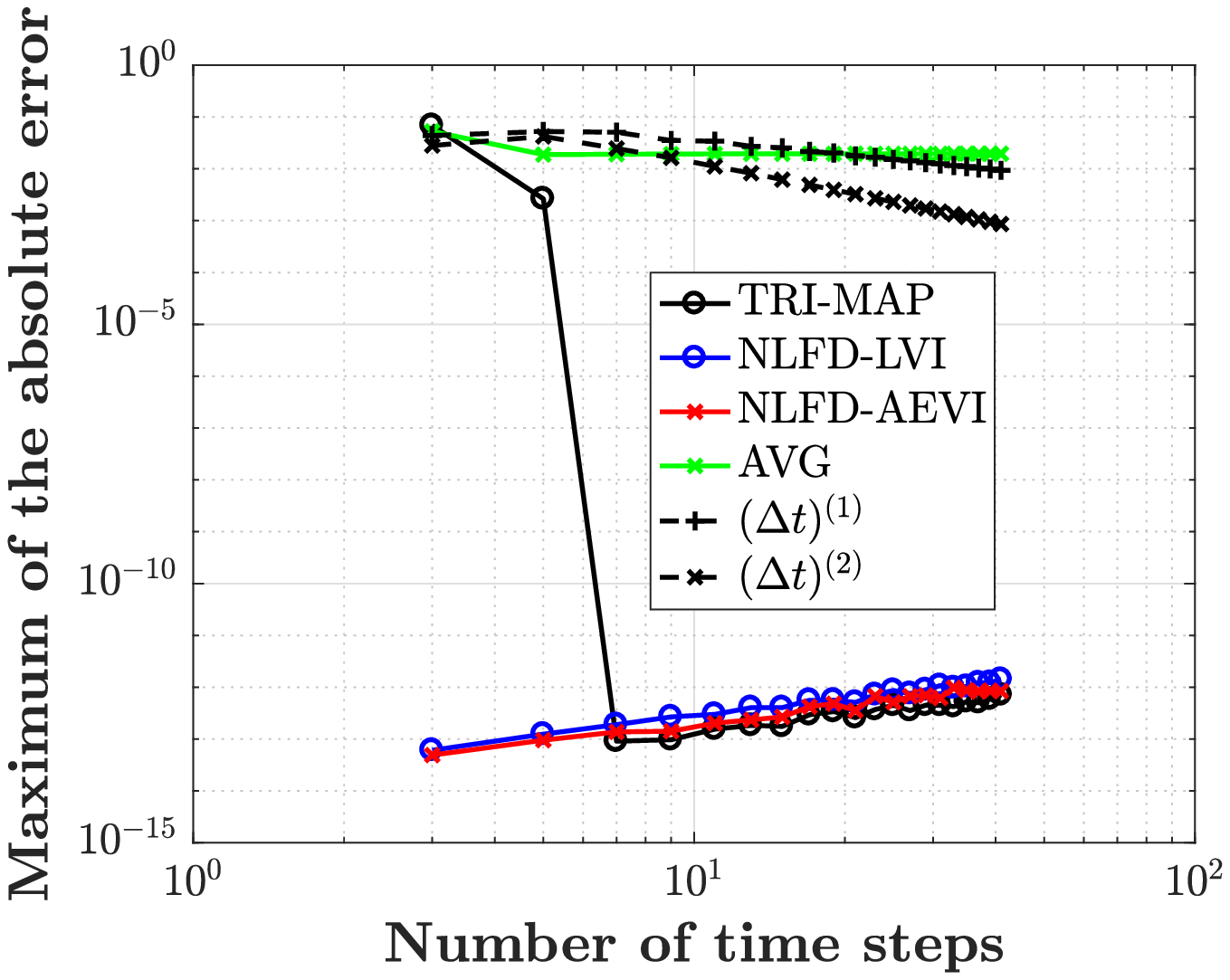}
		\caption{}
	\end{subfigure}
	\begin{subfigure}[b]{0.48\textwidth}
		\centering
		\includegraphics[width=7.5cm]{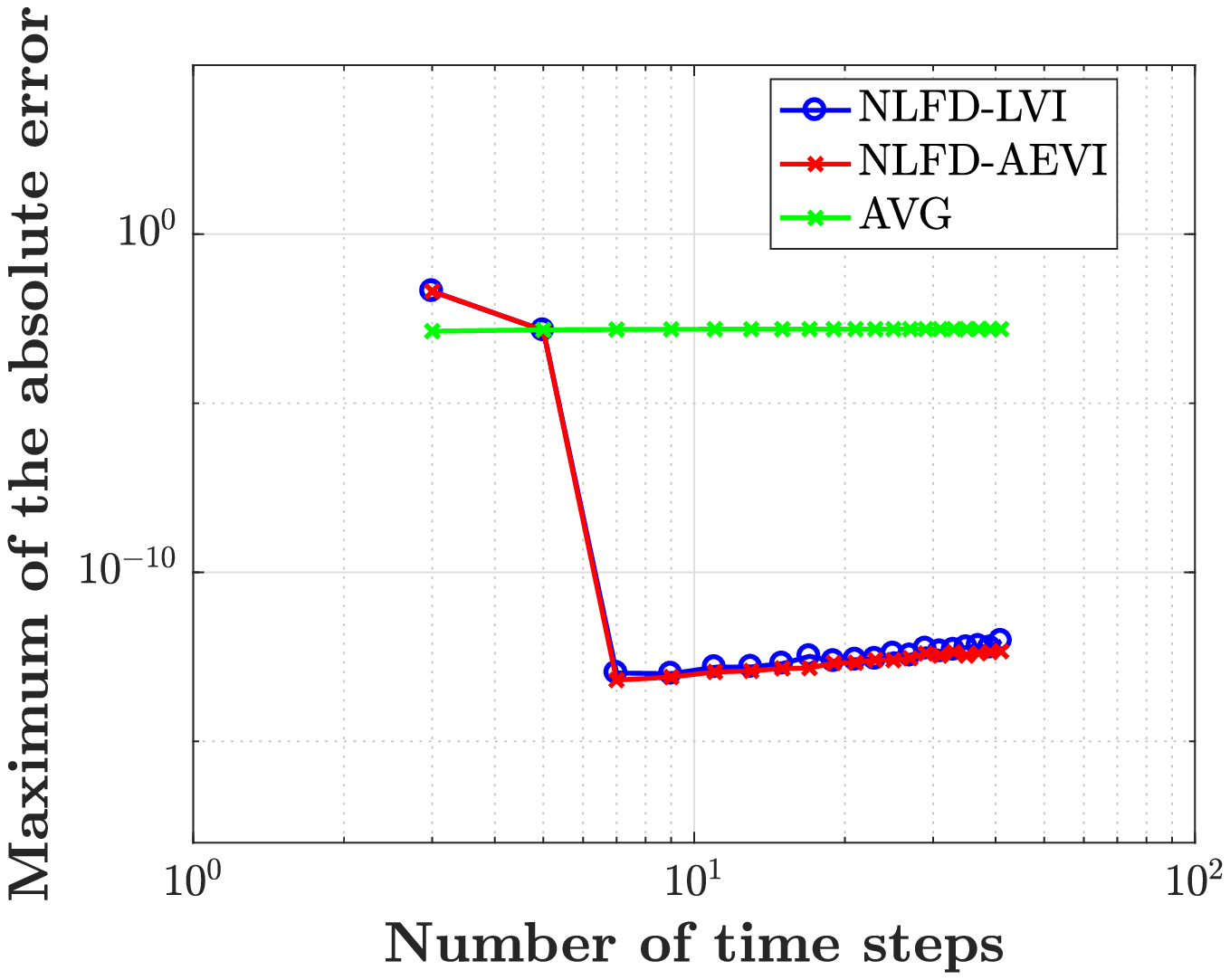}
		\caption{}
	\end{subfigure}
	
	\begin{subfigure}[b]{0.48\textwidth}
		\centering
		\includegraphics[width=7.5cm]{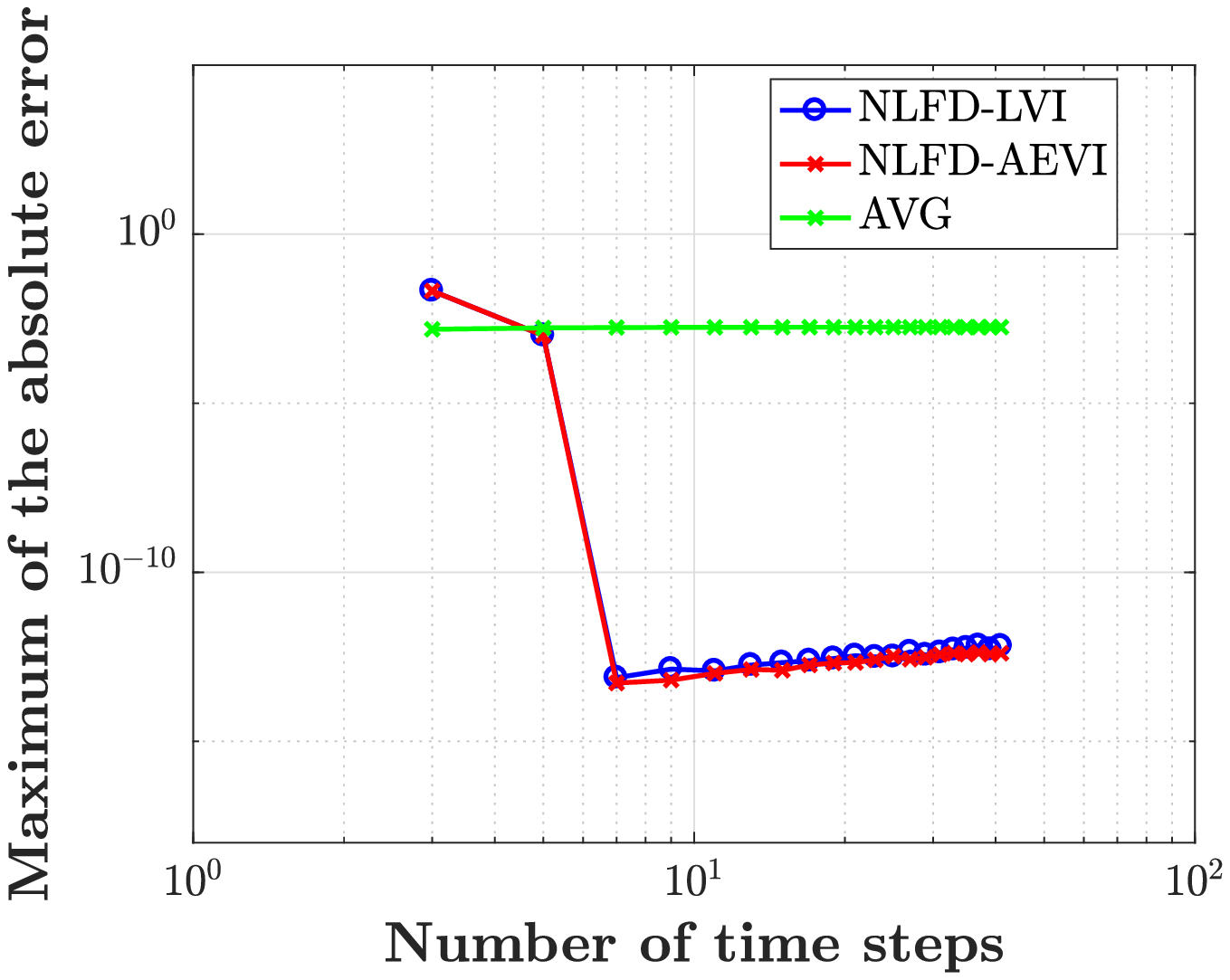}
		\caption{}
	\end{subfigure}
	\begin{subfigure}[b]{0.48\textwidth}
		\centering
		\includegraphics[width=7.5cm]{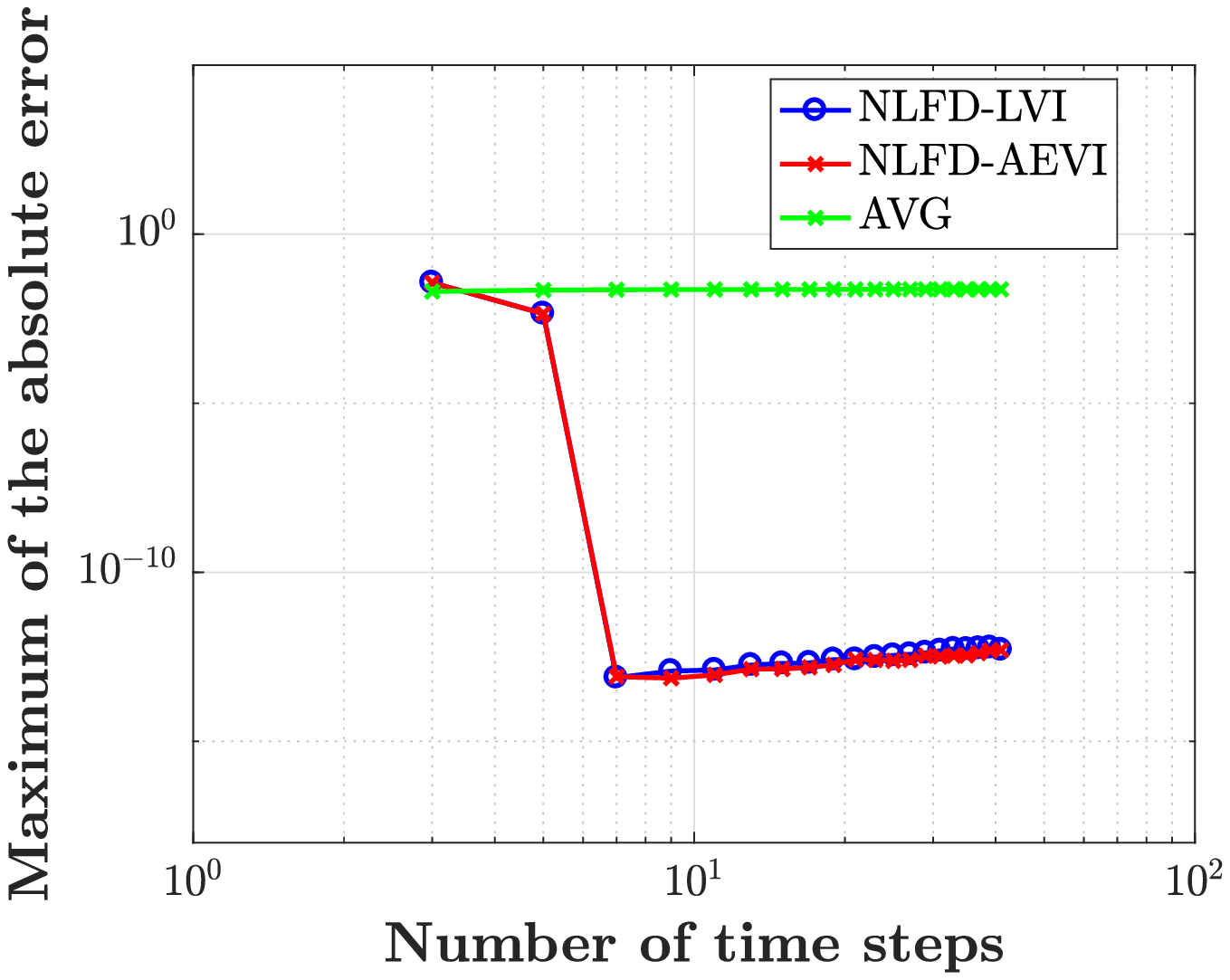}
		\caption{}
	\end{subfigure}		
\caption{Case 4 : (a) Comparison of the sum of the integrated face mesh velocities to the NLFD time derivative of the volume (b) Comparison of the individual integrated face mesh velocity to the values (TRI-MAP) in the $x$ direction (c) in the $y$ direction (d) in the $z$ direction}	
\label{fig_results_comparison_case_4}
\end{figure}

\begin{figure}[!htbp]
\centering
	\begin{subfigure}[b]{0.48\textwidth}
		\centering
		\includegraphics[width=7.5cm]{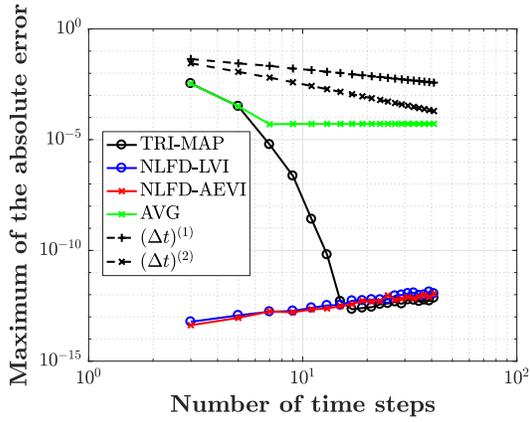}
		\caption{}
	\end{subfigure}
	\begin{subfigure}[b]{0.48\textwidth}
		\centering
		\includegraphics[width=7.5cm]{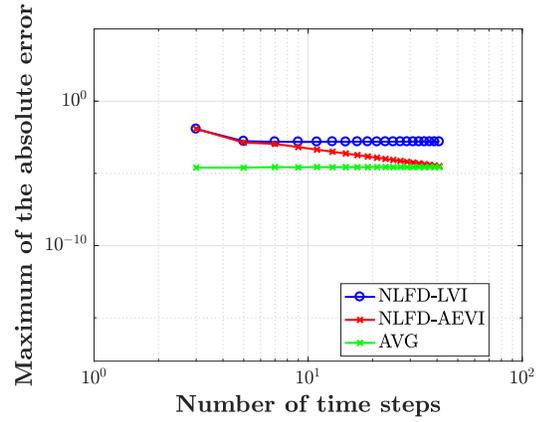}
		\caption{}
	\end{subfigure}
	
	\begin{subfigure}[b]{0.48\textwidth}
		\centering
		\includegraphics[width=7.5cm]{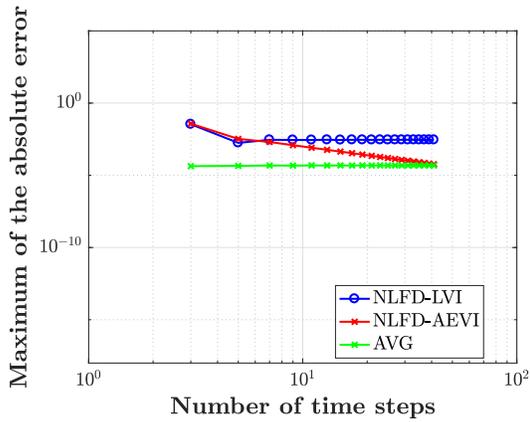}
		\caption{}
	\end{subfigure}
	\begin{subfigure}[b]{0.48\textwidth}
		\centering
		\includegraphics[width=7.5cm]{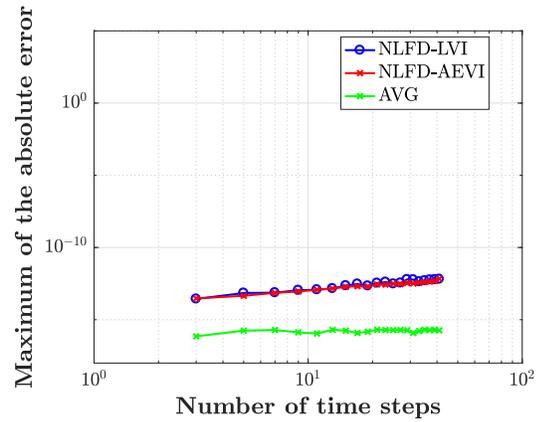}
		\caption{}
	\end{subfigure}		
\caption{Case 5 : (a) Comparison of the sum of the integrated face mesh velocities to the NLFD time derivative of the volume (b) Comparison of the individual integrated face mesh velocity to the values (TRI-MAP) in the $x$ direction (c) in the $y$ direction (d) in the $z$ direction}	
\label{fig_results_comparison_case_5}
\end{figure}

\subsection{Time-Spectral Method}
The numerical results for the Time-Spectral method are the same as that shown for NLFD-LVI and NLFD-AEVI depending on which approach is retained to compute the volumetric increments. For this reason, the graphs are not reproduced in this article. The comparisons and conclusions derived for the NLFD approach hold for the Time-Spectral method as well.

\section{Discussion and Conclusion}
\label{sec_con_disc}
The limits of the previous method of Tardiff et al.~\cite{Tardif2017} (NLFD-LVI) were clarified and demonstrated numerically and a modified approach (NLFD-AEVI) has been presented that ensures the satisfaction of the Geometric Conservation Law for a flow solver based on either the NLFD or Time-Spectral discretization of the ALE formulation of the Euler equations. The methods NLFD-AEVI and NLFD-LVI aim to satisfy the GCL by computing the integrated face mesh velocities according to the numerical discretization of the flow solver and take as input the face volumetric increments. The accuracy of the methods was shown to be highly dependent on the computation of the correct volumetric increments and in the worst cases considered converged at first-to-second-order for the NLFD-AEVI approach (corollary~\ref{th_error}) or zeroth-order for the NLFD-LVI procedure. The integrated face mesh velocities themselves may not converge to the correct value as demonstrated in our numerical test.  Although the approaches have been verified to preserve uniform flow for any number of harmonics; such a low order of accuracy defeats the purpose of spectral in time methods. Hence an alternate novel approach has been developed based on a trilinear mapping between the physical domain and the computational space which allows the evaluation of the exact cell volume and integrated face mesh velocities. The disadvantage of this method is that it is not consistent with the discretization of the flow solver, meaning that freestream preservation is not satisfied for any number of harmonics as it is with the modified approach, NLFD-AEVI. However such inconvenience is compensated by its spectral rate of convergence, which is sufficient to ensure the satisfaction of the GCL and preserve uniform flow. 

\section{Acknowledgment}
We would like to gratefully acknowledge the financial support of the Natural Sciences and Engineering Research Council of Canada and McGill University.

\section{Appendix}
\subsection{Truncation error on the volumetric increment between two successive time steps}
\label{appendix_trunc_error}

\begin{lemma}
In the context of the theorem \ref{th_new_method}, and under the definitions \ref{def_vol_hex} and \ref{def_vol_inc}, for any face $m$ the truncation error $\epsilon_{m,h}^T$ on the volumetric increment between two time steps $t_{n-1}$ and $t_n$ is of order two and can be written as,
\begin{equation}
\epsilon_{m,h}^T(t_n) = \mathcal{E}_m^T(t_{n-1})\tau^2 + \mathcal{O}(\tau^3).
\end{equation}
where $\mathcal{E}_m^T$ is a scalar periodic function and $\tau=t_n-t_{n-1}$.
\end{lemma}

\begin{proof}
We introduce the scalar triple product application $\cal{L}$ defined by,
\begin{equation}
\cal{L}=\left\lbrace
\begin{array}{lcl}
\mathbb{R}^3\times\mathbb{R}^3\times\mathbb{R}^3 & \rightarrow & \mathbb{R} \\
(v_1,v_2,v_3) & \rightarrow & v_1 \cdot (v_2 \times v_3) = \text{det}(v_1,v_2,v_3)
\end{array}
\right.
\end{equation}
due to the properties of the determinant this application is a 3-linear alternating form meaning that if any of the three vector is a linear combination of the two others the result is zero.

Recalling that the volume of any hexahedra is computed as a function of the position vectors of the vertices in the physical space $\mathbf{r}_i$ for $i=1,...,8$ through the following equation (see definition \ref{def_vol_hex}) :
\begin{equation}
\left\{
\begin{array}{lcl}
\Omega_{h} & = &(\Omega_{4321}+\Omega_{5678}+\Omega_{3487}+\Omega_{1256}+\Omega_{4158}+\Omega_{2376}), \\
\\
\mbox{with } \Omega_{ijkl} & = &\dfrac{1}{12}\mathcal{L}(\mathbf{r}_j+\mathbf{r}_k,\mathbf{r}_i+\mathbf{r}_j,\mathbf{r}_i+\mathbf{r}_l). \\
\end{array}
\right.
\label{eq_trimap_hexvol}
\end{equation}

Let $m$ be a face defined by the vertices $\mathbf{r}_i, \mbox{ }i=a,b,c,d$ (see Figure \ref{fig_non_periodic_OmegaM_move}), at the $n^{th}$ time sample a Taylor expansion gives :
\begin{equation}
\left\{
\begin{array}{llcl}
& \mathbf{r}_i(t_n) & = & \mathbf{l}_i(t_n) + \boldsymbol{\epsilon}_i(t_n), \\
\text{Linear approximation : } & \mathbf{l}_i(t_n) & = & \displaystyle \mathbf{r}_i(t_{n-1}) + \left( \frac{\partial \mathbf{r}_i}{\partial t}(t_{n-1})\right)\tau, \\
\text{Truncation error : }  & \boldsymbol{\epsilon}_i(t_n) & = & \displaystyle \frac{1}{2}\left(\frac{\partial^2\mathbf{r}_i}{\partial t^2}(t_{n-1})\right)\tau^2 + \mathcal{O}(\tau^3).
\end{array}
\right.
\end{equation}

The volumetric increment between the two time instances $t_{n-1}$ and $t_n$ is approximated by a hexahedra defined using the vertex positions with the following indexation,
\begin{equation}
\left\{
\begin{array}{lclclclclclclcl}
\mathbf{r}_1 & = & \mathbf{r}_a(t_{n-1}), & \mathbf{r}_2 & = & \mathbf{r}_b(t_{n-1}), & \mathbf{r}_3 & = & \mathbf{r}_c(t_{n-1}), & \mathbf{r}_4 & = & \mathbf{r}_d(t_{n-1}), \\
\mathbf{r}_5 & = & \mathbf{r}_a(t_{n}), & \mathbf{r}_6 & = & \mathbf{r}_b(t_{n}), & \mathbf{r}_7 & = & \mathbf{r}_c(t_{n}), & \mathbf{r}_8 & = & \mathbf{r}_d(t_{n}). 
\end{array}
\right.
\end{equation}

By substituting the Taylor expansions into the vertex positions to compute the volume of the volumetric increment through equation (\ref{eq_trimap_hexvol}), and then by exploiting the 3-linearity of the triple product application the order of the truncation error is evaluated. The lowest order terms of the truncation error are given by one of the following generic forms :
\begin{equation}
\begin{array}{l}
\mathcal{L}(\mathbf{r}_i(t_{n-1}),\mathbf{r}_j(t_{n-1}),\boldsymbol{\epsilon}_k(t_{n})) = \displaystyle \left\{ \mathbf{r}_i(t_{n-1}) \cdot \left[ \mathbf{r}_j(t_{n-1}) \times \frac{1}{2}\left(\frac{\partial^2\mathbf{r}_k}{\partial t^2}(t_{n-1})\right) \right] \right\} \tau^2 + \mathcal{O}(\tau^3),\\
\\
\mathcal{L}(\mathbf{r}_i(t_{n-1}),\boldsymbol{\epsilon}_j(t_{n}),\mathbf{r}_k(t_{n-1})) = \displaystyle \left\{ \mathbf{r}_i(t_{n-1}) \cdot \left[ \frac{1}{2}\left(\frac{\partial^2\mathbf{r}_j}{\partial t^2}(t_{n-1})\right) \times \mathbf{r}_k(t_{n-1}) \right] \right\} \tau^2 + \mathcal{O}(\tau^3),\\
\\
\mathcal{L}(\boldsymbol{\epsilon}_i(t_{n}),\mathbf{r}_j(t_{n-1}),\mathbf{r}_k(t_{n-1})) = \displaystyle \left\{ \frac{1}{2}\left(\frac{\partial^2\mathbf{r}_i}{\partial t^2}(t_{n-1})\right) \cdot \left[ \mathbf{r}_j(t_{n-1}) \times \mathbf{r}_k(t_{n-1}) \right] \right\} \tau^2 + \mathcal{O}(\tau^3).
\end{array}
\label{eq_low_order_trunc_error}
\end{equation}
where $i,j$ and $k$ are the vertices indices.

Therefore for any face $m$, the truncation error $\epsilon_{m,h}^T$ on the volumetric increment between two time steps $t_{n-1}$ and $t_n$ is of order two. In addition, it is possible to write the lowest order term of the error as a linear combination of the previous forms equation (\ref{eq_low_order_trunc_error}), thus there exists a scalar function $\mathcal{E}_m^T$ depending on the vertices paths $\mathbf{r}_i$ and their second temporal derivatives $\dfrac{\partial^2 \mathbf{r}_i}{\partial t^2}$ such that,
\begin{equation}
\epsilon_{m,h}^T(t_n) = \mathcal{E}_m^T(t_{n-1})\tau^2 + \mathcal{O}(\tau^3).
\end{equation}
Due to the temporal periodicity of the vertices paths $\mathbf{r}_i$, the function $\mathcal{E}_m^T$ is also periodic.
\end{proof}

\section*{References}

\bibliography{mybibfile}

\begin{thebibliography}{10}
\expandafter\ifx\csname url\endcsname\relax
  \def\url#1{\texttt{#1}}\fi
\expandafter\ifx\csname urlprefix\endcsname\relax\def\urlprefix{URL }\fi
\expandafter\ifx\csname href\endcsname\relax
  \def\href#1#2{#2} \def\path#1{#1}\fi

\bibitem{Hall2002}
K.~C. Hall, J.~P. Thomas, W.~S. Clark, Computation of unsteady nonlinear flows
  in cascades using a harmonic balance technique, AIAA Journal 40~(5) (2002)
  879--886.
\newblock \href {http://dx.doi.org/10.2514/2.1754} {\path{doi:10.2514/2.1754}}.

\bibitem{McMullen2002a}
M.~McMullen, A.~Jameson, J.~Alonso, Application of a non-linear frequency
  domain solver to the euler and navier-stokes equations, in: 40th AIAA
  Aerospace Sciences Meeting {\&} Exhibit, American Institute of Aeronautics
  and Astronautics, Reno, Nevada, 2002.
\newblock \href {http://dx.doi.org/10.2514/6.2002-120}
  {\path{doi:10.2514/6.2002-120}}.

\bibitem{McMullen2006a}
M.~McMullen, A.~Jameson, J.~Alonso, Demonstration of nonlinear frequency domain
  methods, AIAA Journal 44~(7) (2006) 1428--1435.
\newblock \href {http://dx.doi.org/10.2514/1.15127}
  {\path{doi:10.2514/1.15127}}.

\bibitem{McMullen2006}
M.~S. McMullen, A.~Jameson, The computational efficiency of non-linear
  frequency domain methods, Journal of Computational Physics 212~(2) (2006)
  637--661.
\newblock \href {http://dx.doi.org/10.1016/j.jcp.2005.07.021}
  {\path{doi:10.1016/j.jcp.2005.07.021}}.

\bibitem{Gopinath2005}
A.~Gopinath, A.~Jameson, Time spectral method for periodic unsteady
  computations over two- and three- dimensional bodies, in: 43rd AIAA Aerospace
  Sciences Meeting and Exhibit, American Institute of Aeronautics and
  Astronautics, Reno, Nevada, 2005.
\newblock \href {http://dx.doi.org/10.2514/6.2005-1220}
  {\path{doi:10.2514/6.2005-1220}}.

\bibitem{Gopinath2006}
A.~Gopinath, A.~Jameson, Application of the time spectral method to periodic
  unsteady vortex shedding, in: 44th AIAA Aerospace Sciences Meeting and
  Exhibit, American Institute of Aeronautics and Astronautics, Reno, Nevada,
  2006.
\newblock \href {http://dx.doi.org/10.2514/6.2006-449}
  {\path{doi:10.2514/6.2006-449}}.

\bibitem{Kachra2006}
F.~Kachra, S.~Nadarajah, Aeroelastic solutions using the time accurate and
  non-linear frequency domain methods, in: 44th AIAA Aerospace Sciences Meeting
  and Exhibit, American Institute of Aeronautics and Astronautics, Reno,
  Nevada, 2006.
\newblock \href {http://dx.doi.org/10.2514/6.2006-445}
  {\path{doi:10.2514/6.2006-445}}.

\bibitem{Tardif2017}
P.-O. Tardif, S.~Nadarajah, Three-dimensional aeroelastic solutions via the
  nonlinear frequency-domain method, AIAA Journal 55~(10) (2017) 3553--3569.
\newblock \href {http://dx.doi.org/10.2514/1.J054849}
  {\path{doi:10.2514/1.J054849}}.

\bibitem{Thomas1979}
P.~D. Thomas, C.~K. Lombard, {Geometric Conservation Law and Its Application to
  Flow Computations on Moving Grids}, AIAA Journal 17~(10) (1979) 1030--1037.
\newblock \href {http://dx.doi.org/10.2514/3.61273}
  {\path{doi:10.2514/3.61273}}.

\bibitem{Lesoinne1996}
M.~Lesoinne, C.~Farhat, {Geometric conservation laws for flow problems with
  moving boundaries and deformable meshes, and their impact on aeroelastic
  computations}, Computer Methods in Applied Mechanics and Engineering
  134~(1-2) (1996) 71--90.
\newblock \href {http://dx.doi.org/10.1016/0045-7825(96)01028-6}
  {\path{doi:10.1016/0045-7825(96)01028-6}}.

\bibitem{Herve2000}
H.~Guillard, C.~Farhat, {On the significance of the geometric conservation law
  for flow computations on moving meshes}, Computer Methods in Applied
  Mechanics and Engineering 190~(11-12) (2000) 1467--1482.
\newblock \href {http://dx.doi.org/10.1016/S0045-7825(00)00173-0}
  {\path{doi:10.1016/S0045-7825(00)00173-0}}.

\bibitem{Farhat2001}
C.~Farhat, P.~Geuzaine, C.~Grandmont, {The Discrete Geometric Conservation Law
  and the Nonlinear Stability of ALE Schemes for the Solution of Flow Problems
  on Moving Grids}, Journal of Computational Physics 174~(2) (2001) 669--694.
\newblock \href {http://dx.doi.org/10.1006/JCPH.2001.6932}
  {\path{doi:10.1006/JCPH.2001.6932}}.

\bibitem{Mavripilis2006}
D.~J. Mavriplis, Z.~Yang, {Construction of the discrete geometric conservation
  law for high-order time-accurate simulations on dynamic meshes}, Journal of
  Computational Physics 213~(2) (2006) 557--573.
\newblock \href {http://dx.doi.org/10.1016/j.jcp.2005.08.018}
  {\path{doi:10.1016/j.jcp.2005.08.018}}.

\bibitem{Blazek2005}
J.~Blajek, {Computational Fluid Dynamics: Principles and Applications},
  Elsevier, 1988.

\bibitem{JAMESON1981}
A.~Jameson, W.~Schmidt, E.~Turkel, {Numerical solution of the Euler equations
  by finite volume methods using Runge Kutta time stepping schemes}, in: 14th
  Fluid and Plasma Dynamics Conference, American Institute of Aeronautics and
  Astronautics, Reno, Nevada, 1981.
\newblock \href {http://dx.doi.org/10.2514/6.1981-1259}
  {\path{doi:10.2514/6.1981-1259}}.

\bibitem{Jameson1995}
A.~Jameson, {Analysis and design of numerical schemes for gas dynamics, 1 :
  artificial diffusion, upwind biasing, and their effect on accuracy and
  multigrid convergence}, International Journal of Computational Fluid Dynamics
  4~(3) (1995) 171--218.
\newblock \href {http://dx.doi.org/10.1080/10618569508904524}
  {\path{doi:10.1080/10618569508904524}}.

\bibitem{Wendland1995}
H.~Wendland, {Piecewise polynomial, positive definite and compactly supported
  radial functions of minimal degree}, Advances in Computational Mathematics
  4~(1) (1995) 389--396.
\newblock \href {http://dx.doi.org/10.1007/BF02123482}
  {\path{doi:10.1007/BF02123482}}.

\bibitem{Zhang1993}
H.~Zhang, M.~Reggio, J.~Y. Tr{\'{e}}panier, R.~Camarero, {Discrete form of the
  GCL for moving meshes and its implementation in CFD schemes}, Computers and
  Fluids 22~(1) (1993) 9--23.
\newblock \href {http://dx.doi.org/10.1016/0045-7930(93)90003-R}
  {\path{doi:10.1016/0045-7930(93)90003-R}}.

\bibitem{Mavripilis2011}
D.~J. Mavriplis, C.~R. Nastase, {On the geometric conservation law for
  high-order discontinuous Galerkin discretizations on dynamically deforming
  meshes}, Journal of Computational Physics 230~(11) (2011) 4285--4300.
\newblock \href {http://dx.doi.org/10.1016/j.jcp.2011.01.022}
  {\path{doi:10.1016/j.jcp.2011.01.022}}.

\bibitem{Dukowicz1988}
J.~K. Dukowicz, {Efficient volume computation for three-dimensional hexahedral
  cells}, Journal of Computational Physics 74~(2) (1988) 493--496.
\newblock \href {http://dx.doi.org/10.1016/0021-9991(88)90091-5}
  {\path{doi:10.1016/0021-9991(88)90091-5}}.

\bibitem{Lopez2017}
J.~I. L{\'{o}}pez, M.~Brovka, J.~M. Escobar, R.~Montenegro, G.~V. Socorro,
  {Strategies for optimization of hexahedral meshes and their comparative
  study}, Engineering with Computers 33~(1) (2017) 33--43.
\newblock \href {http://dx.doi.org/10.1007/s00366-016-0454-1}
  {\path{doi:10.1007/s00366-016-0454-1}}.

\bibitem{Knabner2003}
P.~Knabner, S.~Korotov, G.~Summ, {Conditions for the invertibility of the
  isoparametric mapping for hexahedral finite elements}, Finite Elements in
  Analysis and Design 40~(2) (2003) 159--172.
\newblock \href {http://dx.doi.org/10.1016/S0168-874X(02)00196-8}
  {\path{doi:10.1016/S0168-874X(02)00196-8}}.

\bibitem{Zwanenburg2016}
P.~Zwanenburg, S.~Nadarajah, {Equivalence between the Energy Stable Flux
  Reconstruction and Filtered Discontinuous Galerkin Schemes}, Journal of
  Computational Physics 306 (2016) 343--369.
\newblock \href {http://dx.doi.org/10.1016/j.jcp.2015.11.036}
  {\path{doi:10.1016/j.jcp.2015.11.036}}.

\end{thebibliography}

\end{document}